\documentclass[aps,pra,onecolumn,superscriptaddress,longbibliography,floatfix,amsmath,amssymb]{revtex4-2}

\usepackage{graphicx}
\usepackage{xcolor}
\usepackage{braket}
\usepackage{mathtools}
\usepackage{amsmath,amssymb,amsthm}
\usepackage{bm}
\usepackage{physics}
\graphicspath{ {./Figures/} }

\newtheorem{defi}{Definition}
\newtheorem{prop}{Proposition}

\newtheorem{lem}{Lemma}

\newtheorem{form}{Formulae}

\newcommand{\beprop}{\begin{prop}}
\newcommand{\enprop}{\end{prop}}
\newcommand{\befor}{\begin{form}}
\newcommand{\enfor}{\end{form}}
\newcommand{\bprf}{\begin{proof}}
\newcommand{\eprf}{\end{proof}}

\usepackage{bbm}
\newcommand{\ii}{\mathsf{i}}

\def\R{\mathbb{R}}

\def\C{\mathbb{C}}
\def\Z {\mathbb{Z}}
\def\S {\mathbb{S}}

\def\lg{\langle }
\def\rg{\rangle }

\def\vap{\varpi}

\def\bn{\mathbf{n}}

\def\ud{\mathrm{d}}

\def\sfM{\mathsf{M}}

\def\sfMv{\hat{\mathsf{M}}^{\vap}}

\def\bu{\mathbbm{1}}

\def\vectheta{\bm{\theta}}
\def\T{\mathbb{T}}

\def\vecm{\bm{m}}
\def\vecn{\bm{n}}
\def\vecgkernel{\pmb{\gamma}-\pmb{\gamma}^\prime}
\def\vecalpha{\bm{\alpha}}
\def\bg{\boldsymbol{\gamma}}

\def\ellboundone{\frac{\ell_1}{2}}
\def\ellboundtwo{\frac{\ell_2}{2}}

\DeclareFontFamily{U}{mathx}{}
\DeclareFontShape{U}{mathx}{m}{n}{<-> mathx10}{}
\DeclareSymbolFont{mathx}{U}{mathx}{m}{n}
\DeclareMathAccent{\widehat}{0}{mathx}{"70}
\DeclareMathAccent{\widecheck}{0}{mathx}{"71}

\usepackage[colorlinks=true,linkcolor=blue,citecolor=blue,urlcolor=blue]{hyperref}

\begin{document}

\title{Stellar Representation of Gottesman--Kitaev--Preskill States on the Double Discrete Cylinder}

\author{Aidan Zlotak}
\affiliation{Department of Physics, Worcester Polytechnic Institute, Worcester, Massachusetts 01609, USA}

\author{Nicolas Fabre}
\email{nicolas.fabre@telecom-paris.fr}
\affiliation{Telecom Paris-LTCI, Institut Polytechnique de Paris, 19 Place Marguerite Perey, 91120 Palaiseau, France}

\author{Romain Murenzi}
\affiliation{Department of Physics, Worcester Polytechnic Institute, Worcester, Massachusetts 01609, USA}

\date{\today}

\begin{abstract}
Covariant integral quantization is developed on the double discrete cylinder phase space, with the torus as configuration space, and yields the Wigner distribution and the stellar representation. This phase space is particularly well suited to translationally invariant states, such as the Gottesman--Kitaev--Preskill (GKP) states. These states are central to fault-tolerant quantum computing architectures, yet their quantum resources are still under investigation. The stellar representations of ideal GKP logical states, which behave as phase eigenstates of the double discrete cylinder, have no zeros on this phase space, whereas those of their superpositions have infinitely many. The vacuum, although Gaussian, has zeros on this phase space, and the zeros of finitely squeezed GKP states move with the widths of the peaks and of the envelope, and reach this phase space for fine-tuned widths. These features echo the role of the vacuum as the resource that makes GKP-based architectures hard to simulate.
\end{abstract}

\maketitle

%%%%%%%%%%%%%%%%%%%%%%%%%%%%%%%%%%%%%%%%%%
\section{Introduction}
Continuous-variable (CV) systems are the basis of many protocols in quantum communication, such as the GG02 quantum key distribution protocol~\cite{grosshans_continuous_2002,grosshans_quantum_2003,jouguet_experimental_2013,pirandola_advances_2020}, and in quantum sensing~\cite{lawrie_quantum_2019}, where squeezed states of light already enhance the sensitivity of gravitational-wave detectors~\cite{PhysRevLett.123.231108,PhysRevLett.123.231107}. In quantum computation, CV systems have enabled the generation of large-scale cluster states~\cite{yokoyama_ultra-large-scale_2013,asavanant_generation_2019,larsen_deterministic_2019}, demonstrations of quantum computational advantage with Gaussian boson sampling~\cite{zhong_quantum_2020,madsen_quantum_2022}, and modular photonic quantum computers~\cite{aghaeerad_scaling_2025,larsen_integrated_2025}.

Phase-space representations, such as the Wigner and Husimi distributions, are a central tool for the analysis of CV states. The geometry of the phase space follows from the symmetry group of the system, and phase-space distributions have been constructed on the plane, the sphere~\cite{stratonovich1957distributions,arecchi_atomic_1972,agarwal_relation_1981,varilly_moyal_1989,dowling_wigner_1994,Klimov:00,yang_wigner_2019,davis_wigner_2021,kowalski_wigner_2023}, the discrete cylinder~\cite{mukunda_wigner_1979,de_bievre_coherent_1989,PhysRevA.49.3255,gonzalez_quantum_2003,rigas_orbital_2011,gazeau_integral_2022}, the hyperboloid~\cite{seyfarth_wigner_2020,fabre_local_2023} and the torus~\cite{davis_identifying_2024}, by the Brif--Mann construction~\cite{brif_general_1998} or by covariant integral quantization~\cite{gazeau_coherent_2009,ali_coherent_2014,bergeron_are_2013,gazeau_integral_2014,gazeau_covariant_2016,gazeau_classical_2018,gazeau_integral_2022,fabre_time-frequency_2022}. Phase-space distributions also reveal non-classical features. The negativity of the Wigner distribution is a necessary resource for a quantum computational advantage~\cite{hudson_when_1974,mari_positive_2012,veitch_resource_2014}, and is equivalent to contextuality for CV measurements~\cite{haferkamp_equivalence_2021,PhysRevLett.129.230401}. The zeros of the Husimi distribution, whose number defines the stellar rank~\cite{chabaud_stellar_2020,chabaud_classical_2021}, quantify the non-Gaussian character of a state in the plane~\cite{PhysRevLett.130.090602}, and the zeros of the Majorana constellation characterize spin states on the sphere~\cite{majorana_atomi_1932,bjork_extremal_2015,bouchard_quantum_2017,davis_stellar_2023}, as well as states on the discrete cylinder~\cite{fabre_majorana_2023}. These non-classical indicators depend on the geometry of the phase space, which must therefore be adapted to the symmetries of the state. Among all phase spaces, the discrete cylinder offers a natural setting for twisted photons~\cite{allen_orbital_1992,rigas_non-negative_2010,rigas_orbital_2011,fabre_majorana_2023}, the charge and the phase of a Josephson junction~\cite{koch_charge-insensitive_2007}, and the integer and angular variables of rotor codes~\cite{vuillot_homological_2024}. It also arises for continuous variables, that is, for any pair of canonically conjugate variables obeying the Heisenberg algebra, such as the quadratures of the electromagnetic field~\cite{braunstein_quantum_2005}, the transverse position and momentum of a photon~\cite{tasca_continuous_2011}, or the time and frequency of a single photon~\cite{fabre_time-frequency_2022}. The decomposition of the two variables modulo a period and its reciprocal period defines modular variables, which are bounded and commuting operators~\cite{aharonov_modular_1969,zak_kq-representation_1972,busch_what_1986,englert_periodic_2006,ketterer_quantum_2016,pantaleoni_zak_2023}. Quantum correlations in modular variables have been observed with photon pairs~\cite{carvalho_experimental_2012,barros_free-space_2017}, and sequential measurements of the modular position and momentum have been performed on the motion of a trapped ion~\cite{fluhmann_sequential_2018}.

Modular variables are also the natural observables of the Gottesman--Kitaev--Preskill (GKP) code, which encodes a logical qubit into a lattice of the phase space of an oscillator~\cite{gottesman_encoding_2001,grimsmo_quantum_2021,terhal_towards_2020,conrad_gottesman-kitaev-preskill_2022}, and of its time-frequency counterpart for single photons~\cite{fabre_generation_2020}. The phase space associated with the two modular variables is a double discrete cylinder, the product of two discrete cylinders, on which a Wigner distribution has been constructed~\cite{fabre_wigner_2020}. GKP states have been prepared with trapped ions~\cite{fluhmann_encoding_2019,deneeve_error_2022}, superconducting cavities~\cite{campagne-ibarcq_quantum_2020,kudra_robust_2022,sivak_real-time_2023}, quadratures of the electromagnetic field~\cite{konno_logical_2024,larsen_integrated_2025} and the time-frequency degrees of freedom of single photons~\cite{fabre_generation_2020}, and have been proposed with free electrons~\cite{baranes_free-electron_2023}, while several preparation schemes have been designed for optical and microwave fields~\cite{vasconcelos_all-optical_2010,terhal_encoding_2016,weigand_generating_2018}. GKP states are the building blocks of fault-tolerant architectures~\cite{fukui_high-threshold_2018,baragiola_all-gaussian_2019,noh_fault-tolerant_2020,tzitrin_progress_2020,bourassa_blueprint_2021,larsen_integrated_2025}. The origin of the quantum resources of GKP states remains subtle. Ideal GKP states are non-Gaussian, yet circuits made of ideal GKP states, Gaussian operations and Gaussian measurements can be efficiently simulated~\cite{calcluth_efficient_2022}, while the addition of the vacuum makes these circuits hard to simulate~\cite{calcluth_vacuum_2023}. In the same way, the ideal logical states have a negative Wigner distribution in the plane but a positive one on the double cylinder~\cite{fabre_wigner_2020}. On the torus phase space of the Zak--Gross Wigner distribution, defined for GKP qudits of odd dimension, ideal stabilizer GKP states have a positive Wigner distribution, whereas magic GKP states, finitely squeezed GKP states and the vacuum have a negative one~\cite{davis_identifying_2024,calcluth_classical_2025}. This negativity accounts for the role of the vacuum, whose Wigner distribution is positive in the plane, as a resource that makes GKP circuits hard to simulate~\cite{davis_identifying_2024}. A phase space adapted to the translational symmetry of GKP states is thus a natural tool to locate the resources of these states.

In this work, covariant integral quantization~\cite{gazeau_integral_2014,gazeau_integral_2022} is developed on the double discrete cylinder, the product of a two-dimensional grid and of a torus. A weight function on this phase space defines a quantization operator, as a weighted superposition of the displacement operators of the phase space, whose integral kernel follows from a partial Fourier transform of the weight. Displacing this operator over the phase space associates an operator with every classical function, and a quasi-probability distribution, or semi-classical portrait, with every quantum state. A constant weight makes the quantization operator proportional to the parity operator, and the associated distribution is the Wigner distribution. A weight given by the overlap between a fiducial state and its displaced copies makes the quantization operator the projector onto this fiducial state; the displaced copies are coherent states, and the associated distribution is the Husimi distribution, the squared overlap between the state and these coherent states. Modular variables, obtained by decomposing two canonically conjugate continuous variables modulo a period and its reciprocal period, make the torus the configuration space of these variables; through the Zak transform, every state is represented by a function on the torus, expanded on an integer basis whose dependence on the choice of the unit cell is inherited by the phase-space distributions. The wrapped coherent states of a particle on a circle~\cite{Kowalski_1996} appear as the coherent states generated by a periodized Gaussian fiducial vector, and define a stellar function and a Husimi distribution on the double discrete cylinder. The relevant zeros are those lying on the lattice of heights of the two cylinders, whose number defines a stellar rank on this phase space.

Position eigenstates and ideal GKP logical states have no zero on this lattice, consistently with their positive Wigner distribution on the double cylinder, whereas the superpositions of ideal logical states have infinitely many zeros on the lattice, consistently with the negativity of their Wigner distribution. The vacuum, whose Husimi distribution has no zero in the plane, has zeros on the lattice, at positions that depend on the choice of the unit cell. The zeros of the ideal GKP states also move when the peaks and the envelope acquire a finite width. When the envelope contains many peaks, the modular wavefunction factorizes, and the zeros of the ideal logical states, which form whole circles of the phase space at half-integer heights, move in height with the widths, as located analytically; these circles reach the lattice for fine-tuned widths, and the Husimi distribution extends over a number of heights set by the widths. The lattice zeros of the ideal superpositions drift away from the lattice, except for four zeros protected by parity. The appearance of zeros on the lattice once the vacuum and finite squeezing enter the description echoes the role of the vacuum as the resource that makes GKP-based architectures hard to simulate~\cite{calcluth_vacuum_2023}.

This paper is organized as follows. Section~\ref{Torus_intro} introduces the torus as a configuration space, with its Hilbert space and its translation, modulation and parity operators, and develops the covariant integral quantization on the double discrete cylinder, which yields the semi-classical portraits and the Wigner and Husimi distributions. Section~\ref{sec:modular} introduces the modular variables, the Zak transform and the integer basis, and constructs the wrapped coherent states, the stellar function and the Husimi distribution on the double discrete cylinder, together with a stellar rank. Section~\ref{sec:examples} applies this construction to examples, and locates the zeros of the stellar function for position eigenstates, for the vacuum, and for ideal and physical GKP states and their superpositions. Section~\ref{sec:conclusion} summarizes the results and discusses perspectives.

%%%%%%%%%%%%%%%%%%%%%%%%%%%%%%%%%%%%%%%%%%
\section{The double discrete cylinder and covariant integral quantization}
\label{Torus_intro}

This section introduces the torus as a configuration space and the double discrete cylinder as the associated phase space, together with the translation, modulation and parity operators acting on functions on the torus. Covariant integral quantization is then developed on this phase space. It associates an operator with every classical function and a semi-classical portrait with every operator, and yields the Wigner and Husimi distributions for two particular choices of the weight function. This geometry describes twisted two-photon states and is particularly relevant for GKP states: the modular variables introduced in Sec.~\ref{sec:modular} show how the torus, or equivalently a two-dimensional discrete lattice, appears as the configuration space of continuous variables.

\subsection{The torus and its Hilbert space}
\label{non-uniform 2-torus}

The torus $\T^2_{\pmb \ell}=\S_{r_1}\times\S_{r_2}$ is the product of the two circles
\begin{equation}
\S_{r_j}=\left\{\gamma_j\in\left[-\frac{\ell_j}{2},\frac{\ell_j}{2}\right)\,\mathrm{mod}\ \ell_j\right\},\qquad j=1,2,
\end{equation}
of circumferences $\ell_1$ and $\ell_2$, that is, of radii
\begin{equation}
r_1=\frac{\ell_1}{2\pi},\qquad r_2=\frac{\ell_2}{2\pi}.
\end{equation}
The two radii are tuned independently through the periods $\ell_1$ and $\ell_2$, so that the torus is in general asymmetric. The periods are gathered into the vector
\begin{equation}
\label{ellvec}
\bm\ell=\binom{\ell_1}{\ell_2}.
\end{equation}
A state is described by a function $\psi(\pmb{\gamma})=\left<\pmb{\gamma}|\psi\right>$ of the Hilbert space $\mathcal{H}=L^2(\T^2_{\pmb \ell},\ud^2\pmb{\gamma})$ of square-integrable complex-valued functions on the torus, with the norm
\begin{equation}
\label{ortho_cond}
\Vert \psi\Vert^2=\int_{\T^2_{\pmb \ell}} \ud^2\pmb{\gamma} \ \overline{\psi(\pmb{\gamma})} \psi(\pmb{\gamma})= \int_{\T^2_{\pmb \ell}} \ud^2\pmb{\gamma} \ |\psi(\pmb{\gamma})|^2< \infty\, ,
\end{equation}
where $\overline{\psi(\pmb{\gamma})}$ denotes the complex conjugate of $\psi(\pmb{\gamma})$. Every function on the torus is expanded in a double Fourier series,
\begin{equation}
\begin{gathered}
\psi(\pmb{\gamma}) = \sum_{m_1,m_2\in\Z} c_{m_1m_2}\, e^{\ii m_1 \frac{2\pi}{\ell_1}\gamma_1} e^{\ii m_2 \frac{2\pi}{\ell_2}\gamma_2},\\
c_{m_1m_2} = \frac{1}{\ell_1\ell_2} \int_{-\ell_1/2}^{\ell_1/2}\ud\gamma_1 \int_{-\ell_2/2}^{\ell_2/2}\ud\gamma_2 \ \psi(\pmb{\gamma})\, e^{-\ii m_1 \frac{2\pi}{\ell_1}\gamma_1} e^{-\ii m_2 \frac{2\pi}{\ell_2}\gamma_2},
\end{gathered}
\end{equation}
which defines the unitary Fourier transform and its inverse,
\begin{align}
\label{Fourier_transform}
F[\psi](m_1,m_2) &= \frac{1}{\sqrt{\ell_1\ell_2}}\int_{-\ellboundone}^{\ellboundone}\ud\gamma_1\int_{-\ellboundtwo}^{\ellboundtwo}\ud\gamma_2 \ \psi(\pmb{\gamma})\,e^{-\ii(m_1\gamma_1\frac{2\pi}{\ell_1}+m_2\gamma_2\frac{2\pi}{\ell_2})},\\
\label{inv_Fourier_transform}
\psi(\pmb{\gamma}) &= \frac{1}{\sqrt{\ell_1\ell_2}}\sum_{m_1,m_2\in\Z}F[\psi](m_1,m_2)\, e^{\ii(m_1\gamma_1\frac{2\pi}{\ell_1}+m_2\gamma_2\frac{2\pi}{\ell_2})}.
\end{align}
The conjugate variables
\begin{equation}
\vecm_{\pmb \ell} = \begin{pmatrix}
    \frac{2\pi}{\ell_1} m_1 \\ \frac{2\pi}{\ell_2}m_2
\end{pmatrix},\qquad (m_1,m_2)\in\Z^2,
\end{equation}
form the discrete grid $\Z^2_{\pmb\ell} = \frac{2\pi}{\ell_1}\Z \times\frac{2\pi}{\ell_2}\Z$, which is the dual of the torus, and the phase space is the product $\Gamma=\Z^2_{\pmb\ell}\times\T^2_{\pmb\ell}$, called the double discrete cylinder. A Hilbertian basis of $\mathcal{H}$ is the countable set of normalized Fourier exponentials
\begin{equation}
\chi_{\vecm_{\pmb \ell}}(\pmb{\gamma}) = \frac{e^{\ii \vecm_{\pmb \ell}\cdot\pmb{\gamma}}}{\sqrt{\ell_1\ell_2}},\qquad \vecm_{\pmb \ell}\in\Z^2_{\pmb\ell},
\end{equation}
which satisfies the orthogonality and completeness conditions
\begin{equation}
\label{ortcomp}
\lg \chi_{\bn}| \chi_{\bn^{\prime}}\rg=\delta_{\bn \bn^{\prime}}\, , \quad  \sum_{\bn\in \Z^2_{\pmb\ell}} | \chi_{\bn}\rg\lg \chi_{\bn}| = \bu \, ,
\end{equation}
where $\bu$ is the identity operator on $\mathcal{H}$. In the Dirac notation, the identity is resolved as well by the kets $|\pmb\gamma\rg$ and $|\bn\rg$,
\begin{equation}
\label{resid2}
\int_{\T^2_{\pmb \ell}}\ud^2\pmb\gamma\,|\pmb\gamma\rg\lg\pmb\gamma| = \bu=  \sum_{\bn\in \Z^2_{\pmb\ell}} |\bn\rg\lg \bn|\, ,
\end{equation}
and the Parseval identity reads $\left<\psi|\psi\right>=\sum_{m_1,m_2\in\Z}|F[\psi](m_1,m_2)|^2$.

\subsection{Translation, modulation and parity operators}
\label{sec:operators}

Three elementary operators act on $\mathcal{H}$. The parity operator $\hat{\mathrm{P}}$, the translation operator $\hat T_{\vectheta}$, where $\vectheta$ is a vector of angular positions on $\T^2_{\pmb \ell}$, and the modulation operator $\hat E_{\vecm_{\pmb \ell}}$, with $\vecm_{\pmb \ell}\in \Z^2_{\pmb\ell}$, are defined by
\begin{align}
\label{Par_oper}
(\hat{\mathrm{P}} \psi)(\pmb{\gamma}) &= \psi(-\pmb{\gamma}),\\
\label{T_op}
(\hat T_{\vectheta} \psi)(\pmb{\gamma}) &= \psi(\pmb{\gamma}- \vectheta),\\
\label{E_operator}
(\hat E_{\vecm_{\pmb \ell}}\psi)(\pmb{\gamma}) &= e^{\ii \vecm_{\pmb \ell} \cdot \pmb{\gamma}} \psi(\pmb{\gamma}).
\end{align}
Combining the translation and modulation operators of Eqs.~\eqref{T_op} and~\eqref{E_operator} gives the displacement operator of the double discrete cylinder,
\begin{equation}
\label{eq:U_def}
\left(\hat U(\vecm_{\pmb \ell},\vectheta)\psi\right)(\pmb{\gamma}) = \psi_{\vecm_{\pmb \ell},\vectheta}(\pmb{\gamma}) = e^{-\ii\frac{\vecm_{\pmb \ell}\cdot\vectheta}{2}}e^{\ii\vecm_{\pmb \ell}\cdot\pmb{\gamma}}\psi(\pmb{\gamma}-\vectheta),
\end{equation}
which plays the role of the Weyl displacement operator of the plane. The operators $\hat U(\vecm_{\pmb \ell},\vectheta)$ form a representation of the Weyl--Heisenberg group of the torus, whose structure and composition rules are given in Appendix~\ref{weyl}.

\subsection{Quantization operator}
\label{section4}

Following Refs.~\cite{gazeau_integral_2014,gazeau_integral_2022,fabre_majorana_2023}, a function $\vap$ on the phase space $\Gamma$, called the weight function, defines the quantization operator
\begin{equation}
\label{affbop}
\sfMv = \frac{1}{\ell_1\ell_2}\sum_{\vecm_{\pmb \ell} \in \Z^2_{\pmb\ell}}\int_{\T^2_{\pmb \ell}}\ud^2\bm \theta \ \vap(\vecm_{\pmb \ell},\bm \theta)\hat U(\vecm_{\pmb \ell},\bm \theta)\, ,
\end{equation}
a weighted superposition of the displacement operators of Eq.~\eqref{eq:U_def}. The weight function is chosen such that $\sfMv$ is bounded and symmetric, that is, self-adjoint on the Hilbert space of physical states.
\begin{prop}
    With the chosen weight $\vap(\vecm_{\pmb \ell},\bm \theta)$, the operator $\sfMv$ has the following properties:
    \begin{enumerate}
        \item[(i)] It is the integral operator
        \begin{equation*}
            (\sfMv \phi)(\pmb \gamma) = \int_{\T^2_{\pmb \ell}}\ud^2 \pmb \gamma^{\prime}\, \mathcal{M}^{\vap}(\pmb \gamma,\pmb \gamma^{\prime})\phi(\pmb \gamma^{\prime}),
        \end{equation*}
        with the integral kernel
        \begin{align*}
            \mathcal{M}^{\vap}(\pmb \gamma,\pmb \gamma^{\prime}) =
            \frac{1}{\ell_1\ell_2}\sum_{\vecm_{\pmb \ell}\in\Z^2_{\pmb\ell}}\,\varpi(\vecm_{\pmb \ell},\vecgkernel)\,e^{\ii\frac{\vecm_{\pmb \ell}\cdot(\pmb \gamma+\pmb \gamma^{\prime})}{2}}
\equiv\widetilde{\varpi}_{p_{1}}\left(\frac{\pmb\gamma+\pmb\gamma^{\prime}}{2},\vecgkernel\right),
        \end{align*}
        where $\widetilde{\vap}_{p_1}$ is the partial Fourier transform of $\vap$ with respect to its first variable $\vecm_{\pmb \ell}\in \Z^2_{\pmb\ell}$.
        \item[(ii)] It is symmetric if and only if the weight satisfies
\begin{align*}
\sfMv =\big(\sfMv\big)^{\dag} \Leftrightarrow \overline{\varpi(-\vecm_{\pmb \ell},\pmb \gamma)}=\varpi(\vecm_{\pmb \ell},-\pmb \gamma)\,.
\end{align*}
        \item[(iii)] It is trace class, with the trace
\begin{equation*}
\mathrm{Tr}(\sfMv )=\varpi(0,0)\,.
\end{equation*}
 \item[(iv)] Its symmetry implies the following condition on the kernel,
\begin{equation*}
\overline{{\mathcal M}^{\varpi}(\pmb \gamma, \pmb \gamma^{\prime})}= {\mathcal M}^{\varpi}(\pmb \gamma^{\prime},\pmb \gamma),
\end{equation*}
that is,
\begin{equation*}
\overline{\widetilde{\varpi}_{p_{1}}\left(\frac{\pmb\gamma+\pmb\gamma^\prime}{2},\vecgkernel\right)}=\widetilde{\varpi}_{p_{1}}\left(\frac{\pmb\gamma+\pmb\gamma^\prime}{2},-(\vecgkernel)\right).
\end{equation*}
\end{enumerate}
\end{prop}

\subsection{Covariant integral quantization}

The displacement operators transport the quantization operator of Eq.~\eqref{affbop} over the phase space,
\begin{equation}
\label{WeylMom}
\sfMv(\vecm_{\pmb \ell},\bm \theta )=\hat U(\vecm_{\pmb \ell} , \bm \theta )\sfMv \hat U^{\dag}( \vecm_{\pmb \ell} , \bm \theta )\,,
\end{equation}
and the family of operators of Eq.~\eqref{WeylMom} defines the quantization of classical functions on $\Gamma$.
\begin{defi}
The Weyl--Heisenberg covariant integral quantization is the linear map
\begin{equation}
\label{genqvap}
f\mapsto \mathrm{Op}^{\vap}_f= \frac{1}{\ell_1\ell_2}\sum_{\vecm_{\pmb \ell}\in\Z^2_{\pmb\ell}}\int_{\T^2_{\pmb \ell}}\ud^2\bm\theta\, f(\vecm_{\pmb \ell},\bm \theta )\, \sfMv(\vecm_{\pmb \ell},\bm \theta)\,.
\end{equation}
\end{defi}
\begin{prop}
\label{CIQ_proof}
    The operator $\mathrm{Op}_f^\varpi$ is the integral operator on $\mathcal{H}$
    \begin{align}
    \label{A_operator}
        (\mathrm{Op}_f^\varpi \phi)(\pmb \gamma) = \int_{\T^2_{\pmb \ell}} \ud^2\pmb\gamma^{\prime}\, \mathcal{A}(\pmb\gamma,\pmb\gamma^\prime)\phi(\pmb\gamma^\prime)\,,
    \end{align}
    with the integral kernel
    \begin{align}
    \label{A_kernel}
        \mathcal{A}(\pmb\gamma,\pmb\gamma^\prime) =\frac{1}{\ell_1\ell_2}\sum_{\vecm_{\pmb \ell}^\prime \in \Z^2_{\pmb\ell}} F_s[f](\vecm_{\pmb \ell}^\prime,\vecgkernel)\,\varpi(\vecm_{\pmb \ell}^\prime,\vecgkernel)\,e^{\ii \frac{\vecm_{\pmb \ell}^\prime\cdot(\pmb\gamma+\pmb\gamma^\prime)}{2}}\, ,
    \end{align}
   where $F_s[f]$ is the symplectic Fourier transform
    \begin{align}
     F_s[f](\vecm_{\pmb \ell},\bm \theta) = \frac{1}{\ell_1\ell_2}\sum_{\vecn \in \Z^2_{\pmb\ell}}\int_{\T^2_{\pmb \ell}}\ud^2\vecalpha\, f(\vecn,\vecalpha)\,e^{-\ii(\vecm_{\pmb \ell}\cdot\vecalpha-\vecn\cdot\bm \theta)}\,.
    \end{align}
\end{prop}
\noindent The proof of Proposition~\ref{CIQ_proof} follows from the composition rules of the operators $\hat U$ given in Appendix~\ref{weyl}. The symplectic Fourier transform and its conjugate,
\begin{equation}
\label{eq:sympl_conj}
\overline{F_s}[f](\vecm_{\pmb \ell},\bm \theta) = \frac{1}{\ell_1\ell_2}\sum_{\vecn \in \Z^2_{\pmb\ell}}\int_{\T^2_{\pmb \ell}}\ud^2\vecalpha\, f(\vecn,\vecalpha)\,e^{\ii(\vecm_{\pmb \ell}\cdot\vecalpha-\vecn\cdot\bm \theta)}\, ,
\end{equation}
are not the inverse of each other, but each is its own inverse, $\overline{F_s}^2 = 1 = F_s^2$, and they satisfy the parity relation
\begin{align}
    \overline{F_s}F_s = F_s\overline{F_s} = \mathsf{P},\qquad (\mathsf{P}f)(\vecm_{\pmb \ell},\bm b) = f(-\vecm_{\pmb \ell},-\bm b).
\end{align}
The integral kernel of Eq.~\eqref{A_kernel} thus reduces the quantization of a classical function to the computation of its symplectic Fourier transform.
%\begin{figure}[H]
    %\centering
   % \includegraphics[width=0.7\linewidth]{Quantization Flowchart.png}
  %  \caption{Starting from the top left, we show the general process of using Covariant Integral Quantization. One begins either with a function on the configuration space or on the phase space. If starting with the configuration space, we first find its phase space portrait. After having this function on the phase space, we then use Eqs.~\eqref{A_operator} and~\eqref{A_kernel} to quantize the function into an operator $\hat{A}_f^\varpi$. This operator is then applied back to a function on the configuration space.}
 %   \label{Quantization Flowchart}
%\end{figure}

\subsection{Semi-classical portrait}

Let $\vap$ and $\vap^\prime$ be two weight functions on $\Gamma$, normalized as $\vap(\bm{0},\bm{0}) = \vap^\prime(\bm 0,\bm 0) = 1$, such that $\hat{\sfM}^{\vap^\prime}$ is a positive operator of unit trace, that is, a density operator. An operator $\hat{A}$ on $\mathcal{H}$ is mapped onto the function on $\Gamma$
\begin{align}
    \label{semiClass}
    \check{A}(\vecm_{\pmb \ell},\vectheta) := \mathrm{Tr}\left[\hat{A}\, \hat U(\vecm_{\pmb \ell},\bm \theta)\hat{\sfM}^{\varpi^{\prime}}\,\hat U^{\dagger}(\vecm_{\pmb \ell},\bm \theta) \right] = \mathrm{Tr}\left[\hat{A}\,\hat{\sfM}^{\varpi^{\prime}}(\vecm_{\pmb \ell},\bm \theta)\right]\,.
\end{align}
When $\hat{A}$ is itself the quantization $\mathrm{Op}_f^{\vap}$ of a classical function $f$, this defines the semi-classical portrait of $f$,
\begin{align}
    f(\vecm_{\pmb \ell},\bm \theta) \mapsto \check{f}(\vecm_{\pmb \ell},\bm \theta) \equiv \widecheck{\mathrm{Op}}_f^{\vap,\vap^\prime}(\vecm_{\pmb \ell},\bm \theta) = \mathrm{Tr}\left[\mathrm{Op}_f^{\vap}\, \hat{\sfM}^{\varpi^{\prime}}(\vecm_{\pmb \ell},\bm \theta)\right],
\end{align}
which is the convolution
\begin{align*}
   \widecheck{\mathrm{Op}}_f^{\vap,\vap^\prime}(\vecm_{\pmb \ell},\bm \theta) &= \frac{1}{\ell_1\ell_2}\sum_{\vecm_{\pmb \ell}^{\prime}}\int \ud^2\bm \theta^{\prime}\, f(\vecm_{\pmb \ell}-\vecm_{\pmb \ell}^{\prime},\bm \theta-\bm \theta^{\prime})\, \mathrm{Tr}\left[\sfMv\, \hat{\sfM}^{\varpi^{\prime}}(\vecm_{\pmb \ell}^{\prime},\bm \theta^{\prime})\right] \,,
\end{align*}
and simplifies to
\begin{align}
\label{semiclassical simple}
    \widecheck{\mathrm{Op}}_f^{\vap,\vap^\prime}(\vecm_{\pmb \ell},\bm \theta) = \frac{1}{(\ell_1\ell_2)^2} \sum_{\vecm^{\prime}_{\pmb\ell},\vecn\in\Z^2_{\pmb\ell}}\int \ud^2\vectheta^{\prime} \int\ud^2\pmb\gamma\, f(\vecm_{\pmb \ell}^{\prime},\vectheta^{\prime})\, \vap(\vecn,\pmb\gamma)\,\vap^\prime(-\vecn,-\pmb\gamma) \nonumber \\\times\  e^{\ii [(\vecm_{\pmb\ell}^{\prime}-\vecm_{\pmb\ell})\cdot \pmb\gamma +\vecn\cdot(\vectheta-\vectheta^{\prime})]}.
\end{align}
Substituting $f(\vecm_{\pmb\ell},\vectheta) =1$ into Eq.~\eqref{semiclassical simple} gives back the constant function $1$. As an example, the semi-classical portrait is computed for the function
\begin{align}
\label{Pdist}
    f(\vecm,\vectheta) = N\,e^{-\left(\frac{m_1^2}{2\sigma_1^2} + \frac{m_2^2}{2\sigma_2^2}\right)}\cos^2\left(k_1\frac{2\pi }{\ell_1}\theta_1\right)\cos^2\left(k_2\frac{2\pi }{\ell_2}\theta_2\right),
\end{align}
where $N$ is a normalization factor and the integers $k_{1/2}$ control the number of zeros of $f(\vecm,\vectheta)$. This function is non-negative and normalized to one, and is therefore a probability distribution on $\Gamma$. The weight functions are chosen as
\begin{align}
\label{Pdist_weight}
    \vap(\vecm,\vectheta) = \vap^\prime(-\vecm,-\vectheta) = e^{-\frac{4\pi}{\ell_1\ell_2}\left(\frac{m_1^2}{2\sigma_1^2}+\frac{m_2^2}{2\sigma_2^2}\right)} \cos^2\left(c_1\theta_1+c_2\theta_2\right).
\end{align}
The parameters of the distribution and of the weight functions control the number and the positions of the zeros, so that the semi-classical portrait shows how the zeros of a classical distribution are transferred through the quantization.
\begin{figure*}[t]
\centering
(a)\\[1mm]
\includegraphics[width=0.8\textwidth]{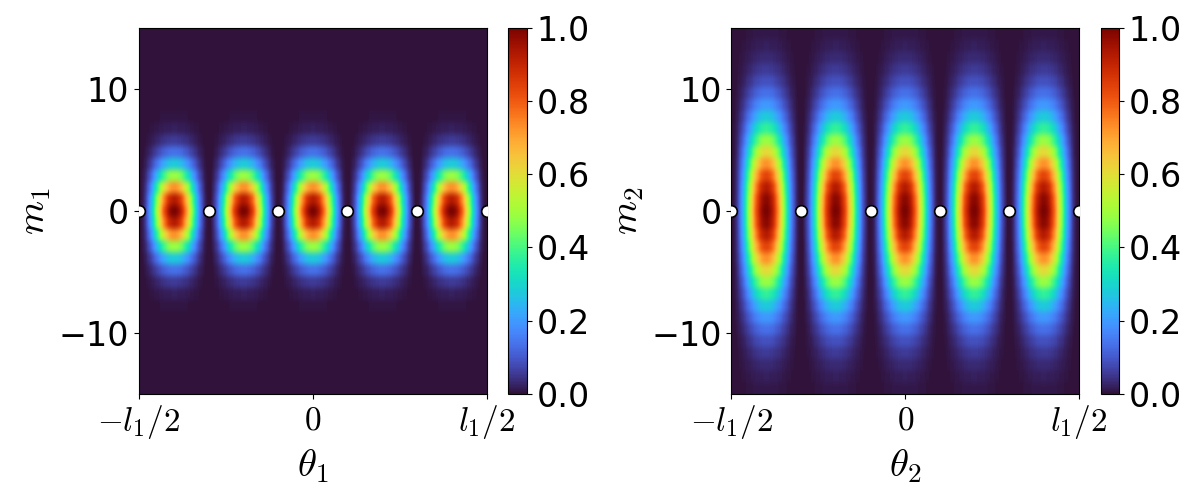}\\[3mm]
(b)\\[1mm]
\includegraphics[width=0.8\textwidth]{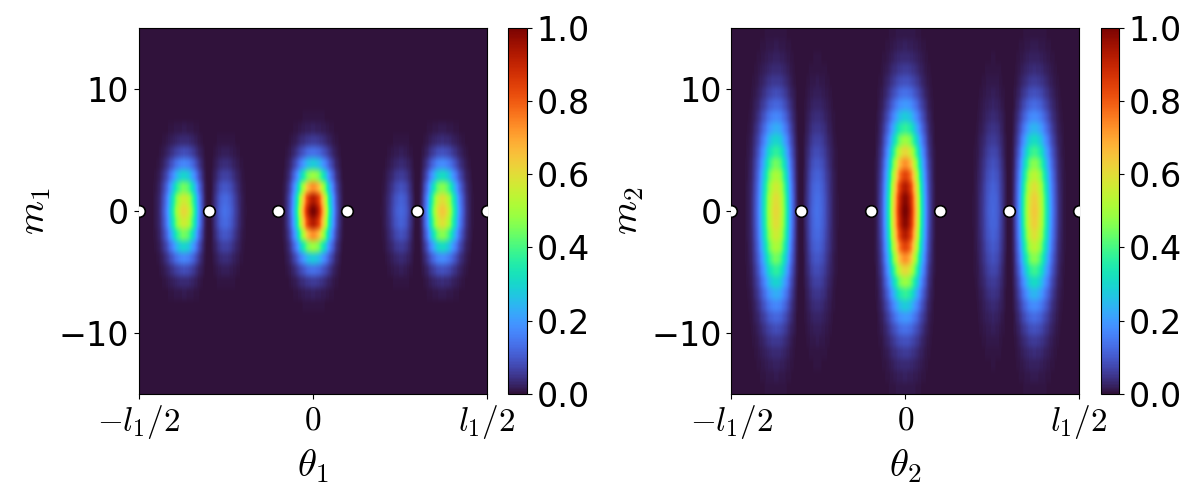}
\caption{(a)~Function $f(\vecm,\vectheta)$ of Eq.~\eqref{Pdist}, with $k_{1/2}=5$, $\sigma_1=2.5$ and $\sigma_2=5$. The four-dimensional distribution is represented by its two marginal distributions on the planes $(m_1,\theta_1)$ and $(m_2,\theta_2)$, obtained by summing over the two other variables. Along the lines $m_{1/2}=0$, the distribution has six zeros. (b)~Semi-classical portrait of $f(\vecm,\vectheta)$, computed with the weight $\vap(\vecm,\vectheta)=\vap^\prime(-\vecm,-\vectheta)$ of Eq.~\eqref{Pdist_weight}. The distribution has zeros along several directions of the phase space, of which those along $m_{1/2}=0$ are the most visible; the semi-classical portrait preserves these zeros.}
\label{semiclassical zeroes}
\end{figure*}
\noindent Figure~\ref{semiclassical zeroes}(a) shows the probability distribution, with its zeros labeled in white for visibility, and Fig.~\ref{semiclassical zeroes}(b) its semi-classical portrait, in which the zeros along $m_{1/2} = 0$ are preserved. The profile along $\theta_1$, shown in Fig.~\ref{semiclassical profile}, makes this conservation of the zeros explicit.
\begin{figure}[t]
\centering
\includegraphics[width=0.9\linewidth]{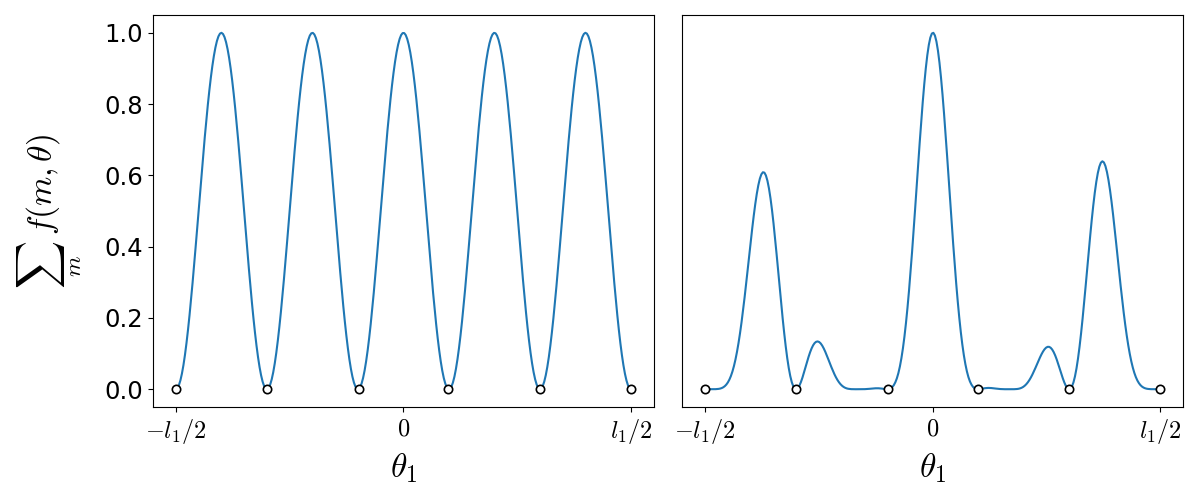}
\caption{Profile along $\theta_1$ of the function $f(\vecm,\vectheta)$ of Eq.~\eqref{Pdist} and of its semi-classical portrait, computed with the weight of Eq.~\eqref{Pdist_weight} and the parameters of Fig.~\ref{semiclassical zeroes}. The semi-classical portrait has a reduced amplitude, but its zeros appear at the same positions and in the same number as those of $f(\vecm,\vectheta)$. The profile along $\theta_2$ differs only by its amplitude.}
\label{semiclassical profile}
\end{figure}

\subsection{Wigner and Husimi distributions}
\label{section5}

In the plane, the Wigner distribution is the expectation value of the displaced parity operator, and the Husimi distribution is the expectation value of the displaced projector onto the vacuum. Both constructions follow from covariant integral quantization with two particular weight functions.

\subsubsection{Wigner distribution}

The unit weight
\begin{align}
    \varpi_P(\vecm_{\pmb \ell},\bm \theta) = 1
\end{align}
gives, through Eq.~\eqref{affbop}, the operator
\begin{align}
   \hat{\sfM}^1 = \frac{1}{\ell_1\ell_2}\sum_{\vecm_{\pmb \ell} \in \Z^2_{\pmb\ell}}\int_{\T^2_{\pmb \ell}}\ud^2\bm \theta \ \hat U(\vecm_{\pmb \ell},\bm \theta)\, ,
\end{align}
which acts on functions on the torus as four times the parity operator of Eq.~\eqref{Par_oper}, $(\hat{\sfM}^1\phi)(\pmb\gamma)=4\phi(-\pmb\gamma)$. Indeed, the kernel of Proposition~(i) contains, for each circle, the sum $\frac{1}{\ell_j}\sum_{m_j\in\Z}e^{\ii\pi m_j(\gamma_j+\gamma_j^{\prime})/\ell_j}=2\,\delta(\gamma_j+\gamma_j^{\prime})$, valid for $\gamma_j$ and $\gamma_j^{\prime}$ in $[-\ell_j/2,\ell_j/2)$. The Weyl transport of the parity operator then gives the Wigner distribution on $\Gamma = \Z^2_{\pmb\ell} \times \T^2_{\pmb \ell}$,
\begin{align}
   W_{\ket{\phi}}(\vecm_{\pmb \ell},\vectheta) = \frac{1}{4\ell_1\ell_2}\int_{-\ell_1}^{\ell_1}\ud z_1\int_{-\ell_2}^{\ell_2}\ud z_2\, e^{\ii\vecm_{\pmb \ell}\cdot\bm z} \,\overline{\phi\left(\vectheta+\frac{\bm z}{2}\right)}\phi\left(\vectheta-\frac{\bm z}{2}\right)\,,
\end{align}
normalized as $\sum_{\vecm_{\pmb \ell}\in\Z^2_{\pmb\ell}}\int_{\T^2_{\pmb \ell}}\ud^2\vectheta\,W_{\ket{\phi}}(\vecm_{\pmb \ell},\vectheta)=1$.
\begin{prop}
    The Wigner distribution on $\Gamma = \Z^2_{\pmb\ell}\times\T^2_{\pmb \ell}$ is the expectation value of the Weyl transport of the parity operator $\hat{\mathrm{P}}$,
    \begin{align}
         W_{\ket{\phi}}(\vecm_{\pmb \ell},\bm \theta) = \frac{1}{\ell_1\ell_2}\left<\phi|\hat{\mathrm{P}}(\vecm_{\pmb \ell},\bm \theta)\phi\right>, \qquad
         \hat{\mathrm{P}}(\vecm_{\pmb \ell},\bm \theta) = \hat U(\vecm_{\pmb \ell},\bm \theta)\hat{\mathrm{P}}\,\hat U^{\dagger}(\vecm_{\pmb \ell},\bm \theta)\,.
    \end{align}
\end{prop}
\noindent In the plane, the Wigner distribution satisfies the overlap formula
\begin{align}
     2\pi\int_{\R^2} \ud x\,\ud p\ W_{\hat\rho}(x,p)W_{\hat\rho'}(x,p)=\mathrm{Tr}(\hat{\rho}\hat{\rho}'),
\end{align}
with $\hbar=1$. When the Wigner distributions of the states, operations and measurements of a circuit are non-negative, this formula allows the circuit to be simulated efficiently by sampling~\cite{mari_positive_2012}, so that Wigner negativity is a necessary resource for a quantum computational advantage.
%\nicolas{Do not forget to cite some works, and I would advise to put the complete expression. }

\subsubsection{Husimi distribution}

The Husimi distribution follows from another choice of the weight function, in which the parity operator is replaced by the projector onto a fiducial state. On the double discrete cylinder, the role of the vacuum is played by the fiducial state
\begin{equation}
\label{eq:fiducial_vacuum}
\ket{\psi_0}=\ket{\phi_0}\otimes\ket{\phi_0},\qquad \ket{\phi_0}=A\sum_{j\in\Z}e^{-j^{2}/2}\ket{j},\qquad A=\Big(\sum_{j\in\Z}e^{-j^2}\Big)^{-1/2},
\end{equation}
where $\ket{j}$ denotes the Fourier basis state $\chi$ of integer label $j$ on each circle of the torus, so that $\ket{\psi_0}=A^2\sum_{m_1,m_2\in\Z}e^{-(m_1^2+m_2^2)/2}\ket{m_1,m_2}$. The state $\ket{\phi_0}$ is a Gaussian superposition of the integer states of a discrete cylinder, centered at the height zero, and its wavefunction on the torus is a periodized Gaussian (Sec.~\ref{Fid_vectors}). The weight function
\begin{equation}
\label{eq:husimi_weight}
\vap_0(\vecm_{\pmb \ell},\bm \theta)=\bra{\psi_0}\hat U^{\dagger}(\vecm_{\pmb \ell},\bm \theta)\ket{\psi_0}
\end{equation}
is normalized as $\vap_0(\bm 0,\bm 0)=1$. Since $\mathrm{Tr}[\hat U^{\dagger}(\vecm_{\pmb \ell}',\bm\theta')\hat U(\vecm_{\pmb \ell},\bm\theta)]=\ell_1\ell_2\,\delta_{\vecm_{\pmb \ell},\vecm_{\pmb \ell}'}\,\delta^2(\bm\theta-\bm\theta')$ (Appendix~\ref{weyl}), Eq.~\eqref{affbop} gives $\mathrm{Tr}[\hat U^{\dagger}(\vecm_{\pmb \ell},\bm\theta)\hat{\sfM}^{\vap}]=\vap(\vecm_{\pmb \ell},\bm\theta)$ for any weight, so that the weight of Eq.~\eqref{eq:husimi_weight} yields the projector onto the fiducial state,
\begin{equation}
\hat{\sfM}^{\vap_0}=\ket{\psi_0}\bra{\psi_0}.
\end{equation}
The Weyl transport of this projector, $\hat{\sfM}^{\vap_0}(\vecm_{\pmb \ell},\bm\theta)=\ket{\psi_{0;\vecm_{\pmb \ell},\bm\theta}}\bra{\psi_{0;\vecm_{\pmb \ell},\bm\theta}}$, is the projector onto the displaced fiducial state $\ket{\psi_{0;\vecm_{\pmb \ell},\bm\theta}}=\hat U(\vecm_{\pmb \ell},\bm\theta)\ket{\psi_0}$, which is a coherent state of the double discrete cylinder. In the same way as the Wigner distribution is the expectation value of the displaced parity operator, the Husimi distribution of a state $\ket{\psi}$ is the expectation value of the displaced projector onto the fiducial state,
\begin{equation}
\label{HusimiDef}
\begin{split}
Q_{\ket{\psi}}(\vecm_{\pmb \ell},\bm \theta)&=\bra{\psi}\hat{\sfM}^{\vap_0}(\vecm_{\pmb \ell},\bm\theta)\ket{\psi}=\bra{\psi}\hat U(\vecm_{\pmb \ell},\bm\theta)\ket{\psi_0}\bra{\psi_0}\hat U^{\dagger}(\vecm_{\pmb \ell},\bm\theta)\ket{\psi}\\
&=\big|\bra{\psi_{0;\vecm_{\pmb \ell},\bm\theta}}\ket{\psi}\big|^2 .
\end{split}
\end{equation}
Contrary to the Wigner distribution, the Husimi distribution is non-negative, and its non-classical features are carried by its zeros instead of its negativity. The displaced fiducial states of Eq.~\eqref{HusimiDef} are the wrapped coherent states of a particle on a circle, as shown in Sec.~\ref{sec:wrapped_cs}, and the zeros of the Husimi distribution are studied in Secs.~\ref{Stellar_section} and~\ref{sec:examples}.
%It is important to note that Eq.~\eqref{stell_rep1} contains the holomorphic portion of $\text{Tr}(\hat{\rho}\hat{\rho}^\prime)$ \nicolas{holomorphic portion??}.

%\nicolas{I will strongly suggest a schematic by bloc with the different steps. It will help the QIP community to adopt this formalism. }

\section{Modular variables and stellar representation on the double discrete cylinder}
\label{sec:modular}

This section applies the construction of Sec.~\ref{Torus_intro} to continuous variables, that is, to any pair of canonically conjugate variables obeying the Heisenberg algebra, called position and momentum in the following. The position and momentum are decomposed into integer and modular parts, which identifies the torus as the configuration space and the double discrete cylinder as the phase space. The modular wavefunction is the Zak transform of the wavefunction, and its expansion on an integer basis makes the phase-space distributions depend on the choice of the unit cell. The coherent states of the double discrete cylinder are then identified with the wrapped coherent states of a particle on a circle~\cite{Kowalski_1996}, from which the stellar function and the Husimi distribution are constructed.

\subsection{Modular position and momentum variables}
\label{sec:modular_variables}

The case of interest for continuous variables is the one in which the two periods introduced in Sec.~\ref{non-uniform 2-torus} are related by
\begin{equation}
\label{Fourier_rel}
\ell_2=\frac{2\pi}{\ell_1}.
\end{equation}
With this relation, one period fixes the other through the Fourier duality between a lattice of spacing $\ell_1$ and its reciprocal lattice, and the two periodic variables become conjugate in the same way as position and momentum. In the following, $\ell\equiv\ell_1$, so that $\ell_2=2\pi/\ell$, and the two periodic variables are renamed $\gamma_1\equiv\overline{x}$ and $\gamma_2\equiv\overline{p}$. Units are chosen such that $\hbar=1$, hence $[\hat x,\hat p]=\ii$.
The position and momentum operators are decomposed as
\begin{equation}
\label{eq:modular_decomposition}
\hat x=\ell\hat N_1+\hat{\overline{x}},\qquad \hat p=\frac{2\pi}{\ell}\hat N_2+\hat{\overline{p}},
\end{equation}
where $\hat N_1$ and $\hat N_2$ have integer spectra and the modular variables take the values $\overline{x}\in[-\ell/2,\ell/2)$ and $\overline{p}\in[-\pi/\ell,\pi/\ell)$. Since the product of the two periods equals $2\pi$, the modular operators commute, $[\hat{\overline{x}},\hat{\overline{p}}]=0$~\cite{aharonov_modular_1969,busch_what_1986}, although $[\hat x,\hat p]=\ii$. This property follows from the two displacement operators $e^{2\ii\pi\hat x/\ell}$ and $e^{-\ii\ell\hat p}$. With $\hat A=2\ii\pi\hat x/\ell$ and $\hat B=-\ii\ell\hat p$, the commutator $[\hat A,\hat B]=2\pi[\hat x,\hat p]=2\ii\pi$ is a scalar, and the relation $e^{\hat A}e^{\hat B}=e^{\hat B}e^{\hat A}e^{[\hat A,\hat B]}$, valid whenever $[\hat A,\hat B]$ commutes with $\hat A$ and $\hat B$, gives
\begin{equation}
\label{eq:commuting_displacements}
e^{2\ii\pi\hat x/\ell}\,e^{-\ii\ell\hat p}=e^{2\ii\pi}\,e^{-\ii\ell\hat p}\,e^{2\ii\pi\hat x/\ell}=e^{-\ii\ell\hat p}\,e^{2\ii\pi\hat x/\ell}.
\end{equation}
Each unitary determines its modular operator uniquely on the chosen interval, so that $\hat{\overline{x}}$ and $\hat{\overline{p}}$ are functions of two commuting unitaries, and therefore commute. The commuting pair $(\hat{\overline{x}},\hat{\overline{p}})$ admits the common eigenstates
\begin{equation}
\label{eq:zak_state}
\ket{\overline{x},\overline{p}}=\sqrt{\ell}\sum_{n\in\Z}e^{\ii n\ell\overline{p}}\ket{x=\overline{x}+n\ell},
\end{equation}
which are also eigenstates of the two displacement operators,
\begin{equation}
\label{eq:displacement_eigenvalues}
e^{2\ii\pi\hat x/\ell}\ket{\overline{x},\overline{p}}=e^{2\ii\pi\overline{x}/\ell}\ket{\overline{x},\overline{p}},\qquad e^{-\ii\ell\hat p}\ket{\overline{x},\overline{p}}=e^{-\ii\ell\overline{p}}\ket{\overline{x},\overline{p}}.
\end{equation}
The modular states $\ket{\overline{x},\overline{p}}$ are normalized as $\bra{\overline{x},\overline{p}}\ket{\overline{x}',\overline{p}'}=2\pi\,\delta(\overline{x}-\overline{x}')\delta(\overline{p}-\overline{p}')$ and resolve the identity as $\frac{1}{2\pi}\int_{-\ell/2}^{\ell/2}\!\!\dd\overline{x}\int_{-\pi/\ell}^{\pi/\ell}\!\!\dd\overline{p}\,\ket{\overline{x},\overline{p}}\bra{\overline{x},\overline{p}}=\bu$ (Appendix~\ref{app:modular}). The modular basis therefore satisfies the same orthogonality and completeness conditions as the basis of the torus introduced in Sec.~\ref{non-uniform 2-torus}, up to the factor $2\pi$, with the integration running over the unit cell $[-\ell/2,\ell/2)\times[-\pi/\ell,\pi/\ell)$ of area $2\pi$.

\subsection{The modular wavefunction as a Zak transform and quasi-periodicity conditions}

The modular wavefunction of a state $\ket\psi$,
\begin{equation}
\label{eq:modular_wf}
\psi(\overline{x},\overline{p})=\bra{\overline{x},\overline{p}}\ket{\psi}=\sqrt{\ell}\sum_{n\in\Z}\psi(\overline{x}+n\ell)\,e^{-\ii n\ell\overline{p}},
\end{equation}
is, up to the factor $\sqrt\ell$, the Zak transform of the wavefunction $\psi(x)$~\cite{zak_kq-representation_1972,ketterer_quantum_2016}. For a fixed $\overline{x}$, Eq.~\eqref{eq:modular_wf} is a Fourier series in $\overline{p}$ whose coefficients are the samples of $\psi(x)$ on the comb $\overline{x}+\ell\Z$. Inverting this Fourier series recovers the wavefunction,
\begin{equation}
\label{eq:inverse_zak}
\psi(\overline{x}+n\ell)=\frac{\sqrt\ell}{2\pi}\int_{-\pi/\ell}^{\pi/\ell}\dd\overline{p}\;e^{\ii n\ell\overline{p}}\,\psi(\overline{x},\overline{p}),
\end{equation}
and the resolution of the identity gives the decomposition of the state on the modular basis,
\begin{equation}
\label{eq:zak_decomposition}
\begin{gathered}
\ket\psi=\frac{1}{2\pi}\int_{-\ell/2}^{\ell/2}\!\!\dd\overline{x}\int_{-\pi/\ell}^{\pi/\ell}\!\!\dd\overline{p}\;\psi(\overline{x},\overline{p})\ket{\overline{x},\overline{p}},\\
\frac{1}{2\pi}\int_{-\ell/2}^{\ell/2}\!\!\dd\overline{x}\int_{-\pi/\ell}^{\pi/\ell}\!\!\dd\overline{p}\,|\psi(\overline{x},\overline{p})|^2=\int_\R|\psi(x)|^2\dd x =1.
\end{gathered}
\end{equation}
The Zak transform is thus a unitary map from $L^2(\R)$ onto functions on the unit cell, in the same way as the Fourier transform is a unitary map from $L^2(\R)$ onto itself.
 The substitution $n\to n+1$ in Eq.~\eqref{eq:modular_wf} and the periodicity of $e^{-\ii n\ell\overline{p}}$ give the quasi-periodicity conditions
\begin{equation}
\label{eq:quasi_periodicity}
\psi(\overline{x}+\ell,\overline{p})=e^{\ii\ell\overline{p}}\,\psi(\overline{x},\overline{p}),\qquad \psi\Big(\overline{x},\overline{p}+\frac{2\pi}{\ell}\Big)=\psi(\overline{x},\overline{p}).
\end{equation}

\subsection{Integer basis, configuration spaces and phase space}
\label{sec:config}

The Fourier series of the modular states on the unit cell defines the integer basis $\ket{k,q}$, $k,q\in\Z$:
\begin{equation}
\label{eq:integer_basis}
\begin{gathered}
\ket{\overline{x},\overline{p}}=\sum_{k,q\in\Z}e^{-2\ii\pi k\overline{x}/\ell}\,e^{-\ii\ell q\overline{p}}\ket{k,q},\\
\ket{k,q}=\frac{1}{2\pi}\int_{-\ell/2}^{\ell/2}\!\!\dd\overline{x}\int_{-\pi/\ell}^{\pi/\ell}\!\!\dd\overline{p}\;e^{2\ii\pi k\overline{x}/\ell}\,e^{\ii\ell q\overline{p}}\ket{\overline{x},\overline{p}}.
\end{gathered}
\end{equation}
For the special case of Eq.~\eqref{Fourier_rel}, the integer labels $(m_1,m_2)$ of Sec.~\ref{non-uniform 2-torus} become the labels $(k,q)$ of the integer basis, and the translation $\vectheta$ entering the operators of Sec.~\ref{sec:operators} becomes a translation $\bm b$ of the modular variables. In the position representation, inserting Eq.~\eqref{eq:zak_state} into Eq.~\eqref{eq:integer_basis} and using $\int_{-\pi/\ell}^{\pi/\ell}\dd\overline{p}\,e^{\ii(n+q)\ell\overline{p}}=\frac{2\pi}{\ell}\delta_{n,-q}$ gives
\begin{equation}
\label{eq:integer_basis_position}
\ket{k,q}=\frac{1}{\sqrt\ell}\int_{-\ell/2}^{\ell/2}\dd\overline{x}\;e^{2\ii\pi k\overline{x}/\ell}\ket{x=\overline{x}-q\ell}.
\end{equation}
The integer $q$ labels the cell of width $\ell$ centered at $x=-q\ell$, and the integer $k$ labels the Fourier harmonic within this cell. Since $e^{2\ii\pi kq}=1$, the wavefunction of $\ket{k,q}$ is the plane wave $\ell^{-1/2}e^{2\ii\pi kx/\ell}$ of momentum $2\pi k/\ell$, truncated to the cell centered at $x=-q\ell$. The state $\ket{k,q}$ therefore specifies the position with resolution $\ell$ and the momentum with resolution $2\pi/\ell$, and each state of the integer basis occupies a phase-space area $\ell\times2\pi/\ell=2\pi$, that is one Planck cell. Equivalently, the integer basis reflects the factorization $L^2(\R)\simeq L^2([-\ell/2,\ell/2))\otimes\ell^2(\Z)$ of the Hilbert space into a fine position within a cell, expanded on Fourier harmonics labeled by $k$, and a coarse position labeled by the cell index $q=-N_1$. The integer basis is the common eigenbasis of two commuting integer operators $\hat K$ and $\hat Q$, defined by
\begin{equation}
\label{eq:KQ}
\hat K\ket{k,q}=k\ket{k,q},\qquad \hat Q\ket{k,q}=q\ket{k,q},\qquad [\hat K,\hat Q]=0 .
\end{equation}
The operator $\hat Q=-\hat N_1$ counts the cells, and $\hat K$ counts the Fourier harmonics within a cell. The two displacement operators of Eq.~\eqref{eq:commuting_displacements} act on the integer basis as ladder operators,
\begin{equation}
\label{eq:KQ_ladder}
e^{2\ii\pi\hat{\overline{x}}/\ell}\ket{k,q}=\ket{k+1,q},\qquad e^{\ii\ell\hat{\overline{p}}}\ket{k,q}=\ket{k,q+1},
\end{equation}
as follows from Eq.~\eqref{eq:integer_basis_position}, with $e^{2\ii\pi\hat x/\ell}=e^{2\ii\pi\hat{\overline{x}}/\ell}$, $e^{\ii\ell\hat p}=e^{\ii\ell\hat{\overline{p}}}$ and $e^{\ii\ell\hat p}\ket{x}=\ket{x-\ell}$. The pairs $(\hat{\overline{x}},\hat K)$ and $(\hat{\overline{p}},\hat Q)$ are thus the angle and integer operators of two discrete cylinders, in the same way as the angle and the integer momentum of a particle on a circle. The coefficients of a state $\ket\psi$ in this orthonormal basis are
\begin{equation}
\label{eq:psi_kq}
\begin{gathered}
\psi_{kq}=\bra{k,q}\ket{\psi}=\frac{1}{\sqrt\ell}\int_{-\ell/2}^{\ell/2}\dd\overline{x}\;\psi(\overline{x}-q\ell)\,e^{-2\ii\pi k\overline{x}/\ell},\\
\sum_{k,q\in\Z}|\psi_{kq}|^2=\int_\R|\psi(x)|^2\dd x=1 .
\end{gathered}
\end{equation}

\noindent The coefficients $\psi_{kq}$ are the Fourier coefficients of the modular wavefunction $\psi(\overline{x},\overline{p})$ on the unit cell. Their Fourier series defines the function
\begin{equation}
\label{eq:periodic_representative}
\psi_{\mathrm{per}}(\overline{x},\overline{p})=\sum_{k,q\in\Z}\psi_{kq}\,e^{2\ii\pi k\overline{x}/\ell}\,e^{\ii\ell q\overline{p}},
\end{equation}
which is periodic in both variables and therefore a function on the torus $\T^2_{\pmb\ell}$. A Fourier series is periodic, whereas $\psi(\overline{x},\overline{p})$ is only quasi-periodic in $\overline{x}$. The function $\psi_{\mathrm{per}}(\overline{x},\overline{p})$ therefore coincides with $\psi(\overline{x},\overline{p})$ inside the unit cell, but jumps at the edges $\overline{x}=\pm\ell/2$, where $\psi(\ell/2,\overline{p})=e^{\ii\ell\overline{p}}\psi(-\ell/2,\overline{p})$ differs from $\psi(-\ell/2,\overline{p})$. The size of this jump is set by the value of $\psi(\overline{x},\overline{p})$ at the edges, and slows the decay of the coefficients $\psi_{kq}$ with $k$ down to $1/k$, as for the Fourier series of any function with a discontinuity. The map $\ket\psi\mapsto\psi_{\mathrm{per}}$ is unitary from $L^2(\R)$ onto $L^2(\T^2_{\pmb\ell})$, equipped with the inner product
\begin{align}
    \bra{\phi}\ket{\psi}=\frac{1}{2\pi}\int_{-\ell/2}^{\ell/2}\!\!\dd\overline{x}\int_{-\pi/\ell}^{\pi/\ell}\!\!\dd\overline{p}\,\overline{\phi_{\mathrm{per}}(\overline{x},\overline{p})}\,\psi_{\mathrm{per}}(\overline{x},\overline{p}),
\end{align}
and it is through this map that the functions on the torus of Sec.~\ref{non-uniform 2-torus} represent continuous-variable states. Consequently, the integer basis, and the coherent states and Husimi distribution built on the integer basis below, depend on where the edges of the unit cell are placed along $\overline{x}$. The edges are placed at $\overline{x}=\pm\ell/2$ throughout. The choice of the unit cell along $\overline{p}$ plays no role, since $\psi(\overline{x},\overline{p})$ is periodic in $\overline{p}$.

The torus $\T^2_{\pmb\ell}=\S_{r_1}\times\S_{r_2}$ of the modular variables is the product of two circles of circumferences $\ell$ and $2\pi/\ell$, that is of radii $r_1=\ell/(2\pi)$ and $r_2=1/\ell$. Its dual is the grid $\Z^2_{\pmb\ell}=\frac{2\pi}{\ell}\Z\times\ell\Z$, which carries the discrete variables $2\pi k/\ell$ and $\ell q$: the functions $e^{\ii(2\pi k\overline{x}/\ell+\ell q\overline{p})}$, with $(k,q)\in\Z^2$, are exactly the plane waves that are periodic on $\T^2_{\pmb\ell}$. The phase space
\begin{equation}
\label{phasespace_def}
\Gamma=\Z^2_{\pmb\ell}\times\T^2_{\pmb\ell}
\end{equation}
is thus the product of a configuration space and of its dual, in the same way as the usual phase space $\R^2=\R\times\R$ is the product of the real line and of its dual. Grouping each circle with its dual discrete variable, $\Gamma\simeq\big(\S_{r_1}\times\frac{2\pi}{\ell}\Z\big)\times\big(\S_{r_2}\times\ell\Z\big)$ is the product of two discrete cylinders, which gives its name to the double discrete cylinder. The four variables $(\overline{x},\overline{p},k,q)$ can be split in four ways into a configuration space and its dual, which are listed in Table~\ref{configuration space breakdown}. The two mixed bases are defined from the integer basis by partial Fourier series,
\begin{equation}
\label{eq:mixed_bases}
\begin{gathered}
\ket{\overline{x},q}=\sum_{k\in\Z}e^{-2\ii\pi k\overline{x}/\ell}\ket{k,q}=\sqrt\ell\,\ket{x=\overline{x}-q\ell},\\
\ket{k,\overline{p}}=\sum_{q\in\Z}e^{-\ii\ell q\overline{p}}\ket{k,q},
\end{gathered}
\end{equation}
where the second equality for $\ket{\overline{x},q}$ follows from Eq.~\eqref{eq:integer_basis_position} and from the relation $\sum_{k\in\Z}e^{2\ii\pi k(\overline{x}'-\overline{x})/\ell}=\ell\,\delta(\overline{x}'-\overline{x}),$ valid for $\overline{x}$ and $\overline{x}'$ in the unit cell. The configuration space $\S_{r_1}\times\ell\Z$ is therefore the real line cut into cells of width $\ell$, and the position eigenstates are recovered up to normalization. Summing the remaining index in either mixed basis gives back the modular states, $\ket{\overline{x},\overline{p}}=\sum_q e^{-\ii\ell q\overline{p}}\ket{\overline{x},q}=\sum_k e^{-2\ii\pi k\overline{x}/\ell}\ket{k,\overline{p}}$, consistently with Eq.~\eqref{eq:zak_state}.
\begin{table*}[t]
\caption{Splittings of the phase space $\Gamma$ of Eq.~\eqref{phasespace_def} into a configuration space and its dual space, with the corresponding bases. The two remaining splittings are obtained by exchanging the configuration space and the dual space in each row. All splittings share the same phase space $\Gamma$.}
\label{configuration space breakdown}
\begin{ruledtabular}
\begin{tabular}{cccc}
Configuration space & Basis & Dual space & Basis \\
\hline
$\T^2_{\pmb\ell}=\S_{r_1}\times\S_{r_2}$ & $\ket{\overline{x},\overline{p}}$ & $\Z^2_{\pmb\ell}=\frac{2\pi}{\ell}\Z\times\ell\Z$ & $\ket{k,q}$ \\
$\S_{r_1}\times\ell\Z$ & $\ket{\overline{x},q}$ & $\frac{2\pi}{\ell}\Z\times\S_{r_2}$ & $\ket{k,\overline{p}}$
\end{tabular}
\end{ruledtabular}
\end{table*}

\noindent The torus $\T^2_{\pmb\ell}$ is taken as the configuration space in the following, so that states are represented by the periodic functions $\psi_{\mathrm{per}}(\overline{x},\overline{p})$ of Eq.~\eqref{eq:periodic_representative}, and their dual variables are the integers $(k,q)$. The vectors
\begin{equation}
\bm z=\binom{\overline{x}}{\overline{p}},\qquad \vecm_{\pmb\ell}=\binom{2\pi k/\ell}{\ell q}
\end{equation}
denote a point of the torus and a point of the grid $\Z^2_{\pmb\ell}$, respectively, so that $\vecm_{\pmb\ell}\cdot\bm z=2\pi k\overline{x}/\ell+\ell q\overline{p}$, and a translation on the torus is denoted by $\bm b$.

\subsection{Wrapped coherent states}
\label{sec:wrapped_cs}
\label{Fid_vectors}

\noindent The coherent states of the double discrete cylinder, on which the Husimi distribution of Eq.~\eqref{HusimiDef} is built, are obtained by displacing the fiducial state $\ket{\psi_0}$ of Eq.~\eqref{eq:fiducial_vacuum} with the operators $\hat U(\vecm_{\pmb\ell},\bm b)$. Any nonzero fiducial vector on the torus generates a family of coherent states that resolves the identity; the general construction, the reproducing kernels and the periodization of fiducial vectors are given in Appendix~\ref{app:periodization}. The fiducial state of Eq.~\eqref{eq:fiducial_vacuum} singles out the wrapped coherent states of a particle on a circle~\cite{Kowalski_1996}.

For a particle on a circle, with angle operator $\hat\phi$ and integer operator $\hat N$ satisfying $[\hat\phi,\hat N]=\ii$, the wrapped coherent states are the eigenstates of the operator
\begin{equation}
\label{eq:Z_op}
\hat Z=e^{\ii\hat\phi-\hat N}=e^{\ii\hat\phi}e^{-\hat N-1/2},\qquad \hat Z\ket{n}=e^{-(n+1/2)}\ket{n+1},
\end{equation}
where $\ket{n}$ are the eigenstates of $\hat N$~\cite{Kowalski_1996}. The operator $\hat Z$ is not unitary, and is therefore not a displacement operator. The wrapped coherent states read
\begin{equation}
\label{eq:wrapped_cs}
\begin{gathered}
\ket{m,\phi}=A\,e^{-m^2/2}\sum_{k\in\Z}e^{-k^2/2}\,e^{k(m-\ii\phi)}\ket{k},\\
\hat Z\ket{m,\phi}=e^{-(m-\ii\phi)}\ket{m,\phi},\qquad A=\Big(\sum_{k\in\Z}e^{-k^2}\Big)^{-1/2},
\end{gathered}
\end{equation}
where the height $m$ and the angle $\phi$ combine into the complex eigenvalue label $m-\ii\phi$, and $\ket{m=0,\phi=0}=\ket{\phi_0}$ is the fiducial state of Eq.~\eqref{eq:fiducial_vacuum}. On the double discrete cylinder, the two angles are $\phi_1=2\pi\overline{x}/\ell$ and $\phi_2=\ell\overline{p}$, and the two integer operators are the operators $\hat K$ and $\hat Q$ of Eq.~\eqref{eq:KQ}. The coherent states are the joint eigenstates of $\hat Z_1=e^{2\ii\pi\hat{\overline{x}}/\ell-\hat K}$ and $\hat Z_2=e^{\ii\ell\hat{\overline{p}}-\hat Q}$,
\begin{equation}
\label{eq:coherent}
\ket{m_1,\overline{x};m_2,\overline{p}}=A^2e^{-(m_1^2+m_2^2)/2}\sum_{k,q\in\Z}e^{-(k^2+q^2)/2}
 e^{k(m_1-2\ii\pi\overline{x}/\ell)}\,e^{q(m_2-\ii\ell\overline{p})}\ket{k,q},
\end{equation}
where the integers $m_1,m_2\in\Z$, called heights, label the points of the grid $\Z^2_{\pmb\ell}$ (Appendix~\ref{app:coherent}). The phase space is made of the position cylinder, with points $(m_1,\overline{x})$, and of the momentum cylinder, with points $(m_2,\overline{p})$. At fixed heights, the points $(\overline{x},\overline{p})$ form a copy of the torus, called the slice $(m_1,m_2)$ of $\Gamma$. The torus representation of the coherent states illustrates their construction. For integer heights, the function on the torus that represents the coherent state of Eq.~\eqref{eq:coherent} is the displaced fiducial vector (Appendix~\ref{app:coherent})
\begin{equation}
\label{eq:cs_fiducial}
\begin{gathered}
\bra{\overline{x}',\overline{p}'}\ket{m_1,\overline{x};m_2,\overline{p}}=e^{-\ii\vecm_{\pmb\ell}\cdot\bm b/2}\,\big(\hat U(\vecm_{\pmb\ell},\bm b)\psi_0\big)(\bm z'),\\
\psi_0=A^2\,\textsf{Per}\,\left(G_{\bm r}\right),\qquad G_{\bm r}(x,y)=e^{-x^2/2r_1^2-y^2/2r_2^2},
\end{gathered}
\end{equation}
with $\bm z^{\prime}=(\overline{x}^{\prime},\overline{p}^{\prime})$, $\bm b=(\overline{x},\overline{p})$, $\vecm_{\pmb\ell}=(2\pi m_1/\ell,\ell m_2)$, the displacement operator $\hat U$ of Eq.~\eqref{eq:U_def}, and the periodization operator $\textsf{Per}$, which sums the copies of a function translated by all the periods of the torus (Appendix~\ref{app:periodization}). The fiducial vector $\psi_0$ is thus a periodized Gaussian, or wrapped Gaussian, whose widths equal the radii $r_1=\ell/(2\pi)$ and $r_2=1/\ell$ of the two circles of the torus. Figure~\ref{Wrapped_gauss_visual} shows an example of a wrapped Gaussian on the torus: a Gaussian envelope, stratified by narrower Gaussian packets, which are the copies translated by the periods of the torus (Appendix~\ref{app:periodization}). The coherent states of Eq.~\eqref{eq:coherent} are normalized, peaked at $(\overline{x},\overline{p})$, and resolve the identity (Appendix~\ref{app:coherent}),
\begin{equation}
\label{eq:resolution}
\frac{1}{2\pi}\sum_{m_1,m_2\in\Z}\int_{-\ell/2}^{\ell/2}\!\!\dd\overline{x}\int_{-\pi/\ell}^{\pi/\ell}\!\!\dd\overline{p}\;\ket{m_1,\overline{x};m_2,\overline{p}}\bra{m_1,\overline{x};m_2,\overline{p}}=\bu .
\end{equation}
\begin{figure*}[t]
\centering
\includegraphics[width=0.5\textwidth]{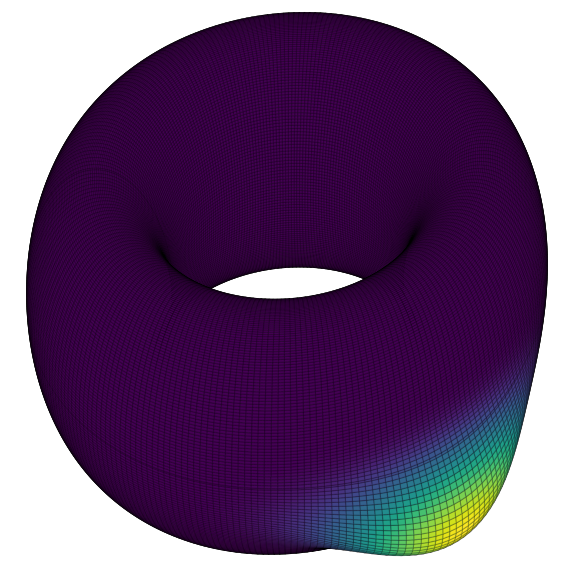}\\[1.2cm]
\begin{minipage}[t]{0.48\textwidth}
\centering
\includegraphics[width=\textwidth]{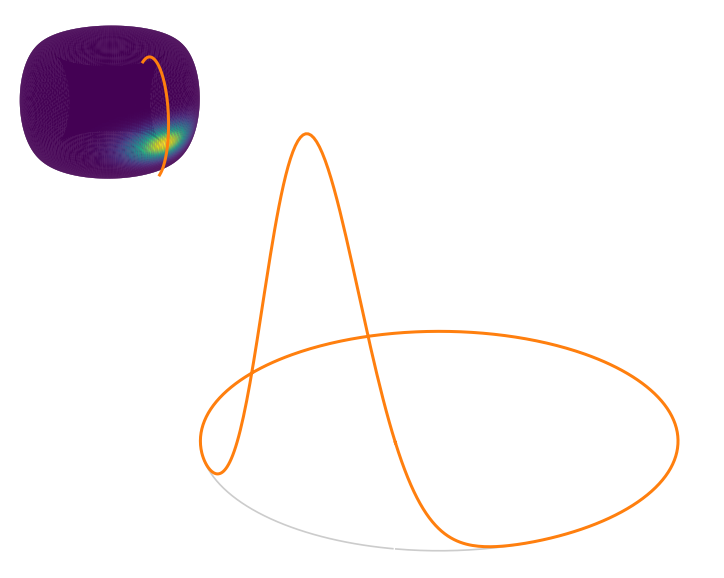}\\
(a)
\end{minipage}\hfill
\begin{minipage}[t]{0.48\textwidth}
\centering
\includegraphics[width=\textwidth]{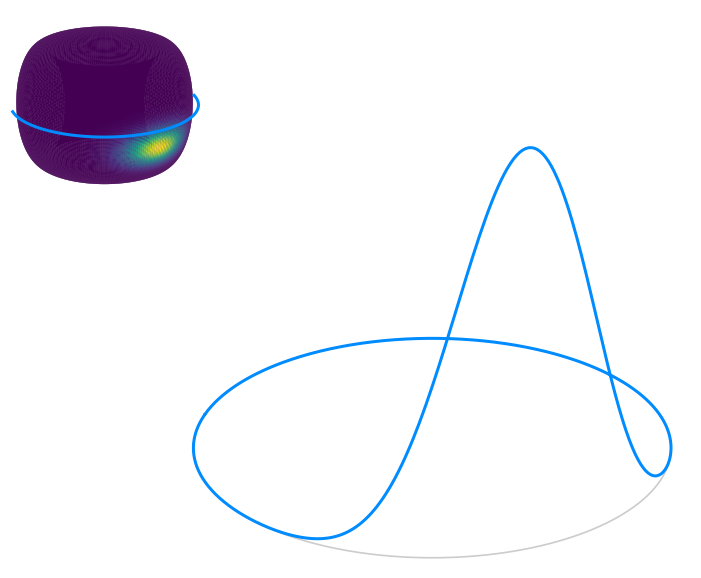}\\
(b)
\end{minipage}
\caption{Wrapped Gaussian of Eq.~\eqref{wrapped_gaussian}, a periodized Gaussian of the type used as fiducial vector of the coherent states on the torus, with the period $\ell=\sqrt{\pi}$. Since the periodization operator is not unitary, the function is normalized after periodization, and its modulus is represented. Top: modulus on the torus. (a)~Cut along $\overline{x}$, with $\sigma_1=0.1$. (b)~Cut along $\overline{p}$, with $\sigma_2=0.2$.}
\label{Wrapped_gauss_visual}
\end{figure*}

\subsection{Stellar function and Husimi distribution}
\label{Stellar_section}
\label{sec:stellar_husimi}

The Wigner distribution on the double cylinder~\cite{fabre_wigner_2020} is adapted to translationally invariant states and removes the redundancy of their planar representation. The stellar function and the Husimi distribution are now constructed on the same phase space from the wrapped coherent states. In the plane, the number of zeros of the stellar function defines the stellar rank, which quantifies the non-Gaussian character of a state~\cite{chabaud_stellar_2020}. On the double discrete cylinder, the relevant zeros are those lying on the lattice of heights of the phase space, which leads to the definition of a stellar rank on $\Gamma$.

\subsubsection{Stellar function}

The stellar function $F_{\ket{\psi}}(z_1,z_2)$ of a state $\ket{\psi}$ is the overlap between $\ket{\psi}$ and the coherent states, without the Gaussian prefactor of Eq.~\eqref{eq:coherent},
\begin{equation}
\label{eq:stellar}
F_{\ket{\psi}}(z_1,z_2)=e^{(m_1^2+m_2^2)/2}\bra{m_1,\overline{x};m_2,\overline{p}}\ket{\psi}=A^2\sum_{k,q\in\Z}\psi_{kq}\,e^{-(k^2+q^2)/2}\,e^{kz_1+qz_2},
\end{equation}
with the coefficients $\psi_{kq}$ of Eq.~\eqref{eq:psi_kq} and the complex variables
\begin{equation}
\label{eq:z1z2}
z_1=m_1+\frac{2\ii\pi}{\ell}\overline{x},\qquad z_2=m_2+\ii\ell\overline{p} .
\end{equation}
As in the plane~\cite{chabaud_stellar_2020}, removing the Gaussian prefactor leaves a function that depends holomorphically on the complex variables $(z_1,z_2)$, whose real parts are the heights and whose imaginary parts are the angles.

\subsubsection{Husimi distribution and its marginals}

The Husimi distribution of Eq.~\eqref{HusimiDef}, built on the coherent states of Eq.~\eqref{eq:coherent}, follows from the stellar function as
\begin{equation}
\label{eq:husimi}
Q_{\ket{\psi}}(m_1,\overline{x};m_2,\overline{p})=\big|\bra{m_1,\overline{x};m_2,\overline{p}}\ket{\psi}\big|^2=e^{-m_1^2-m_2^2}\,|F_{\ket{\psi}}(z_1,z_2)|^2,
\end{equation}
and is normalized by Eq.~\eqref{eq:resolution} as
\begin{equation}
\frac{1}{2\pi}\sum_{m_1,m_2\in\Z}\int_{-\ell/2}^{\ell/2}\!\!\dd\overline{x}\int_{-\pi/\ell}^{\pi/\ell}\!\!\dd\overline{p}\,Q_{\ket{\psi}}(m_1,\overline{x};m_2,\overline{p})=1 .
\end{equation}
The weight carried by the slice $(m_1,m_2)$ is (Appendix~\ref{app:coherent})
\begin{equation}
\label{eq:slice_weight}
\begin{split}
w_{\ket{\psi}}(m_1,m_2)&=\frac{1}{2\pi}\int_{-\ell/2}^{\ell/2}\!\!\dd\overline{x}\int_{-\pi/\ell}^{\pi/\ell}\!\!\dd\overline{p}\;Q_{\ket{\psi}}(m_1,\overline{x};m_2,\overline{p})\\
&=A^4\sum_{k,q\in\Z}|\psi_{kq}|^2\,e^{-(k-m_1)^2-(q-m_2)^2}.
\end{split}
\end{equation}
Summing $Q_{\ket{\psi}}$ over one of the two cylinders defines the marginal distributions on the position and momentum cylinders,
\begin{equation}
\label{eq:cylinder_husimi}
\begin{gathered}
Q_1(m_1,\overline{x})=\frac{1}{2\pi}\sum_{m_2\in\Z}\int_{-\pi/\ell}^{\pi/\ell}\dd\overline{p}\;Q_{\ket{\psi}}(m_1,\overline{x};m_2,\overline{p}),\\
Q_2(m_2,\overline{p})=\frac{1}{2\pi}\sum_{m_1\in\Z}\int_{-\ell/2}^{\ell/2}\dd\overline{x}\;Q_{\ket{\psi}}(m_1,\overline{x};m_2,\overline{p}),
\end{gathered}
\end{equation}
normalized as $\sum_{m_1\in\Z}\int_{-\ell/2}^{\ell/2}\dd\overline{x}\,Q_1(m_1,\overline{x})=\sum_{m_2\in\Z}\int_{-\pi/\ell}^{\pi/\ell}\dd\overline{p}\,Q_2(m_2,\overline{p})=1$. These two marginals represent a state on two cylinders, which can be drawn, whereas $\Gamma$ itself is four-dimensional. In general, they are marginals of the Husimi distribution and not Husimi distributions of the state on each cylinder. For a separable state, $\psi_{kq}=f_kg_q$, the Husimi distribution factorizes into the Husimi distributions of $f$ on the position cylinder and of $g$ on the momentum cylinder, which are proportional to $Q_1(m_1,\overline{x})$ and $Q_2(m_2,\overline{p})$, and the double discrete cylinder splits into two independent discrete cylinders. The marginals are needed for non-separable states such as the vacuum, whose modular wavefunction couples the two modular variables beyond the single-cell approximation. Integrating $Q_1$ and $Q_2$ over the angles gives the weights of the heights,
\begin{equation}
\label{eq:height_weights}
\begin{gathered}
P_1(m_1)=\sum_{m_2\in\Z}w_{\ket{\psi}}(m_1,m_2)=A^2\sum_{k,q\in\Z}|\psi_{kq}|^2e^{-(k-m_1)^2},\\
P_2(m_2)=A^2\sum_{k,q\in\Z}|\psi_{kq}|^2e^{-(q-m_2)^2}.
\end{gathered}
\end{equation}

\subsubsection{Zeros of the stellar function and stellar rank}

The stellar function $F_{\ket{\psi}}(z_1,z_2)$ is an entire function of $(z_1,z_2)\in\C^2$, $2\ii\pi$-periodic in each variable, whose zeros are counted with the argument principle of complex analysis (Appendix~\ref{app:lemma}). The zeros relevant to the Husimi distribution are those lying on $\Gamma$, with integer real parts $\mathrm{Re}\,z_1=m_1$ and $\mathrm{Re}\,z_2=m_2$, and their number, counted with their multiplicities, defines the stellar rank on $\Gamma$. Most of the examples of Sec.~\ref{sec:examples} reduce to the theta sums
\begin{equation}
\label{eq:theta_sum}
S_a(\zeta)=\sum_{k\in\Z}e^{-ak^2}e^{k\zeta},\qquad a>0,\quad \zeta\in\C,
\end{equation}
for which the argument principle gives a simple result: $S_a(\zeta)$ vanishes if and only if $\zeta=(2n+1)a+\ii(2p+1)\pi$, with $n,p\in\Z$, so that each period cell of the complex plane contains exactly one zero, which is simple and located at the center of the cell (Lemma~\ref{lemma:zeros} of Appendix~\ref{app:lemma}).

\section{Stellar representation of position eigenstates, of the vacuum and of GKP states}
\label{sec:examples}

This section applies the stellar representation on the double discrete cylinder of Sec.~\ref{Stellar_section} to several examples. Position eigenstates and ideal GKP states, whose modular wavefunctions are concentrated on circles or points of the torus, have stellar functions given in closed form by the theta sums of Eq.~\eqref{eq:theta_sum}. The vacuum, which is not translationally invariant, requires the exact coefficients in the integer basis. Physical GKP states are treated in the large-comb approximation, in which their stellar functions factorize.

\subsection{Position eigenstates: zeros between the lattice sites}

The position eigenstate
\begin{equation}
\ket{x=n\ell+\overline{x}_0}=\frac{\sqrt\ell}{2\pi}\int_{-\pi/\ell}^{\pi/\ell}\dd\overline{p}\;e^{-\ii n\ell\overline{p}}\ket{\overline{x}_0,\overline{p}}
\end{equation}
has the coefficients $\psi_{kq}=\ell^{-1/2}e^{-2\ii\pi k\overline{x}_0/\ell}\delta_{q,-n}$ (Appendix~\ref{app:modular}) and the stellar function
\begin{equation}
F_{\ket{\psi}}(z_1,z_2)=\frac{A^2}{\sqrt\ell}\,e^{-n^2/2}\,e^{-nz_2}\,S_{1/2}\Big(z_1-\frac{2\ii\pi}{\ell}\overline{x}_0\Big).
\end{equation}
The Husimi distribution of the position eigenstate,
\begin{equation}
Q_{\ket{\psi}}(m_1,\overline{x};m_2,\overline{p})=\frac{A^4}{\ell}\,e^{-(m_2+n)^2}\,e^{-m_1^2}\Big|S_{1/2}\Big(m_1+\frac{2\ii\pi}{\ell}(\overline{x}-\overline{x}_0)\Big)\Big|^2,
\end{equation}
is uniform in $\overline{p}$ and Gaussian in $m_2$ around $-n$. By Lemma~\ref{lemma:zeros}, the zeros of $Q_{\ket{\psi}}$ lie at the half-integer heights $m_1$ and at $\overline{x}=\overline{x}_0+\ell/2$, opposite to the position of the state on the position cylinder: the position eigenstate has no zero on $\Gamma$.

\subsection{Vacuum and Gaussian states centered at the origin}
\label{sec:coherent_state}

\subsubsection{Single-cell approximation}

A Gaussian state of width $\sigma$ centered at the origin,
\begin{equation}
\label{eq:coherent_wf}
\psi(x)=(\pi\sigma^2)^{-1/4}\,e^{-x^2/2\sigma^2},
\end{equation}
reduces to the vacuum for $\sigma=1$, in units where $\hbar=1$. The Gaussian state is not invariant under discrete translations. When $\sigma\ll\ell$, only the cell centered at the origin contributes to the Zak transform of Eq.~\eqref{eq:modular_wf}, and the modular wavefunction reduces to (Appendix~\ref{app:coherent_state})
\begin{equation}
\label{eq:coherent_single_cell}
\psi(\overline{x},\overline{p})\simeq\sqrt\ell\,(\pi\sigma^2)^{-1/4}\,G_\sigma(\overline{x})\,H(\overline{p}),
\end{equation}
where $G_\sigma(\overline{x})=e^{-\overline{x}^2/2\sigma^2}$ and $H(\overline{p})=1$ is a rectangular function of height one covering the whole momentum cell, of width $2\pi/\ell$. Although the Gaussian state is not translationally invariant, its single-cell modular wavefunction is separable. This single-cell form is periodic rather than quasi-periodic in $\overline{x}$, and violates Eq.~\eqref{eq:quasi_periodicity} by terms of relative order $e^{-\ell^2/8\sigma^2}$. In particular, the single-cell form does not vanish, whereas the exact modular wavefunction vanishes at the corner $(\overline{x},\overline{p})=(\ell/2,\pi/\ell)$ of the unit cell, where the terms $n$ and $-1-n$ of the Zak transform cancel.

\subsubsection{Exact coefficients and stellar function}

The coefficients in the integer basis are obtained exactly from Eq.~\eqref{eq:psi_kq}, since each cell contains a portion of the Gaussian (Appendix~\ref{app:coherent_state}):
\begin{equation}
\label{eq:psi_kq_coherent}
\begin{gathered}
\psi_{kq}=(\pi\sigma^2)^{-1/4}\sigma\sqrt{\frac{\pi}{2\ell}}\;e^{-\frac{2\pi^2\sigma^2}{\ell^2}k^2}\Big[\mathrm{erf}\big(\xi^+_{q}+\ii\beta_k\big)-\mathrm{erf}\big(\xi^-_{q}+\ii\beta_k\big)\Big],\\
\xi^\pm_q=\Big(\pm\frac12-q\Big)\frac{\ell}{\sqrt2\,\sigma},\qquad \beta_k=\frac{\sqrt2\pi\sigma k}{\ell}.
\end{gathered}
\end{equation}
In the single-cell approximation, the flat momentum function $H(\overline{p})$ only retains $q=0$, and the stellar function
\begin{equation}
\label{eq:coherent_stellar}
\begin{gathered}
F_{\ket{\psi}}(z_1,z_2)\simeq A^2(\pi\sigma^2)^{-1/4}\sigma\sqrt{\frac{\pi}{2\ell}}\sum_{k\in\Z}e^{-ak^2}e^{kz_1}\\
\times\Big[\mathrm{erf}\Big(\frac{\ell}{2\sqrt2\,\sigma}+\ii\beta_k\Big)+\mathrm{erf}\Big(\frac{\ell}{2\sqrt2\,\sigma}-\ii\beta_k\Big)\Big],\qquad a=\frac12+\frac{2\pi^2\sigma^2}{\ell^2},
\end{gathered}
\end{equation}
depends on $z_1$ only. The Husimi distribution is then uniform in $\overline{p}$ and Gaussian along the heights of the momentum cylinder, $P_2(m_2)\propto e^{-m_2^2}$, while it extends over $|m_1|\lesssim\ell/(2\pi\sigma)$ on the position cylinder. For $\sigma\ll\ell$, the bracket of Eq.~\eqref{eq:coherent_stellar} is equal to $2$ up to terms of order $e^{-\ell^2/8\sigma^2}$, and $F_{\ket{\psi}}(z_1,z_2)\propto S_a(z_1)$. By Lemma~\ref{lemma:zeros}, the zeros of $F_{\ket{\psi}}$ then form the circles $m_1=(2n+1)a$, $\overline{x}=\ell/2$, for all $(m_2,\overline{p})$. These heights are not integers for generic widths, and $Q_{\ket{\psi}}$ has no zero on $\Gamma$ in the single-cell approximation.

\subsubsection{Zeros of the Husimi distribution of the vacuum}

\noindent Beyond the single-cell approximation, the tails of the Gaussian in the neighboring cells, $q\neq0$, couple the two modular variables through the $k$-dependence of the error functions in Eq.~\eqref{eq:psi_kq_coherent}. Since $\psi(x)$ is real and even, the stellar function is real on the circles $\overline{x}\in\{0,\ell/2\}$ of the slices $(m_1,0)$ and on the circles $\overline{p}\in\{0,\pi/\ell\}$ of the slices $(0,m_2)$ (Appendix~\ref{app:coherent_state}). On each of these circles, $F_{\ket{\psi}}(z_1,z_2)$ is therefore a real function of the angle that parametrizes the circle, $\overline{p}$ or $\overline{x}$. A continuous real function that takes values of both signs along a circle vanishes at least twice on it, and $Q_{\ket{\psi}}=e^{-m_1^2-m_2^2}|F_{\ket{\psi}}|^2$ vanishes at the same points. In the single-cell approximation, $F_{\ket{\psi}}$ depends on $z_1$ only, and is therefore constant along the circles $\overline{x}\in\{0,\ell/2\}$ of the slices $(m_1,0)$, which are parametrized by $\overline{p}$. The coupling between the two modular variables makes this real function vary with $\overline{p}$, and it vanishes wherever it changes sign. Such zeros persist under small changes of the widths.

For the vacuum with the period of the square GKP code, $\sigma=1$ and $\ell=2\sqrt\pi$, the marginal distributions are concentrated on a few heights of both cylinders, as shown in Fig.~\ref{fig:coherent}(a)--(c), and the Husimi distribution has zeros on $\Gamma$, located numerically with the method of Appendix~\ref{app:numerics}. The first zeros lie on the circle $\overline{x}=\ell/2$ of the slices $(\pm2,0)$, next to the circle of zeros of the single-cell approximation at the height $a$, as shown in Fig.~\ref{fig:coherent}(d) and~(e). On this circle, $z_1=2+\ii\pi$ and $z_2=\ii\ell\overline{p}$, and the real function $F_{\ket{\psi}}(2+\ii\pi,\ii\ell\overline{p})$ changes sign twice as $\overline{p}$ runs over the unit cell, whereas it is constant in the single-cell approximation, as shown in Fig.~\ref{fig:coherent}(f). These zeros leave the explored range of heights as $\sigma/\ell$ decreases. For $\sigma=1$ and $\ell=2\sqrt\pi$, the single-cell approximation is poor, since the vacuum is not negligible at the edges $\overline{x}=\pm\ell/2$ of the unit cell, and the positions of the zeros depend on the choice of the unit cell: shifting the edges of the unit cell moves the zeros to other slices, although zeros are found for all the tested positions of the edges. Beyond the slices $(\pm2,0)$ shown in Fig.~\ref{fig:coherent}(d,e), zeros are also found on slices of smaller weight. The vacuum is a Gaussian state, whose Husimi distribution has no zero in the plane, but the vacuum is not a coherent state of the double discrete cylinder, and its Husimi distribution on this phase space can vanish. This property echoes the role of the vacuum as a resource in architectures based on GKP states~\cite{calcluth_vacuum_2023}.

\begin{figure*}[t]
\centering
\includegraphics[width=\textwidth]{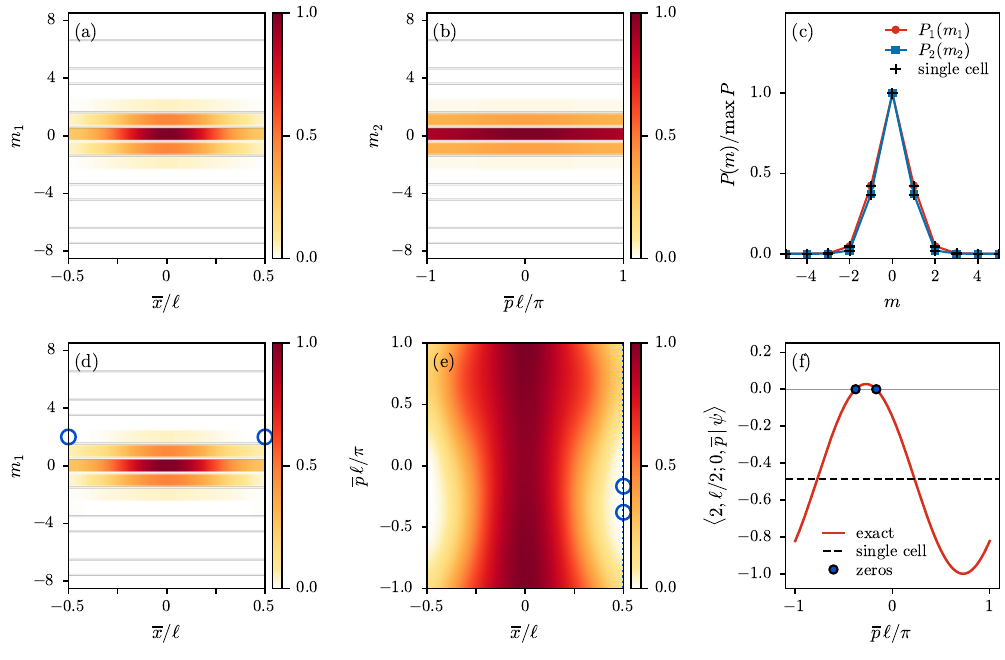}
\caption{Husimi distribution of the vacuum, $\sigma=1$, on the double discrete cylinder with $\ell=2\sqrt\pi$, that is $\sigma/\ell\simeq0.28$, computed from the exact coefficients of Eq.~\eqref{eq:psi_kq_coherent}. Each color map is normalized to its maximum. (a,b)~Marginal distributions on the position and momentum cylinders, $Q_1(m_1,\overline{x})$ and $Q_2(m_2,\overline{p})$, unrolled into planes. (c)~Weights of the heights, $P_1(m_1)$ and $P_2(m_2)$, normalized to their maximum: exact values as lines with symbols, and single-cell approximation as crosses. (d)~Cut $Q_{\ket{\psi}}(m_1,\overline{x};0,-0.379\pi/\ell)$ through a zero, marked by the open circles at $m_1=2$ on the edges $\overline{x}=\pm\ell/2$, which are identified. (e)~Slice $(m_1,m_2)=(2,0)$: the two zeros, marked by open circles, lie on the circle $\overline{x}=\ell/2$, at $\overline{p}\simeq-0.379\,\pi/\ell$ and $-0.166\,\pi/\ell$. (f)~Amplitude $\bra{2,\ell/2;0,\overline{p}}\ket{\psi}$ on this circle, which is real, normalized to its maximum modulus: the exact amplitude changes sign twice, which gives the two zeros of panel~(e), whereas the single-cell amplitude is constant. In the single-cell approximation, the circles of zeros lie at the non-integer heights $m_1=a\simeq2.07$, $3a\simeq6.21,\dots$, and the Husimi distribution has no zero on $\Gamma$.}
\label{fig:coherent}
\end{figure*}

\subsection{Ideal GKP states: no zero for the logical states, lattice zeros for their superpositions}
\label{sec:ideal_gkp}

\subsubsection{Ideal logical states as points of the torus}

The Gottesman--Kitaev--Preskill (GKP) code encodes a logical qubit into the infinite-dimensional Hilbert space of an oscillator, and corrects the small displacement errors that dominate in bosonic systems~\cite{gottesman_encoding_2001}. The ideal, infinitely squeezed logical states are combs of position eigenstates,
\begin{equation}
\label{eq:ideal_gkp}
  \ket{0}_{x}=\sum_{s\in\Z}\ket{x=s\ell-\ell/4},\qquad \ket{1}_{x}=\sum_{s\in\Z}\ket{x=s\ell+\ell/4},
\end{equation}
where $\ell$ is the period of each comb. With $\bra{p}\ket{x}=e^{-\ii px}/\sqrt{2\pi}$, the logical states are also combs in the momentum representation, with spacing $2\pi/\ell$, the state $\ket{1}_x$ carrying the phase pattern $(-1)^n$ with respect to $\ket{0}_x$:
\begin{equation}
  \ket{0}_{x}=\ket{+}_{p}=\frac{\sqrt{2\pi}}{\ell}\sum_{n\in\Z}\ii^{\,n}\ket{\tfrac{2\pi}{\ell}n}_{p},\qquad
  \ket{1}_{x}=\ket{-}_{p}=\frac{\sqrt{2\pi}}{\ell}\sum_{n\in\Z}(-1)^{n}\,\ii^{\,n}\ket{\tfrac{2\pi}{\ell}n}_{p}.
\end{equation}
The common phase $\ii^{\,n}$ results from the shift of both combs by $-\ell/4$, $\ket{\pm}_p=e^{\ii\ell\hat p/4}\ket{\pm}^{(0)}_p$, where $\ket{\pm}^{(0)}_p=\frac{\sqrt{2\pi}}{\ell}\sum_n(\pm1)^n\ket{\frac{2\pi}{\ell}n}_p$ are the combs associated with the unshifted positions $s\ell$ and $s\ell+\ell/2$. The square GKP code corresponds to $\ell=2\sqrt\pi$, for which the spacing $\ell/2$ between the peaks of $\ket{0}_x$ and $\ket{1}_x$ equals the spacing $2\pi/\ell$ of the momentum combs. In the modular representation, each comb reduces to a single point of the torus~\cite{fabre_wigner_2020},
\begin{equation}
\ket{0}_{x}=\ell^{-1/2}\ket{\overline{x}=-\ell/4,\overline{p}=0},\qquad \ket{1}_x=\ell^{-1/2}\ket{\overline{x}=\ell/4,\overline{p}=0},
\end{equation}
with the coefficients $\psi_{kq}=\ell^{-1/2}e^{-2\ii\pi kx_0/\ell}$, where $x_0=\mp\ell/4$ (Appendix~\ref{app:modular}). The two logical states sit at the middle of the two halves of the unit cell along the modular position, with the common modular momentum $\overline{p}=0$, as represented in Fig.~\ref{fig:idealGKP}.
\begin{figure}[t]
\centering
\includegraphics[width=\linewidth]{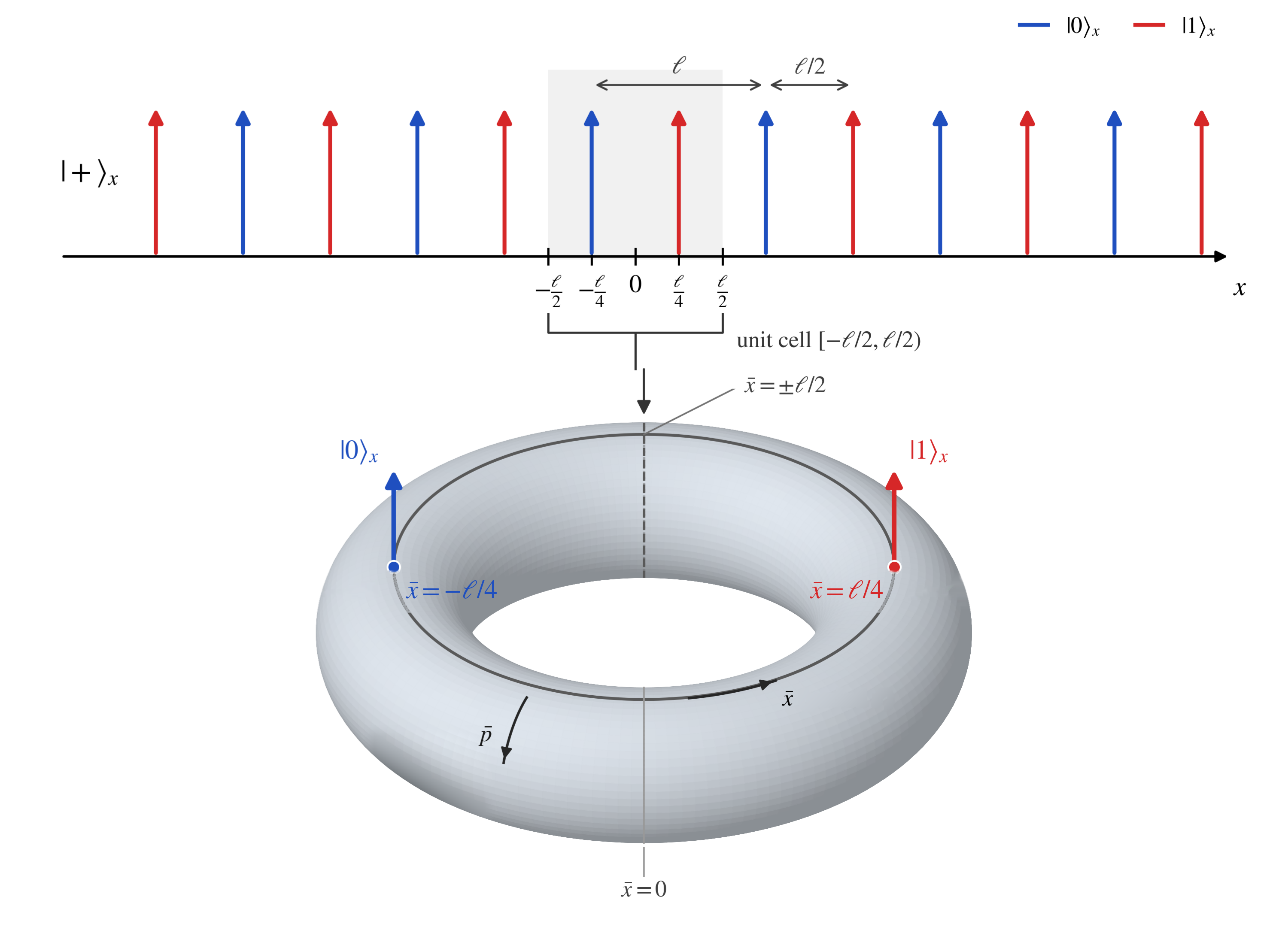}
\caption{Ideal GKP state represented as a comb of Dirac peaks, whose positions depend on the logical state. The comb of period $\ell$ is mapped onto a unit cell $[-\ell/2,\ell/2)$, replicated along the real line as a unit cell in crystallography, and the unit cell is represented as a circle. Combined with the circle of the reciprocal lattice, which carries the modular momentum $\overline{p}$, the unit cell forms the torus of the modular variables. The state shown is $\ket{+}_x$, the superposition of $\ket{0}_x$, located at $\overline{x}=-\ell/4$, and of $\ket{1}_x$, located at $\overline{x}=\ell/4$.}
\label{fig:idealGKP}
\end{figure}
The stellar functions of the ideal logical states factorize,
\begin{equation}
\label{eq:stellar_ideal}
\begin{gathered}
F_{\ket{0}_x}(z_1,z_2)=\frac{A^2}{\sqrt\ell}\,S_{1/2}\Big(z_1+\frac{\ii\pi}{2}\Big)S_{1/2}(z_2),\\
F_{\ket{1}_x}(z_1,z_2)=\frac{A^2}{\sqrt\ell}\,S_{1/2}\Big(z_1-\frac{\ii\pi}{2}\Big)S_{1/2}(z_2).
\end{gathered}
\end{equation}
By Lemma~\ref{lemma:zeros}, the zeros of $F_{\ket{0}_x}$ and $F_{\ket{1}_x}$ lie at half-integer heights: the ideal logical states have no zero on $\Gamma$. This absence of zeros is consistent with the Wigner distribution of the ideal logical states on the double cylinder, which is a positive point-like distribution~\cite{fabre_wigner_2020}, and extends to the double discrete cylinder the absence of zeros of phase eigenstates on the single discrete cylinder~\cite{fabre_majorana_2023}. The non-physical character of the ideal logical states is reflected by the spreading of $Q_1(m_1,\overline{x})$ and $Q_2(m_2,\overline{p})$ over all heights: the quasi-periodicity $S_{1/2}(\zeta+1)=e^{1/2+\zeta}S_{1/2}(\zeta)$ makes $e^{-m^2}|S_{1/2}(m+\ii\phi)|^2$ independent of $m$, so that $Q_1$ and $Q_2$ are identical on all heights, as shown in Fig.~\ref{fig:large_comb}(a,b).

\subsubsection{Superpositions of ideal logical states}

For the superpositions $\ket{\pm}_x=(\ket{0}_x\pm\ket{1}_x)/\sqrt2$, the stellar functions read
\begin{equation}
\label{eq:stellar_ideal_pm}
F_{\ket{\pm}_x}(z_1,z_2)=\frac{A^2}{\sqrt{2\ell}}\Big[S_{1/2}\Big(z_1+\frac{\ii\pi}{2}\Big)\pm S_{1/2}\Big(z_1-\frac{\ii\pi}{2}\Big)\Big]S_{1/2}(z_2),
\end{equation}
and only even values of $k$, for $\ket{+}_x$, or odd values of $k$, for $\ket{-}_x$, survive in the sum over $k$. As a result, the zeros of $F_{\ket{\pm}_x}$ lie exactly on $\Gamma$, at $\overline{x}\in\{0,\ell/2\}$, for every odd height $m_1$ for $\ket{+}_x$ and every even height $m_1$ for $\ket{-}_x$, and for all $(m_2,\overline{p})$ (Appendix~\ref{app:lemma}). These infinitely many zeros correspond to the negativity of the Wigner distribution of the superposition states~\cite{fabre_wigner_2020}.

\subsection{Physical GKP logical states}
\label{sec:physical}

\subsubsection{Comb of Gaussian peaks under a Gaussian envelope}

The ideal GKP states are non-physical, since their preparation would require an infinite energy. Finitely squeezed GKP states are described by a comb of Gaussian peaks of width $\sigma$ under a Gaussian envelope of width $\kappa$,
\begin{equation}
\label{eq:physical_gkp}
\psi_{x_0}(x)=N\,e^{-x^2/2\kappa^2}\sum_{s\in\Z}e^{-(x-s\ell-x_0)^2/2\sigma^2},
\end{equation}
with $x_0=-\ell/4$ for $\ket{0}_x$ and $x_0=+\ell/4$ for $\ket{1}_x$, and a normalization constant $N$. The stellar function and the Husimi distribution of $\ket{\psi_{x_0}}$ are denoted $F_{\ket{\psi_{x_0}}}(z_1,z_2)$ and $Q_{\ket{\psi_{x_0}}}(m_1,\overline{x};m_2,\overline{p})$. The modular wavefunction follows from the Zak transform of Eq.~\eqref{eq:modular_wf},
\begin{equation}
\label{eq:zak_gkp}
\psi_{x_0}(\overline{x},\overline{p})=\sqrt\ell\,N\sum_{n,s\in\Z}e^{-(\overline{x}+n\ell)^2/2\kappa^2}\,e^{-(\overline{x}+n\ell-s\ell-x_0)^2/2\sigma^2}\,e^{-\ii n\ell\overline{p}}.
\end{equation}
The envelope depends both on the modular position $\overline{x}$ and on the cell index $n$, which is conjugate to $\overline{p}$: the envelope couples the two modular variables, and the modular wavefunction $\psi_{x_0}(\overline{x},\overline{p})$ is not separable in general. The analysis below is restricted to the large-comb approximation, in which this coupling is neglected.

\subsubsection{Large-comb approximation: separable modular wavefunction}
\label{sec:large_comb}

When the envelope contains many peaks, $\kappa\gg\ell$, the envelope varies slowly over one cell and can be evaluated at the cell index, $e^{-(\overline{x}+n\ell)^2/2\kappa^2}\simeq e^{-n^2\ell^2/2\kappa^2}$. The modular wavefunction then factorizes into two wrapped Gaussians (Appendix~\ref{app:gkp_coefficients}),
\begin{equation}
\label{eq:large_comb}
\begin{gathered}
\psi_{x_0}(\overline{x},\overline{p})\simeq N'\,g_\sigma(\overline{x}-x_0)\,g_{1/\kappa}(\overline{p}),\\
g_\sigma(\overline{x})=\sum_{s\in\Z}e^{-(\overline{x}+s\ell)^2/2\sigma^2},\\
g_{1/\kappa}(\overline{p})=\sum_{j\in\Z}e^{-\kappa^2(\overline{p}+2\pi j/\ell)^2/2},
\end{gathered}
\end{equation}
where $N'$ is a normalization constant. For $\sigma\ll\ell$ and $1/\kappa\ll\pi/\ell$, a single term of each sum contributes in the unit cell, and the modular wavefunction reduces to a product of two Gaussians~\cite{fabre_wigner_2020},
\begin{equation}
\label{eq:physGKP}
\psi_{x_0}(\overline{x},\overline{p})\simeq N''\,G_{\sigma}(\overline{x}-x_0)\,G_{1/\kappa}(\overline{p}),\qquad G_s(y)=e^{-y^2/2s^2},
\end{equation}
with $\overline{x}$ and $\overline{p}$ restricted to the unit cell. The large-comb form is periodic in $\overline{p}$, as required, but periodic rather than quasi-periodic in $\overline{x}$. The large-comb form therefore violates the first condition of Eq.~\eqref{eq:quasi_periodicity} by terms of the order of the modular wavefunction at the edges of the unit cell, of relative order $e^{-(\ell/2-|x_0|)^2/2\sigma^2}=e^{-\ell^2/32\sigma^2}$, which are also the terms neglected below. Placing the edges of the unit cell at $\overline{x}=\pm\ell/2$, at equal distance from the peaks of both logical states, minimizes these terms.

The coefficients of the product of Gaussians of Eq.~\eqref{eq:physGKP} contain error functions, which account for the truncation of each Gaussian at the edges of the unit cell (Appendix~\ref{app:gkp_coefficients}):
\begin{equation}
\label{eq:psi_kq_erf}
\begin{split}
\psi_{kq}&=\mathcal{N}\,c_k\,d_q,\\
c_k&=e^{-\frac{2\pi^2\sigma^2}{\ell^2}k^2}e^{-2\ii\pi kx_0/\ell}\;\frac12\Big[\mathrm{erf}\Big(\frac{\ell/2-x_0}{\sqrt2\,\sigma}+\ii\frac{\sqrt2\pi\sigma k}{\ell}\Big)\\
&\qquad+\mathrm{erf}\Big(\frac{\ell/2+x_0}{\sqrt2\,\sigma}-\ii\frac{\sqrt2\pi\sigma k}{\ell}\Big)\Big],\\
d_q&=e^{-\frac{\ell^2q^2}{2\kappa^2}}\;\frac12\Big[\mathrm{erf}\Big(\frac{\pi\kappa}{\sqrt2\,\ell}+\ii\frac{\ell q}{\sqrt2\,\kappa}\Big)+\mathrm{erf}\Big(\frac{\pi\kappa}{\sqrt2\,\ell}-\ii\frac{\ell q}{\sqrt2\,\kappa}\Big)\Big].
\end{split}
\end{equation}
For the wrapped Gaussians of Eq.~\eqref{eq:large_comb}, the parts cut at the edges of the unit cell are exactly compensated by the neighboring terms of the sums, and the half-sums of error functions are replaced by one (Appendix~\ref{app:gkp_coefficients}):
\begin{equation}
\label{eq:psi_kq_lc}
\psi_{kq}=\mathcal{N}\,e^{-\frac{2\pi^2\sigma^2}{\ell^2}k^2}\,e^{-2\ii\pi kx_0/\ell}\,e^{-\frac{\ell^2q^2}{2\kappa^2}},\qquad
\mathcal{N}^{-2}=\sum_{k\in\Z}e^{-\frac{4\pi^2\sigma^2}{\ell^2}k^2}\sum_{q\in\Z}e^{-\frac{\ell^2q^2}{\kappa^2}}\simeq\frac{\kappa}{2\sigma}.
\end{equation}
The two sets of coefficients differ by terms bounded by $e^{-\ell^2/32\sigma^2}$ and $e^{-\pi^2\kappa^2/2\ell^2}$ (Appendix~\ref{app:gkp_coefficients}). These terms decay only as $1/|k|$ and $1/|q|$, as for the jump discussed in Sec.~\ref{sec:config}, whereas the coefficients of Eq.~\eqref{eq:psi_kq_lc} decay as Gaussians. The zeros are therefore shifted by exponentially small amounts at low heights, whereas the terms neglected in Eq.~\eqref{eq:psi_kq_lc} dominate the stellar function at heights $|m_1|$ of the order of $\ell^2/(8\pi\sigma^2)$ and beyond. Eq.~\eqref{eq:psi_kq_lc} is used below.

\subsubsection{Factorized stellar function and extent along the heights}

In the large-comb approximation, the stellar function factorizes (Appendix~\ref{app:gkp_coefficients}),
\begin{equation}
\label{eq:stellar_lc}
F_{\ket{\psi_{x_0}}}(z_1,z_2)=A^2\mathcal{N}\;S_a\Big(z_1-\frac{2\ii\pi}{\ell}x_0\Big)\,S_b(z_2),\qquad a=\frac12+\frac{2\pi^2\sigma^2}{\ell^2},\qquad b=\frac12+\frac{\ell^2}{2\kappa^2},
\end{equation}
and reduces, up to normalization, to the ideal case of Eq.~\eqref{eq:stellar_ideal} for $\sigma\to0$ and $\kappa\to\infty$. Contrary to the ideal states, the physical states occupy a finite range of heights. With Eq.~\eqref{eq:psi_kq_lc}, the weights of Eq.~\eqref{eq:height_weights} become (Appendix~\ref{app:gkp_coefficients})
\begin{equation}
\label{eq:extent}
\begin{gathered}
P_1(m_1)\simeq P_1(0)\,\exp\Big(-\frac{c\,m_1^2}{1+c}\Big),\\
P_2(m_2)\simeq P_2(0)\,\exp\Big(-\frac{d\,m_2^2}{1+d}\Big),\\
c=\frac{4\pi^2\sigma^2}{\ell^2},\qquad d=\frac{\ell^2}{\kappa^2},
\end{gathered}
\end{equation}
so that $Q_{\ket{\psi_{x_0}}}$ extends over $|m_1|\lesssim\sqrt{(1+c)/c}\simeq\ell/(2\pi\sigma)$ on the position cylinder and over $|m_2|\lesssim\sqrt{(1+d)/d}\simeq\kappa/\ell$ on the momentum cylinder. Narrow peaks spread the state over many heights of the position cylinder, and a broad envelope over many heights of the momentum cylinder. The ideal limit $\sigma\to0$, $\kappa\to\infty$ recovers the uniform spreading over all heights, as illustrated in Fig.~\ref{fig:large_comb}(c), which compares the weights of the heights of the physical and ideal states. The widths of the peaks and of the envelope thus become visible in the Husimi distribution, as they shape the Wigner distribution on the double cylinder~\cite{fabre_wigner_2020}.

\subsubsection{Circles of zeros and fine-tuned widths}

Lemma~\ref{lemma:zeros} applied to Eq.~\eqref{eq:stellar_lc} locates all the zeros of $F_{\ket{\psi_{x_0}}}(z_1,z_2)$ (Appendix~\ref{app:gkp_coefficients}). On the position cylinder, the zeros form the circles
\begin{equation}
\label{eq:zeros_lc_position}
m_1=(2n+1)\Big(\frac12+\frac{2\pi^2\sigma^2}{\ell^2}\Big),\qquad \overline{x}=-x_0\ (\mathrm{mod}\ \ell),\qquad\text{for all }(m_2,\overline{p}),
\end{equation}
and on the momentum cylinder the circles
\begin{equation}
\label{eq:zeros_lc_momentum}
m_2=(2n'+1)\Big(\frac12+\frac{\ell^2}{2\kappa^2}\Big),\qquad \overline{p}=\pi/\ell\ (\mathrm{mod}\ 2\pi/\ell),\qquad\text{for all }(m_1,\overline{x}),
\end{equation}
with $n,n'\in\Z$. On each cylinder, the zeros sit opposite to the peak of the state, which lies at $\overline{x}=x_0$ and $\overline{p}=0$. The heights of Eqs.~\eqref{eq:zeros_lc_position} and~\eqref{eq:zeros_lc_momentum} are shifted from the half-integers $n+\frac12$ and $n'+\frac12$ by $(2n+1)\,2\pi^2\sigma^2/\ell^2$ and $(2n'+1)\,\ell^2/2\kappa^2$, respectively. For generic widths, none of these heights is an integer, and the Husimi distribution of a logical state has no zero on $\Gamma$ in the large-comb approximation.

Zeros on $\Gamma$ appear for fine-tuned widths. The height $(2n+1)a$ of Eq.~\eqref{eq:zeros_lc_position} equals an integer $\mu_1$ for the odd integer $\nu_1=2n+1$ when $a=\mu_1/\nu_1$, and the height $(2n'+1)b$ of Eq.~\eqref{eq:zeros_lc_momentum} equals an integer $\mu_2$ for the odd integer $\nu_2=2n'+1$ when $b=\mu_2/\nu_2$, that is, when
\begin{equation}
\label{eq:tuned_widths}
\sigma=\frac{\ell}{2\pi}\sqrt{\frac{2\mu_1-\nu_1}{\nu_1}},\qquad \kappa=\ell\sqrt{\frac{\nu_2}{2\mu_2-\nu_2}},
\end{equation}
with $2\mu_1>\nu_1$ and $2\mu_2>\nu_2$. Writing $a=u_1/v_1$ and $b=u_2/v_2$ in lowest terms, the Husimi distribution $Q_{\ket{\psi_{x_0}}}$ then vanishes on the circle $\overline{x}=-x_0$ of every slice with height $m_1=(2j+1)\,u_1$, $j\in\Z$, and on the circle $\overline{p}=\pi/\ell$ of every slice with height $m_2=(2j'+1)\,u_2$, $j'\in\Z$ (Appendix~\ref{app:gkp_coefficients}). For instance, the widths $\sigma=\ell/(2\sqrt3\,\pi)$ and $\kappa=3\ell$ give $a=2/3$ and $b=5/9$, and bring circles of zeros of $Q_{\ket{0}_x}$ onto $\Gamma$, as shown in Fig.~\ref{fig:large_comb}(d,e). The cuts of Fig.~\ref{fig:large_comb}(f) show that the zero is reached at the tuned height $m_1=2$, and not at the neighboring heights $m_1=1$ and $m_1=3$.

\begin{figure*}[t]
\centering
\includegraphics[width=\textwidth]{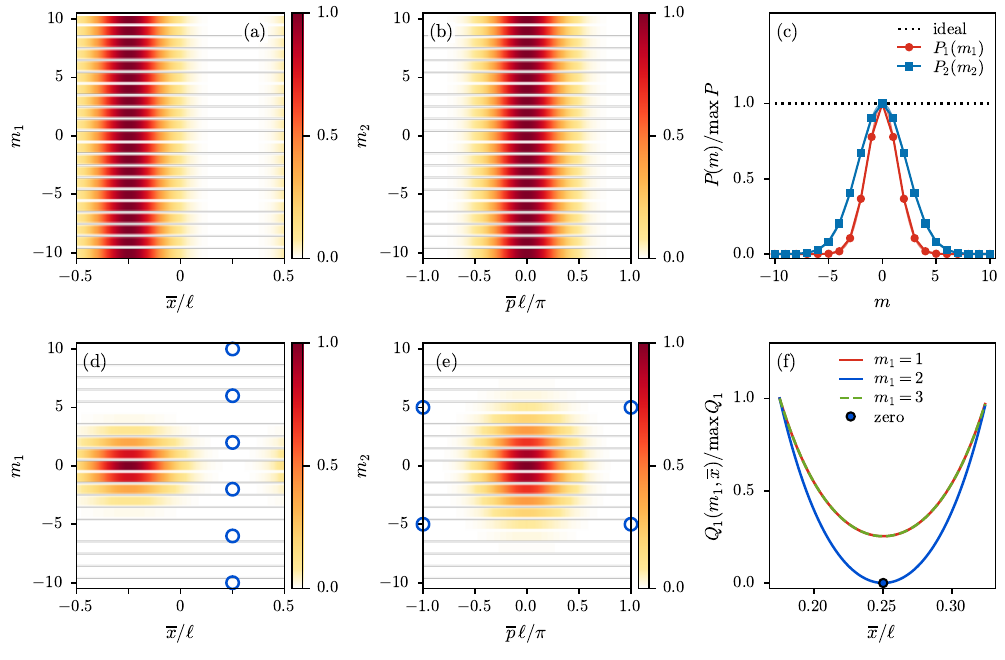}
\caption{Husimi distribution of the $\ket{0}_x$ GKP state on the two cylinders, unrolled into planes of the modular variable and of the height. (a,b)~Ideal state: $Q_1(m_1,\overline{x})$ and $Q_2(m_2,\overline{p})$ are identical on all heights. (d,e)~Physical state in the large-comb approximation, with the tuned widths $\sigma=\ell/(2\sqrt3\,\pi)\simeq0.092\,\ell$ and $\kappa=3\ell$, for which $a=2/3$ and $b=5/9$: $Q_1(m_1,\overline{x})$ and $Q_2(m_2,\overline{p})$ are confined to a finite range of heights, and vanish on the circles $\overline{x}=\ell/4$ at $m_1=\pm2,\pm6,\pm10$ and $\overline{p}=\pi/\ell$ at $m_2=\pm5$ (open circles; the edges $\overline{p}=\pm\pi/\ell$ are identified). Each plane is normalized to its maximum value. (c)~Weights of the heights, $P_1(m_1)$ and $P_2(m_2)$, normalized to their maximum, for the ideal (dotted) and physical (symbols) states; thin lines: Eq.~\eqref{eq:extent}. (f)~Cuts of $Q_1(m_1,\overline{x})$ around $\overline{x}=\ell/4$ for $m_1=1,2,3$, normalized to their maximum in the window: only the cut at $m_1=2$ reaches zero. At the zeros, $|F_{\ket{0}_x}(z_1,z_2)|$ is below $10^{-16}$ times $\sum_{k,q\in\Z}|\psi_{kq}|e^{-(k^2+q^2)/2}e^{km_1+qm_2}$, compared with $3\times10^{-2}$ at $(m_1,\overline{x})=(3,\ell/4)$.}
\label{fig:large_comb}
\end{figure*}

\subsection{Superpositions of physical GKP logical states}
\label{sec:superpositions}

\subsubsection{Overlap between finitely squeezed logical states}

The finitely squeezed logical states are not orthogonal. In the large-comb approximation, Eq.~\eqref{eq:psi_kq_lc} gives (Appendix~\ref{app:overlap}, where the normalization constant $N$ of Eq.~\eqref{eq:physical_gkp} is also given)
\begin{equation}
\label{eq:overlap_lc}
{}_x\!\bra{0}\ket{1}_x=\frac{\sum_k(-1)^ke^{-ck^2}}{\sum_ke^{-ck^2}}=\frac{S_c(\ii\pi)}{S_c(0)}\simeq2\,e^{-\ell^2/16\sigma^2},\qquad c=\frac{4\pi^2\sigma^2}{\ell^2}.
\end{equation}
The exact overlap, obtained from Eq.~\eqref{eq:physical_gkp} without the large-comb approximation, is derived in Appendix~\ref{app:overlap} and reduces to the same leading term for $\sigma\ll\ell\ll\kappa$. Both overlaps are real and positive. The normalized superpositions $\ket{\pm}_x\propto\ket{0}_x\pm\ket{1}_x$ have the stellar functions
\begin{equation}
\label{eq:stellar_pm}
F_{\ket{\pm}_x}(z_1,z_2)=\frac{F_{\ket{\psi_{-\ell/4}}}(z_1,z_2)\pm F_{\ket{\psi_{+\ell/4}}}(z_1,z_2)}{\sqrt{2\left(1\pm{}_x\!\bra{0}\ket{1}_x\right)}} .
\end{equation}

\subsubsection{Drift of the lattice zeros in the large-comb approximation}

In the large-comb approximation, the momentum factor $S_b(z_2)$ of Eq.~\eqref{eq:stellar_lc} is the same for both logical states, and
\begin{equation}
\label{eq:stellar_pm_lc}
F_{\ket{\pm}_x}(z_1,z_2)=\frac{A^2\mathcal{N}}{\sqrt{2\left(1\pm{}_x\!\bra{0}\ket{1}_x\right)}}\Big[S_a\Big(z_1+\frac{\ii\pi}{2}\Big)\pm S_a\Big(z_1-\frac{\ii\pi}{2}\Big)\Big]S_b(z_2).
\end{equation}
The momentum factor $S_b(z_2)$ gives the circles of Eq.~\eqref{eq:zeros_lc_momentum}, common to both superpositions. The position factor retains even or odd values of $k$ only, and its zeros form the circles $\overline{x}\in\{0,\ell/2\}$, for all $(m_2,\overline{p})$, at the heights (Appendix~\ref{app:lemma})
\begin{equation}
\label{eq:zeros_pm_lc}
m_1=(2n+1)\Big(1+\frac{4\pi^2\sigma^2}{\ell^2}\Big)\ \text{for}\ \ket{+}_x,\qquad m_1=2n\Big(1+\frac{4\pi^2\sigma^2}{\ell^2}\Big)\ \text{for}\ \ket{-}_x .
\end{equation}
The lattice zeros of the ideal superposition states thus drift away from $\Gamma$ as the peaks acquire a finite width, by an amount that grows with the height. The only exception is the family $n=0$ of $\ket{-}_x$: the Husimi distribution $Q_{\ket{-}_x}$ vanishes on the circles $\overline{x}=0$ and $\overline{x}=\ell/2$ of all the slices $(0,m_2)$, for all widths.

Part of this family survives beyond the large-comb approximation, as a consequence of parity. The parity $x\to-x$ exchanges the two logical states of Eq.~\eqref{eq:physical_gkp}, $\psi_{x_0}(-x)=\psi_{-x_0}(x)$, so that the state $\ket{-}_x$ is odd. The change of variable $\overline{x}\to-\overline{x}$ in Eq.~\eqref{eq:psi_kq} then gives $\psi_{-k,-q}=-\psi_{kq}$, and Eq.~\eqref{eq:stellar} gives $F_{\ket{-}_x}(-z_1,-z_2)=-F_{\ket{-}_x}(z_1,z_2)$, without any approximation. Since $F_{\ket{-}_x}$ is $2\ii\pi$-periodic in each variable, it vanishes at the four points where $(-z_1,-z_2)=(z_1,z_2)$ modulo $2\ii\pi$, that is, at $z_1\in\{0,\ii\pi\}$ and $z_2\in\{0,\ii\pi\}$, or $(\overline{x},\overline{p})\in\{0,\ell/2\}\times\{0,\pi/\ell\}$ on the slice $(0,0)$. These four zeros on $\Gamma$ are protected by parity for all widths, including beyond the large-comb approximation.

%%%%%%%%%%%%%%%%%%%%%%%%%%%%%%%%%%%%%%%%%%

%%%%%%%%%%%%%%%%%%%%%%%%%%%%%%%%%%%%%%%%%%

\section{Conclusion}
\label{sec:conclusion}
\noindent Covariant integral quantization has been developed on the double discrete cylinder, the product of a two-dimensional grid and of a torus. A weight function on this phase space defines a quantization operator, whose integral kernel follows from a partial Fourier transform of the weight, and the integral kernel of the operator associated with a classical function follows from the symplectic Fourier transform of this function, together with semi-classical portraits; the unit weight yields the parity operator and the Wigner distribution, and coherent states yield the Husimi distribution. Modular variables, defined by the decomposition of the position and momentum modulo a period and its reciprocal period, make the torus the configuration space of continuous variables, through the Zak transform and the integer basis. Since the integer basis is a periodic expansion of a quasi-periodic modular wavefunction, the phase-space distributions depend on the choice of the unit cell in the modular position. The wrapped coherent states of the circle coincide with the coherent states generated by a periodized Gaussian fiducial vector, whose widths equal the radii of the torus, and define a stellar function whose zeros on the lattice of heights give a stellar rank on the double discrete cylinder.

Position eigenstates and ideal GKP logical states have no zero on the lattice of heights, their zeros lying at half-integer heights, whereas the superpositions of ideal logical states have infinitely many zeros on the lattice, in correspondence with the positivity and the negativity of their Wigner distributions on the double cylinder. The vacuum, which is Gaussian in the plane, has zeros on the lattice, located beyond the single-cell approximation at positions that depend on the unit cell, so that zeros on the double discrete cylinder do not coincide with non-Gaussianity in the plane. For finitely squeezed GKP states, the stellar function and the overlap between logical states have been obtained in closed form in the large-comb approximation. When the envelope contains many peaks, the modular wavefunction factorizes, the Husimi distribution extends over a number of heights inversely proportional to the width of the peaks on the position cylinder and proportional to the width of the envelope on the momentum cylinder, and the zeros form whole circles located analytically, which reach the lattice only for fine-tuned widths. The lattice zeros of the superpositions drift away from the lattice with the width of the peaks, except for four zeros of the minus superposition, which are protected by parity for all widths. Zeros on the lattice thus emerge from the ideal GKP states once the vacuum and finite squeezing enter the description, in line with the role of the vacuum in GKP-based architectures that are hard to simulate.

Several extensions follow from this work. Beyond the large-comb approximation, the envelope pulls each peak of a physical GKP state towards the origin by an amount proportional to its distance from the origin, so that the modular position of a peak depends on the cell index, which is conjugate to the modular momentum. The envelope thus couples the two modular variables, and the stellar function no longer factorizes. Its zeros are then expected to form isolated points on the slices of the phase space, which would lie on the lattice of heights without any fine tuning of the widths and persist under small changes of the state; locating these zeros and counting them as a function of the squeezing would quantify how the resources of physical GKP states emerge from the ideal ones. Relating this stellar rank to the negativity of the Wigner distribution on the same phase space, and to the cost of classically simulating GKP-based circuits~\cite{calcluth_efficient_2022,calcluth_vacuum_2023}, is a natural continuation. Making the construction independent of the choice of the unit cell is an open question, which may require coherent states adapted to the quasi-periodicity of the modular wavefunction. Other weight functions of the covariant integral quantization provide further families of phase-space distributions, whose zeros and negativities can be compared with those of the Husimi and Wigner distributions. The formalism also applies to GKP codes on non-rectangular lattices, such as the hexagonal lattice~\cite{gottesman_encoding_2001,albert_performance_2018,conrad_gottesman-kitaev-preskill_2022}, for which the modular wavefunction is not separable even for ideal states, to rotor codes~\cite{vuillot_homological_2024}, to Josephson circuits, and to twisted photons, where the separability of the phase-space distribution into two discrete cylinders could serve as a witness of the entanglement between orbital angular momentum modes.

\begin{acknowledgments}
N. Fabre acknowledges useful discussions with Jack Davis, Ulysse Chabaud and Pérola Milman. The artificial intelligence assistant \textit{Claude Opus 5.5} was used exclusively for checking spelling, typos and grammar. All scientific content was verified by the authors, who take full responsibility for the manuscript.
\end{acknowledgments}

%%%%%%%%%%%%%%%%%%%%%%%%%%%%%%%%%%%%%%%%%%
\appendix

\section{Weyl--Heisenberg group and coherent states on the torus}
%\subsection{Overview:  Weyl Operator Acting on $\mathcal{H}$}
\label{weyl}
The Weyl--Heisenberg group $\mathrm{H}^{\mathrm{sd}}_2$ for the torus $\T^2_{\pmb \ell}$ is the set of 5-tuples $(s, \vecm_{\pmb \ell},\vectheta)\in \R\times\Z^2_{\pmb\ell} \times \T^2_{\pmb\ell}$ equipped with the internal law
\begin{equation}
\begin{split}
&(s,\vecm_{\pmb \ell},\vectheta)\,(s^{\prime}, \vecm_{\pmb \ell}^{\prime},\vectheta^{\prime})\\&= \left(s+s^{\prime} +\frac{\vecm_{\pmb \ell} \cdot \vectheta^{\prime}-\vecm_{\pmb \ell}^{\prime}\cdot\vectheta}{2},\vecm_{\pmb \ell}+\vecm_{\pmb \ell}^{\prime}\ ,\vectheta+\vectheta^{\prime} \, \mathrm{mod}\,\bm\ell\right)\,,
\end{split}
\end{equation}
where the translations are taken modulo the periods $\bm\ell$ of Eq.~\eqref{ellvec}. The neutral element is $e=(0,\bm 0,\bm 0)$ and the inverse is given by $$(s,\vecm_{\pmb \ell},\vectheta)^{-1}= (-s,-\vecm_{\pmb \ell} \ ,-\vectheta\,  \mathrm{mod}\,\bm\ell)\,. $$
The representation $\hat V$ of $\mathrm{H}^{\mathrm{sd}}_2$ on $\mathcal{H}$ is defined by
\begin{equation}
\label{repV}
\mathcal{H}\ni \psi(\bg)\mapsto ({\hat V}(s,\vecm_{\pmb \ell},\vectheta)\psi)(\pmb{\gamma})=e^{\ii\,s}e^{-\ii\,\frac{\vecm_{\pmb \ell}\cdot\vectheta}{2}}e^{\ii\,\vecm_{\pmb\ell}\cdot\pmb{\gamma}}\psi(\pmb{\gamma}-\vectheta)\,.
\end{equation}
This representation is unitary and irreducible.
Due to the central parameter $s\in \R$, it is not square integrable on the group equipped with its semi-discrete Haar measure $\ud(s,\vecm_{\pmb \ell},\vectheta)$, that is, there is no $\psi\in L^2(\T^2_{\pmb \ell},\ud^2\pmb{\gamma})$ such that
\begin{equation}
    \sum_{\vecm_{\pmb \ell} \in \Z^2_{\pmb\ell}}\int_{\R}\int_{\T^2_{\pmb \ell}}\ud s\,\ud^2 \vectheta\,\vert(\psi,{\hat V}(s,\vecm_{\pmb \ell},\vectheta)\psi)\vert^2 < \infty\,.
\end{equation} 
However, it is square integrable on the coset space $\R\backslash\mathrm{H}^{\mathrm{sd}}_2$ of $\mathrm{H}^{\mathrm{sd}}_2$ by its center. This coset space is identified with the semi-discrete double cylinder viewed as a phase space, $\mathcal{C}^2_\mathrm{sd}=\Z^2_{\pmb\ell}\times\T^2_{\pmb \ell}$. The
representation of Eq.~\eqref{repV} induces the unitary (Weyl-like) operator $\hat U$ defined on $\mathcal{C}^2_\mathrm{sd}$ by the section
$s = 0$ on the group, that is, $\hat U(\vecm_{\pmb \ell}, \vectheta) := \hat V(0,\vecm_{\pmb \ell}, \vectheta)$:
\begin{equation}
    \label{UIR}
({\hat U}(\vecm_{\pmb \ell},\vectheta)\psi)(\pmb{\gamma})=e^{-\ii\,\frac{\vecm_{\pmb \ell}\cdot\vectheta}{2}}\,e^{\ii\,\vecm_{\pmb \ell}\cdot\pmb{\gamma}}\psi(\pmb{\gamma}-\vectheta)\,.
\end{equation}
This unitary operator $\hat U(\vecm_{\pmb \ell},\vectheta)$ combines the translation and modulation operators $\hat T_{\vectheta}$ and $\hat E_{\vecm_{\pmb \ell}}$ with a phase that depends on $\vecm_{\pmb \ell}$ and $\vectheta$. The multiplication rules of the group $\mathrm{H}^{\mathrm{sd}}_2$ show that the adjoint of the operator $\hat U(\vecm_{\pmb \ell},\vectheta)$ is
\begin{align}
    \hat U^{\dagger}(\vecm_{\pmb \ell},\vectheta) = \hat U(-\vecm_{\pmb \ell},-\vectheta).
\end{align}
\noindent Direct algebraic manipulations show that, for any unit-norm $\psi$ in $\mathcal{H}$, the family
\begin{equation}
\{{\hat U}(\vecm_{\pmb \ell},\vectheta)\psi\equiv \psi_{(\vecm_{\pmb \ell},\vectheta)}\, , \, (\vecm_{\pmb \ell},\vectheta)\in\Gamma\}
\end{equation}
is an overcomplete family, that is, it resolves the identity on $\mathcal{H}$ in the sense of
\begin{equation}
\label{resUnity}
\frac{1}{\ell_1\ell_2}\sum_{\vecm_{\pmb \ell}\in\Z^2_{\pmb\ell}}\int_{\T^2_{\pmb \ell}} \ud^2\vectheta\, |\psi_{(\vecm_{\pmb \ell},\vectheta)}\rangle\langle \psi_{(\vecm_{\pmb \ell},\vectheta)}| = \bu\,, \quad \Vert \psi \Vert = 1\, .
\end{equation}
\noindent This fundamental property extends to the $\hat U$-transported versions of a density operator $\hat\rho = \sum_i p_i |\psi_i\rangle\langle\psi_i|$ (a mixed quantum state) on $\mathcal{H}$, where $\mathcal{J}\ni i\mapsto p_i$ is a discrete probability distribution on an index set $\mathcal{J}$ (a continuous distribution can be treated in the same way), and all the states $|\psi_i\rangle\in \mathcal{H}$ are unit norm:
\begin{equation}
\label{densop}
\frac{1}{\ell_1\ell_2}\sum_{\vecm_{\pmb \ell}\in\Z^2_{\pmb\ell}}\,\int_{\T^2_{\pmb \ell}} \ud^2\vectheta \ \hat\rho(\vecm_{\pmb \ell},\vectheta)= \bu\, , \quad  \hat\rho(\vecm_{\pmb \ell},\vectheta):={\hat U}(\vecm_{\pmb \ell},\vectheta)\hat\rho {\hat U}^{\dag}(\vecm_{\pmb \ell},\vectheta) \,. 
\end{equation}
This implies that, for any $\phi$ and $\psi$ in $\mathcal{H}$,
\begin{equation}
\label{resunZSpp}
\frac{1}{\ell_1\ell_2}\sum_{\vecm_{\pmb \ell}\in\Z^2_{\pmb\ell}}\,\int_{\T^2_{\pmb \ell}} \ud^2\vectheta\,\lg\phi | \hat\rho(\vecm_{\pmb \ell},\vectheta)|\psi\rangle= \lg\phi|\psi\rg\,.
\end{equation}
\noindent In Eq.~\eqref{resUnity}, the state $\psi_{(\vecm_{\pmb \ell},\vectheta)}$ is a coherent state for the group $\mathrm{H}^{\mathrm{sd}}_2$ in the sense of Perelomov, and the unit-norm state $\psi$ is called a fiducial vector.
The projection of $\phi\in \mathcal{H}$ on $\psi_{\vecm_{\pmb \ell},\vectheta}$, namely
\begin{equation}
\label{pspphi}
\begin{split}
\langle\,\psi_{\vecm_{\pmb \ell},\vectheta}|\phi\rangle
&= \int_{\T^2_{\pmb \ell}} \ud^2\pmb{\gamma}\, \overline{(\hat U(\vecm_{\pmb \ell},\vectheta)\psi)(\pmb{\gamma})}\,\phi(\pmb{\gamma})\\
&=  \int_{\T^2_{\pmb \ell}} \ud^2\pmb{\gamma}\,e^{\ii\,\frac{\vecm_{\pmb \ell}\cdot\vectheta}{2}}e^{-\ii\,\vecm_{\pmb \ell}\cdot\pmb{\gamma}}\,\overline{\psi(\pmb{\gamma}-\vectheta)}\,\phi(\pmb{\gamma})\, ,
\end{split}
\end{equation}
is the phase-space representation of $\phi$ with respect to the family of coherent states $\psi_{(\vecm_{\pmb \ell},\vectheta)}$, and $\vert\langle\,\psi_{\vecm_{\pmb \ell},\vectheta}|\phi\rangle\vert^2$ is the corresponding Husimi function. More generally,
\begin{align}
  \mathrm{Tr}\left(|\phi\rg\lg \phi| \hat\rho(\vecm_{\pmb \ell},\vectheta)\right)= \lg\phi|\hat\rho(\vecm_{\pmb \ell},\vectheta)|\phi\rg  
\end{align}
is the Husimi function of $\phi$ in the phase-space representation provided by the family $\{\hat\rho(\vecm_{\pmb \ell},\vectheta):(\vecm_{\pmb \ell},\vectheta)\in \Gamma\}$.
The present study is restricted to coherent states; the extension of the results to transported density operators is straightforward.

\subsection{Properties of the coherent-state representation on $\mathcal{H}$}

The following relations are used in the analysis of the Weyl operator $\hat U(\vecm_{\pmb \ell},\vectheta)$:
%\begin{align*}
%    \hat U^{\dagger}(\vecm_{\pmb \ell},\vectheta)\psi(\pmb{\gamma}) &= e^{-\ii\frac{\vecm_{\pmb \ell}\cdot\vectheta}{2}}e^{-\ii\,\vecm_{\pmb \ell}\cdot\pmb{\gamma}}\psi(\pmb{\gamma}+\vectheta)\\
%   &= \hat U(-\vecm_{\pmb \ell},-\vectheta).
%\end{align*}
\begin{align*}
    \hat U(\vecm_{\pmb \ell},\vectheta)\hat U(\vecm_{\pmb \ell}^{\prime},\vectheta^{\prime}) &= e^{\ii\frac{(\vecm_{\pmb \ell}\cdot\vectheta^{\prime} - \vecm_{\pmb \ell}^{\prime}\cdot\vectheta)}{2}}\hat U(\vecm_{\pmb \ell}+\vecm_{\pmb \ell}^{\prime},\vectheta+\vectheta^{\prime}),\\
    \hat U(\vecm_{\pmb \ell},\vectheta)\hat U(\vecm_{\pmb \ell}^{\prime},\vectheta^{\prime})\hat U^\dagger(\vecm_{\pmb \ell},\vectheta) &=  e^{\ii (\vecm_{\pmb \ell}\cdot\vectheta^\prime - \vecm_{\pmb \ell}^\prime \cdot\vectheta)}\hat U(\vecm_{\pmb \ell}^{\prime},\vectheta^{\prime}),\\
    \mathrm{Tr}[\hat U(\vecm_{\pmb \ell},\vectheta)] &= \ell_1\ell_2\, \delta_{\vecm_{\pmb \ell}, \bm{0}}\,\delta^2(\vectheta),
\end{align*}
where $\bm{0}$ is the zero vector. On $L^2(\T^2_{\pmb \ell},\ud^2\pmb{\gamma})$, any nonzero square-integrable function $\psi(\pmb{\gamma})$ is a valid fiducial vector and yields a family of coherent states. The most useful fiducial vectors are normalized and well localized in both angular position and angular momentum. Such fiducial vectors are obtained from square-integrable functions on the plane $\R^2$ by periodizing them over the unit cell $[-\frac{\ell_1}{2},\frac{\ell_1}{2})\times[-\frac{\ell_2}{2},\frac{\ell_2}{2})$ (Appendix~\ref{app:periodization}).

\section{Modular states and integer basis}
\label{app:modular}

\subsection{Orthogonality of the modular states}
The orthogonality of the states of Eq.~\eqref{eq:zak_state} follows from $\bra{x}\ket{x'}=\delta(x-x')$ and from the Poisson summation formula $\sum_{n\in\Z}e^{-\ii n\ell\overline{p}}=\frac{2\pi}{\ell}\sum_{j\in\Z}\delta(\overline{p}-2\pi j/\ell)$:
\begin{equation}
\begin{split}
\bra{\overline{x},\overline{p}}\ket{\overline{x}',\overline{p}'}&=\ell\sum_{n,n'\in\Z}e^{-\ii n\ell\overline{p}+\ii n'\ell\overline{p}'}\bra{\overline{x}+n\ell}\ket{\overline{x}'+n'\ell}\\
&=\ell\,\delta(\overline{x}-\overline{x}')\sum_{n\in\Z}e^{-\ii n\ell(\overline{p}-\overline{p}')}\\
&=2\pi\,\delta(\overline{x}-\overline{x}')\,\delta(\overline{p}-\overline{p}'),
\end{split}
\end{equation}
where, for $\overline{x}$ and $\overline{x}'$ in the unit cell, $\bra{\overline{x}+n\ell}\ket{\overline{x}'+n'\ell}$ vanishes unless $n=n'$, and only the term $j=0$ of the Poisson sum contributes for $\overline{p}$ and $\overline{p}'$ in the unit cell.

\subsection{Orthonormality of the integer basis}
With the position representation of Eq.~\eqref{eq:integer_basis_position},
\begin{equation}
\bra{k,q}\ket{k',q'}=\frac{\delta_{qq'}}{\ell}\int_{-\ell/2}^{\ell/2}\dd\overline{x}\,e^{2\ii\pi(k'-k)\overline{x}/\ell}=\delta_{kk'}\delta_{qq'},
\end{equation}
since two different cells do not overlap. The overlap $\bra{k,q}\ket{\overline{x},\overline{p}}=e^{-2\ii\pi k\overline{x}/\ell}e^{-\ii\ell q\overline{p}}$ reproduces the first relation of Eq.~\eqref{eq:integer_basis}. Projecting $\ket{\psi}$ onto $\ket{k,q}$ gives Eq.~\eqref{eq:psi_kq}, and the Parseval identity on each cell, summed over the cells, gives $\sum_{k,q\in\Z}|\psi_{kq}|^2=\int_\R|\psi(x)|^2\dd x$.

\subsection{Winding of the modular wavefunction}

The quasi-periodicity of Eq.~\eqref{eq:quasi_periodicity} forces the modular wavefunction $\psi(\overline{x},\overline{p})$ to vanish somewhere in the unit cell. The argument relies on the winding number of the phase of $\psi(\overline{x},\overline{p})$ along the boundary of the unit cell. Assume first that $\psi(\overline{x},\overline{p})$ does not vanish on this boundary, so that its phase $\varphi=\arg\psi$ can be followed continuously along it. The winding number is
\begin{equation}
\label{eq:winding_def}
\mathcal{W}=\frac{1}{2\pi}\oint_{\partial C}\dd\varphi,
\end{equation}
where $\partial C$ is the boundary of the unit cell $C=[-\ell/2,\ell/2)\times[-\pi/\ell,\pi/\ell)$, traversed counterclockwise in the $(\overline{x},\overline{p})$ plane. The boundary consists of four edges:
\begin{enumerate}
\item the bottom edge, $\overline{p}=-\pi/\ell$, with $\overline{x}$ running from $-\ell/2$ to $\ell/2$;
\item the right edge, $\overline{x}=\ell/2$, with $\overline{p}$ running from $-\pi/\ell$ to $\pi/\ell$;
\item the top edge, $\overline{p}=\pi/\ell$, with $\overline{x}$ running from $\ell/2$ to $-\ell/2$;
\item the left edge, $\overline{x}=-\ell/2$, with $\overline{p}$ running from $\pi/\ell$ to $-\pi/\ell$.
\end{enumerate}
The phase change along the boundary is the sum of the phase changes along the four edges.

The bottom and top edges cancel. By the periodicity of Eq.~\eqref{eq:quasi_periodicity} in $\overline{p}$, $\psi(\overline{x},\pi/\ell)=\psi(\overline{x},-\pi/\ell)$, so that the phase takes the same values on both edges. Since the two edges are traversed in opposite directions, their phase changes are opposite and cancel.

The left and right edges do not cancel. The quasi-periodicity in $\overline{x}$, taken at $\overline{x}=-\ell/2$, gives
\begin{equation}
\label{eq:winding_edges}
\psi(\ell/2,\overline{p})=e^{\ii\ell\overline{p}}\,\psi(-\ell/2,\overline{p}),
\end{equation}
so that the phase on the right edge equals the phase on the left edge plus $\ell\overline{p}$, up to a constant multiple of $2\pi$, which is fixed by continuity. When $\overline{p}$ runs from $-\pi/\ell$ to $\pi/\ell$, the extra phase $\ell\overline{p}$ increases by $\ell\times2\pi/\ell=2\pi$. The phase change along the right edge, traversed upwards, therefore exceeds the phase change along the left edge, traversed upwards, by $2\pi$. Since the left edge is traversed downwards, the two contributions add up to
\begin{equation}
\label{eq:winding_result}
\oint_{\partial C}\dd\varphi=\int_{-\pi/\ell}^{\pi/\ell}\ell\,\dd\overline{p}=2\pi,\qquad \mathcal{W}=1 .
\end{equation}

A nonzero winding number implies a zero inside the unit cell. If $\psi(\overline{x},\overline{p})$ did not vanish anywhere in the cell, the boundary loop could be shrunk continuously to a point while keeping $\psi\neq0$ along the loop. During this deformation the winding number would stay an integer and vary continuously, hence remain constant. It would then equal the winding number of a vanishingly small loop around a point where $\psi\neq0$, which is zero, in contradiction with $\mathcal{W}=1$. A continuous modular wavefunction therefore vanishes at least once inside the unit cell. If $\psi(\overline{x},\overline{p})$ already vanishes on the boundary, as for the vacuum at the corner $(\overline{x},\overline{p})=(\ell/2,\pi/\ell)$, the existence of a zero is immediate.

\subsection{Position eigenstates and combs}
Eq.~\eqref{eq:zak_state} gives
\begin{equation}
\begin{split}
\int_{-\pi/\ell}^{\pi/\ell}\dd\overline{p}\,e^{-\ii n\ell\overline{p}}\ket{\overline{x}_0,\overline{p}}&=\sqrt{\ell}\sum_{n'\in\Z}\ket{x=\overline{x}_0+n'\ell}\int_{-\pi/\ell}^{\pi/\ell}\dd\overline{p}\,e^{\ii(n'-n)\ell\overline{p}}\\
&=\frac{2\pi}{\sqrt\ell}\ket{x=\overline{x}_0+n\ell},
\end{split}
\end{equation}
which is the position eigenstate used in the main text. At $\overline{p}=0$, Eq.~\eqref{eq:zak_state} gives the comb $\sum_{s\in\Z}\ket{x=s\ell+x_0}=\ell^{-1/2}\ket{\overline{x}=x_0,\overline{p}=0}$. With $\bra{k,q}\ket{\overline{x},\overline{p}}=e^{-2\ii\pi k\overline{x}/\ell}e^{-\ii\ell q\overline{p}}$, the coefficients are $\psi_{kq}=\ell^{-1/2}e^{-2\ii\pi k\overline{x}_0/\ell}\delta_{q,-n}$ for the position eigenstate, whose cell is centered at $-q\ell=n\ell$, and $\psi_{kq}=\ell^{-1/2}e^{-2\ii\pi kx_0/\ell}$ for the comb.

\section{Coherent states from periodized fiducial vectors}
\label{app:periodization}

\subsection{Coherent-state transform and resolution of the identity}
Coherent states on the phase space $\Gamma$ are obtained from a fiducial vector $\psi\in L^2(\T^2_{\pmb\ell})$ by the displacement operators $\hat U(\vecm_{\pmb\ell},\bm b)$ of Eq.~\eqref{eq:U_def}, which translate a function by $\bm b$ on the torus and multiply it by the plane wave of wave vector $\vecm_{\pmb\ell}$. The coherent states $\psi_{\vecm_{\pmb\ell},\bm b}=\hat U(\vecm_{\pmb\ell},\bm b)\psi$ define the transform
\begin{equation}
\label{coherentstate}
\phi\;\mapsto\;\bra{\psi_{\vecm_{\pmb\ell},\bm b}}\ket{\phi}\equiv\Phi(\vecm_{\pmb\ell},\bm b),
\end{equation}
which maps a function $\phi$ on the torus to a function $\Phi$ on the double discrete cylinder $\Gamma$. For a Gaussian-like fiducial vector, $|\Phi(\vecm_{\pmb\ell},\bm b)|^2$ is the Husimi distribution of $\phi$ on $\Gamma$, in the same way as the projection on the coherent states $\ket{[\alpha]}$ gives the Husimi distribution on the usual phase space $\R^2$. With the inner product of Eq.~\eqref{eq:periodic_representative}, the coherent states resolve the identity for any nonzero fiducial vector,
\begin{equation}
\label{eq:cs_resolution}
\sum_{k,q\in\Z}\int_{\T^2_{\pmb\ell}}\frac{\dd^2\bm b}{2\pi}\,\ket{\psi_{\vecm_{\pmb\ell},\bm b}}\bra{\psi_{\vecm_{\pmb\ell},\bm b}}=\|\psi\|^2\,\bu,
\end{equation}
where $\|\psi\|^2=\bra{\psi}\ket{\psi}$. This relation follows from the Parseval identity on the torus: for a fixed $\bm b$, $\Phi(\vecm_{\pmb\ell},\bm b)$ is the Fourier coefficient at $(k,q)$ of the product $\overline{\psi(\bm z-\bm b)}\,\phi(\bm z)$, so that the sum over $(k,q)$ of $|\Phi|^2$ equals $\frac{1}{2\pi}\int\dd^2\bm z\,|\psi(\bm z-\bm b)|^2|\phi(\bm z)|^2$, and the integral over $\bm b$ then gives $\|\psi\|^2\|\phi\|^2$.

\subsection{Reproducing kernel}
For a given fiducial vector $\psi$, the reproducing kernel is the overlap of two coherent states,
\begin{equation}
\label{reprodKer}
\begin{split}
K(\vecm_{\pmb\ell},\bm b;\vecm'_{\pmb\ell},\bm b')&=\bra{\psi_{\vecm_{\pmb\ell},\bm b}}\ket{\psi_{\vecm_{\pmb\ell}',\bm b'}}\\
&=\frac{1}{2\pi}\int_{\T^2_{\pmb\ell}}\dd^2\bm z\;\overline{(\hat U(\vecm_{\pmb\ell},\bm b)\psi)(\bm z)}\,(\hat U(\vecm_{\pmb\ell}',\bm b')\psi)(\bm z).
\end{split}
\end{equation}
Inserting the resolution of the identity of Eq.~\eqref{eq:cs_resolution} into Eq.~\eqref{coherentstate} shows that every transformed function obeys the reproducing property
\begin{equation}
\label{eq:reproducing_property}
\Phi(\vecm_{\pmb\ell},\bm b)=\frac{1}{\|\psi\|^2}\sum_{k',q'\in\Z}\int_{\T^2_{\pmb\ell}}\frac{\dd^2\bm b'}{2\pi}\,K(\vecm_{\pmb\ell},\bm b;\vecm_{\pmb\ell}',\bm b')\,\Phi(\vecm_{\pmb\ell}',\bm b').
\end{equation}
The functions $\Phi$ are therefore not arbitrary functions on $\Gamma$: their values at different points are correlated through the kernel, which plays the role of the Gaussian overlap $|\bra{[\alpha]}\ket{[\beta]}|^2=e^{-|\alpha-\beta|^2}$ of the usual coherent states. These properties are summarized in the following proposition.
\begin{prop}
\label{main_properties}
For any elements $\phi, \psi\in\,\mathcal{H}$, such that $\Vert\psi\Vert=1$, the map
\begin{equation}
\phi\rightarrow\langle\psi_{(\vecm_{\pmb \ell},\bm b)}\vert\phi\rangle \equiv \Phi(\vecm_{\pmb \ell},\bm b)
\end{equation}
satisfies the following properties:
\begin{itemize}
\item it is an isometry,
\begin{align}
\label{isom}
\Vert\phi\Vert^2=\frac{1}{2\pi}\int_{\T^2_{\pmb \ell}} \ud^2\bm z\ \vert\phi(\bm z)\vert^2= \sum_{\vecm_{\pmb \ell}\in\Z^2_{\pmb\ell}}\int_{\T^2_{\pmb \ell}} \frac{\ud^2\bm b}{2\pi}\,\vert{\langle\psi_{(\vecm_{\pmb \ell},\bm b)}|\phi\rangle\vert}^2\,;
\end{align}
\item it can be inverted on its range,
\begin{equation}
\label{invert}
\phi(\bm z)= \sum_{\vecm_{\pmb \ell}\in\Z^2_{\pmb\ell}}\int_{\T^2_{\pmb \ell}} \frac{\ud^2\bm b}{2\pi}\,\langle\psi_{(\vecm_{\pmb \ell},\bm b)}|\phi\rangle \,\psi_{(\vecm_{\pmb \ell},\bm b)}(\bm z)\, ;
\end{equation}
\item its range lies in the reproducing kernel space characterized by
\begin{equation}
\label{reprod}
\langle\psi_{(\vecm_{\pmb \ell},\bm b)}|\phi\rangle= \sum_{\vecm_{\pmb \ell}^\prime\in\Z^2_{\pmb\ell}}\int_{\T^2_{\pmb \ell}} \frac{\ud^2\bm b^\prime}{2\pi}\,
K_{\psi}((\vecm_{\pmb \ell},\bm b), (\vecm_{\pmb \ell}^{\prime},\bm b^{\prime}))\,\Phi(\vecm_{\pmb \ell}^{\prime},\bm b^{\prime})\,,
\end{equation}
\end{itemize}
with the kernel $K_{\psi}((\vecm_{\pmb \ell},\bm b), (\vecm_{\pmb \ell}^{\prime},\bm b^{\prime}))=\langle\psi_{(\vecm_{\pmb \ell},\bm b)}|\psi_{(\vecm_{\pmb \ell}^{\prime},\bm b^{\prime})}\rangle$.
\end{prop}

\subsection{Periodization of fiducial vectors}
To construct coherent states and reproducing kernels from fiducial functions on $\T^2_{\pmb\ell}$, these functions must be periodic. A function $\psi$ on $\R^2$, such as a Gaussian, is turned into a function on the torus by the periodization operator
\begin{equation}
\label{periodization}
(\textsf{Per}\,\psi)(\overline{x},\overline{p})=2\pi\sum_{n_1,n_2\in\Z}\psi\Big(\overline{x}+\ell n_1,\,\overline{p}+\frac{2\pi}{\ell}n_2\Big),
\end{equation}
which sums the copies of $\psi$ translated by all the periods of the torus. With the Fourier transform $\mathcal F[\psi](\bm\omega)=\int_{\R^2}\dd^2\bm z\,\psi(\bm z)\,e^{-\ii\bm\omega\cdot\bm z}$, the Poisson summation formula applied in each variable gives
\begin{equation}
\label{eq:poisson_periodization}
2\pi\sum_{n_1,n_2\in\Z}\psi\Big(\overline{x}+\ell n_1,\,\overline{p}+\frac{2\pi}{\ell}n_2\Big)=\sum_{k,q\in\Z}\mathcal F[\psi](\vecm_{\pmb\ell})\,e^{\ii\vecm_{\pmb\ell}\cdot\bm z},
\end{equation}
where the prefactor $2\pi$ compensates the factor $1/(\ell\times2\pi/\ell)=1/(2\pi)$ of the two-dimensional Poisson formula. Periodization on the torus $\T^2_{\pmb\ell}$ is thus dual to sampling on the grid $\Z^2_{\pmb\ell}$: the Fourier coefficients of $\textsf{Per}\,\psi$, with the inner product of Eq.~\eqref{eq:periodic_representative}, are the values of $\mathcal F[\psi]$ at the grid points $\vecm_{\pmb\ell}$. When the width of a Gaussian is small compared with the periods, $\textsf{Per}\,\psi$ is concentrated around the origin of the unit cell and reduces to the term $n_1=n_2=0$ up to exponentially small corrections. The periodization operator is not unitary, so that a periodized function must be normalized again.

\subsection{Wrapped Gaussian}
For a Gaussian without covariance between its two variables, centered at $(\mu_1,\mu_2)$ with the widths $\sigma_1$ and $\sigma_2$,
\begin{align}
\label{multivarGauss_simple}
     \psi(x,y) = \frac{1}{\sqrt{2\pi}\sigma_1\sigma_2} e^{-\left(\frac{(x - \mu_1)^2}{2\sigma_1^2} + \frac{(y - \mu_2)^2}{2\sigma_2^2} \right)}\,,
\end{align}
the periodization of Eq.~\eqref{periodization} gives the wrapped Gaussian
\begin{equation}
    \label{wrapped_gaussian}
    \begin{split}
    (\textsf{Per}\,\psi)(\bm z) = \frac{\sqrt{2\pi}}{\sigma_1\sigma_2}\, e^{-\left(\frac{(\overline{x} - \mu_1)^2}{2\sigma_1^2} + \frac{(\overline{p} - \mu_2)^2}{2\sigma_2^2}\right)}&\, \vartheta_3\left(\frac{\ii\ell(\overline{x}-\mu_1)}{2\pi\sigma_1^2},\frac{\ii \ell^2}{2\pi\sigma_1^2}\right)\\
    &\times\vartheta_3\left(\frac{\ii(\overline{p}-\mu_2)}{\ell\sigma_2^2},\frac{2\ii \pi}{\ell^2\sigma_2^2}\right),
    \end{split}
\end{equation}
where
\begin{equation}
\label{jacobi3}
    \vartheta_3(q,s) = \sum_{n=-\infty}^{\infty} e^{2\ii n \pi q}e^{\ii\pi sn^2}
\end{equation}
is the third Jacobi theta function. The wrapped Gaussian is a Gaussian envelope stratified by narrower Gaussian packets, which come from the theta functions, as shown in Fig.~\ref{Wrapped_gauss_visual}. The fiducial vector $\psi_0$ of the wrapped coherent states of Eq.~\eqref{eq:cs_fiducial} is the wrapped Gaussian with $\mu_1=\mu_2=0$, $\sigma_1=r_1$ and $\sigma_2=r_2$.

\section{Properties of the wrapped coherent states}
\label{app:coherent}

\subsection{Eigenvalue equation}
The commutator $[\ii\hat\phi,-\hat N]=-\ii[\hat\phi,\hat N]=1$ is a scalar, and the Baker--Campbell--Hausdorff formula gives $e^{\ii\hat\phi-\hat N}=e^{\ii\hat\phi}e^{-\hat N}e^{-1/2}$. Since $e^{\ii\hat\phi}\ket{n}=\ket{n+1}$, the operator $\hat Z$ acts as $\hat Z\ket{n}=e^{-(n+1/2)}\ket{n+1}$. For the unnormalized state $\ket{\zeta}=\sum_ke^{-k^2/2}e^{k\zeta}\ket{k}$, with $\zeta\in\C$, the relabeling $k+1\to k$ gives
\begin{equation}
\hat Z\ket{\zeta}=\sum_{k\in\Z}e^{-k^2/2-k-1/2+k\zeta}\ket{k+1}=\sum_{k\in\Z}e^{-(k-1)^2/2-k+1/2+(k-1)\zeta}\ket{k}=e^{-\zeta}\ket{\zeta},
\end{equation}
since $-(k-1)^2/2-k+1/2=-k^2/2$. The wrapped coherent state of Eq.~\eqref{eq:wrapped_cs} is $\ket{\zeta}$ with $\zeta=m-\ii\phi$, multiplied by $Ae^{-m^2/2}$.

\subsection{Normalization and integer heights}
For $m\in\Z$, the relabeling $k\to k+m$ gives $\sum_ke^{-k^2}e^{2km}=e^{m^2}\sum_ke^{-(k-m)^2}=e^{m^2}A^{-2}$, which fixes the prefactor $Ae^{-m^2/2}$ independently of $m$. For a non-integer $m$, the sum $\sum_ke^{-(k-m)^2}$ depends on the fractional part of $m$, and the normalization prefactor would depend on $m$. The states $\ket{m,\phi}$ are therefore normalized for integer heights only, which restricts the height of the cylinder to $\Z$. On the double discrete cylinder, the heights $(m_1,m_2)$ of the coherent states of Eq.~\eqref{eq:coherent} label the point $\vecm_{\pmb\ell}=(2\pi m_1/\ell,\ell m_2)$ of the grid $\Z^2_{\pmb\ell}$ of Sec.~\ref{sec:config}, where the integer labels of the grid are renamed $(m_1,m_2)$ to distinguish them from the labels $(k,q)$ of the integer basis.

\subsection{Coherent states as Weyl-transported periodized Gaussians}
The overlap $\bra{\overline{x}',\overline{p}'}\ket{k,q}=e^{2\ii\pi k\overline{x}'/\ell}e^{\ii\ell q\overline{p}'}$ and Eq.~\eqref{eq:coherent} give the function on the torus that represents the coherent state,
\begin{equation}
\label{eq:app_cs_wf}
\bra{\overline{x}',\overline{p}'}\ket{m_1,\overline{x};m_2,\overline{p}}=A^2e^{-(m_1^2+m_2^2)/2}\,S_{1/2}\Big(m_1+\frac{2\ii\pi}{\ell}(\overline{x}'-\overline{x})\Big)\,S_{1/2}\big(m_2+\ii\ell(\overline{p}'-\overline{p})\big),
\end{equation}
with the theta sum $S_{1/2}$ of Eq.~\eqref{eq:theta_sum}. For an integer height $m$ and a real angle $u$, the relabeling $k\to k+m$ gives
\begin{equation}
\label{eq:app_shift_height}
e^{-m^2/2}S_{1/2}(m+\ii u)=\sum_{k\in\Z}e^{-(k-m)^2/2}e^{\ii ku}=e^{\ii mu}\,S_{1/2}(\ii u).
\end{equation}
With $u_1=2\pi(\overline{x}'-\overline{x})/\ell$ and $u_2=\ell(\overline{p}'-\overline{p})$, the phase $e^{\ii m_1u_1+\ii m_2u_2}$ equals $e^{\ii\vecm_{\pmb\ell}\cdot(\bm z'-\bm b)}$, with $\vecm_{\pmb\ell}=(2\pi m_1/\ell,\ell m_2)$, $\bm z'=(\overline{x}',\overline{p}')$ and $\bm b=(\overline{x},\overline{p})$. Eq.~\eqref{eq:app_cs_wf} therefore becomes
\begin{equation}
\begin{split}
\bra{\overline{x}',\overline{p}'}\ket{m_1,\overline{x};m_2,\overline{p}}&=e^{-\ii\vecm_{\pmb\ell}\cdot\bm b}\,e^{\ii\vecm_{\pmb\ell}\cdot\bm z'}\,\psi_0(\bm z'-\bm b)=e^{-\ii\vecm_{\pmb\ell}\cdot\bm b/2}\,\big(\hat U(\vecm_{\pmb\ell},\bm b)\psi_0\big)(\bm z'),\\
\psi_0(\overline{x},\overline{p})&=A^2S_{1/2}\Big(\frac{2\ii\pi\overline{x}}{\ell}\Big)S_{1/2}(\ii\ell\overline{p}),
\end{split}
\end{equation}
where
\begin{equation}
(\hat U(\vecm_{\pmb\ell},\bm b)\psi_0)(\bm z')=e^{-\ii\vecm_{\pmb\ell}\cdot\bm b/2}e^{\ii\vecm_{\pmb\ell}\cdot\bm z'}\psi_0(\bm z'-\bm b).
\end{equation} The Poisson summation formula gives
\begin{equation}
S_{1/2}(\ii u)=\sum_{k\in\Z}e^{-k^2/2}e^{\ii ku}=\sqrt{2\pi}\sum_{j\in\Z}e^{-(u+2\pi j)^2/2}.
\end{equation}
With $u=2\pi\overline{x}/\ell$, $(u+2\pi j)^2=(\overline{x}+j\ell)^2/r_1^2$ with $r_1=\ell/(2\pi)$, and with $u=\ell\overline{p}$, $(u+2\pi j)^2=(\overline{p}+2\pi j/\ell)^2/r_2^2$ with $r_2=1/\ell$, so that
\begin{equation}
\psi_0(\overline{x},\overline{p})=2\pi A^2\sum_{j_1,j_2\in\Z}e^{-(\overline{x}+j_1\ell)^2/2r_1^2}\,e^{-(\overline{p}+2\pi j_2/\ell)^2/2r_2^2}=A^2\,(\textsf{Per}\,G_{\bm r})(\overline{x},\overline{p}),
\end{equation}
with the periodization operator of Eq.~\eqref{periodization} and $G_{\bm r}(x,y)=e^{-x^2/2r_1^2-y^2/2r_2^2}$, which gives Eq.~\eqref{eq:cs_fiducial}. The Fourier coefficients of $\psi_0$ are $A^2e^{-(k^2+q^2)/2}$, so that $\|\psi_0\|^2=A^4\sum_ke^{-k^2}\sum_qe^{-q^2}=1$. The modulus of Eq.~\eqref{eq:app_cs_wf} is maximal at $(\overline{x}',\overline{p}')=(\overline{x},\overline{p})$, where all the terms of the theta sums are positive.

\subsection{Resolution of the identity and slice weights}
Inserting Eq.~\eqref{eq:coherent} into the left-hand side of Eq.~\eqref{eq:resolution} and using the orthogonality relations
\begin{equation}
\int_{-\ell/2}^{\ell/2}\dd\overline{x}\,e^{-2\ii\pi(k-k')\overline{x}/\ell}=\ell\,\delta_{kk'},\qquad \int_{-\pi/\ell}^{\pi/\ell}\dd\overline{p}\,e^{-\ii\ell(q-q')\overline{p}}=\frac{2\pi}{\ell}\delta_{qq'}
\end{equation}
gives
\begin{equation}
\begin{split}
&\frac{1}{2\pi}\sum_{m_1,m_2\in\Z}\int_{-\ell/2}^{\ell/2}\!\!\dd\overline{x}\int_{-\pi/\ell}^{\pi/\ell}\!\!\dd\overline{p}\,\ket{m_1,\overline{x};m_2,\overline{p}}\bra{m_1,\overline{x};m_2,\overline{p}}\\
&\qquad=A^4\sum_{k,q\in\Z}\Big[\sum_{m_1}e^{-(m_1-k)^2}\Big]\Big[\sum_{m_2}e^{-(m_2-q)^2}\Big]\ket{k,q}\bra{k,q}=\bu,
\end{split}
\end{equation}
where $e^{-m_1^2+2km_1-k^2}=e^{-(m_1-k)^2}$ and $\sum_{m}e^{-(m-k)^2}=A^{-2}$ for integer $k$. The same calculation performed on $\bra{\psi}\cdot\ket{\psi}$ at fixed $(m_1,m_2)$, without the sums over the heights, gives the slice weight of Eq.~\eqref{eq:slice_weight}. Summing over $m_2$ with $\sum_{m_2}e^{-(m_2-q)^2}=A^{-2}$ gives $P_1(m_1)$ in Eq.~\eqref{eq:height_weights}, and $P_2(m_2)$ follows in the same way.

\section{Zeros of the stellar function}
\label{app:lemma}
\label{sec:zeros}

\subsection{Number of zeros and stellar rank}
Since the coefficients of Eq.~\eqref{eq:stellar} are suppressed by $e^{-(k^2+q^2)/2}$, the stellar function $F_{\ket{\psi}}(z_1,z_2)$ is an entire function of $(z_1,z_2)\in\C^2$, $2\ii\pi$-periodic in each variable. Its zeros are counted in the traditional way of complex analysis. At fixed $z_2$, the function $z_1\mapsto F_{\ket{\psi}}(z_1,z_2)$ is holomorphic, and by the argument principle the number of its zeros inside a closed contour $\mathcal C$ of the complex plane, counted with their multiplicities, is
\begin{equation}
\label{eq:argument_principle}
N_{\mathcal C}=\frac{1}{2\ii\pi}\oint_{\mathcal C}\dd z_1\,\frac{\partial_{z_1}F_{\ket{\psi}}(z_1,z_2)}{F_{\ket{\psi}}(z_1,z_2)},
\end{equation}
that is, the number of turns of the phase of $F_{\ket{\psi}}$ along $\mathcal C$, and the same holds in the variable $z_2$. The zeros relevant to the Husimi distribution are those lying on $\Gamma$, that is, with integer real parts $\mathrm{Re}\,z_1=m_1$ and $\mathrm{Re}\,z_2=m_2$. The stellar rank on $\Gamma$ is the number of these zeros, counted with their multiplicities.

For a separable state, the stellar function factorizes as $F_{\ket{\psi}}(z_1,z_2)=F_1(z_1)F_2(z_2)$, and the zeros of each factor are isolated points of $\C$. By periodicity, a factor with one zero has infinitely many zeros. A zero on $\Gamma$ requires one of the two factors to vanish at an integer height, which is not the case for generic parameters; when it does, the zeros on $\Gamma$ form whole circles of the phase space, and the stellar rank is infinite. For a non-separable state, the zero set of $F_{\ket{\psi}}$ is an analytic curve of $\C^2$, which meets each slice $(m_1,m_2)$ at isolated points of the torus, since the slice and the zero set are both two-dimensional surfaces of the four-dimensional space $\C^2$. These zeros persist under small changes of the state. On a slice, a zero is detected by the turn of the phase of $F_{\ket{\psi}}$ along a small loop around it in the $(\overline{x},\overline{p})$ plane,
\begin{equation}
\label{eq:index}
\mathrm{ind}(m_1,m_2;\overline{x},\overline{p})=\frac{1}{2\pi}\oint\dd\arg F_{\ket{\psi}}(z_1,z_2)=\pm1
\end{equation}
for a simple zero, which is the quantity used to locate the zeros numerically (Appendix~\ref{app:numerics}).

\subsection{Zeros of the theta sums}
The theta sums of Eq.~\eqref{eq:theta_sum} are third Jacobi theta functions,
\begin{equation}
\label{eq:theta_jacobi}
S_a(\zeta)=\sum_{k\in\Z}e^{-ak^2}e^{k\zeta}=\vartheta_3\Big(\frac{\zeta}{2\ii\pi},\frac{\ii a}{\pi}\Big),\qquad a>0,
\end{equation}
with $\vartheta_3$ defined in Eq.~\eqref{jacobi3}. Their zeros are given by the following lemma.
\begin{lem}
\label{lemma:zeros}
$S_a(\zeta)=0$ if and only if $\zeta=(2n+1)a+\ii(2p+1)\pi$ with $n,p\in\Z$. These zeros are simple, and each period cell contains exactly one of them.
\end{lem}
\begin{proof}
The sum $S_a(\zeta)$ is $2\ii\pi$-periodic, and the relabeling $k-1\to k$ gives the quasi-periodicity
\begin{equation}
S_a(\zeta+2a)=\sum_{k\in\Z}e^{-a(k^2-2k)}e^{k\zeta}=e^{a+\zeta}S_a(\zeta).
\end{equation}
At $\zeta=a+\ii\pi$, $S_a(\zeta)=\sum_k(-1)^ke^{-ak^2+ak}=e^{a/4}\sum_k(-1)^ke^{-a(k-1/2)^2}$, which changes sign under $k\to1-k$ and therefore vanishes. By periodicity and quasi-periodicity, $S_a(\zeta)$ vanishes at all $\zeta=(2n+1)a+\ii(2p+1)\pi$, with $n,p\in\Z$. The argument principle of Eq.~\eqref{eq:argument_principle}, applied to the period cell of vertices $\zeta_0$, $\zeta_0+2a$, $\zeta_0+2a+2\ii\pi$ and $\zeta_0+2\ii\pi$, chosen without zero on its boundary, counts the zeros inside the cell: the horizontal sides cancel by periodicity, and the vertical sides contribute $\frac{1}{2\ii\pi}\int_0^{2\ii\pi}\dd(a+\zeta)=1$ by quasi-periodicity. Each period cell thus contains exactly one zero, which is simple.
\end{proof}

\subsection{Superpositions}
For $\ket{\pm}_x$, the combination
\begin{equation}
S_a(\zeta+\ii\pi/2)\pm S_a(\zeta-\ii\pi/2)=\sum_{k\in\Z}e^{-ak^2}e^{k\zeta}\,(e^{\ii\pi k/2}\pm e^{-\ii\pi k/2})
\end{equation}
retains even or odd values of $k$ only. With $k=2j$ in the first case, $2\cos(\pi j)=2(-1)^j$, and with $k=2j+1$ in the second case, $2\ii\sin(\pi(2j+1)/2)=2\ii(-1)^j$, so that
\begin{align}
S_a\Big(\zeta+\frac{\ii\pi}{2}\Big)+S_a\Big(\zeta-\frac{\ii\pi}{2}\Big)&=2\sum_j(-1)^je^{-4aj^2}e^{2j\zeta}=2\,S_{4a}(2\zeta+\ii\pi),\\
S_a\Big(\zeta+\frac{\ii\pi}{2}\Big)-S_a\Big(\zeta-\frac{\ii\pi}{2}\Big)&=2\ii\sum_j(-1)^je^{-a(2j+1)^2}e^{(2j+1)\zeta}\nonumber\\
&=2\ii\,e^{\zeta-a}\,S_{4a}(2\zeta-4a+\ii\pi).
\end{align}
Lemma~\ref{lemma:zeros} applied to $S_{4a}$ gives $2\zeta+\ii\pi=(2n+1)4a+\ii(2p+1)\pi$, that is $\zeta=(2n+1)2a+\ii p\pi$, for the sum, and $2\zeta-4a+\ii\pi=(2n+1)4a+\ii(2p+1)\pi$, that is $\zeta=4a(n+1)+\ii p\pi$, for the difference, with $n,p\in\Z$. With $\zeta=z_1$, the zeros lie at $\overline{x}=p\ell/2$, that is at $\overline{x}\in\{0,\ell/2\}$ modulo $\ell$, and at the heights $m_1=(2n+1)2a$ for the sum and $m_1=4an$ for the difference, after the relabeling $n+1\to n$. For $a=1/2$, these heights are the odd and even integers of the ideal case, and for $a=\frac12+\frac{2\pi^2\sigma^2}{\ell^2}$ they are the heights of Eq.~\eqref{eq:zeros_pm_lc}.

\section{Vacuum and Gaussian states centered at the origin}
\label{app:coherent_state}

\subsection{Coefficients}
With Eq.~\eqref{eq:psi_kq} and the change of variable $x=\overline{x}-q\ell$, for which $e^{-2\ii\pi k\overline{x}/\ell}=e^{-2\ii\pi kx/\ell}$, the coefficients of the Gaussian state of Eq.~\eqref{eq:coherent_wf} are integrals of a Gaussian over the cell $I_q=[(-q-\frac12)\ell,(-q+\frac12)\ell)$. With $\beta_k=\sqrt2\pi\sigma k/\ell$, completing the square gives
\begin{equation}
-\frac{x^2}{2\sigma^2}-\frac{2\ii\pi kx}{\ell}=-\frac{(x+\ii\sqrt2\sigma\beta_k)^2}{2\sigma^2}-\beta_k^2,
\end{equation}
and the change of variable $v=(x+\ii\sqrt2\sigma\beta_k)/(\sqrt2\sigma)$, which maps the edges of $I_q$ onto $\xi^\pm_q+\ii\beta_k$, gives
\begin{equation}
\begin{split}
\int_{I_q}\dd x\,e^{-x^2/2\sigma^2}e^{-2\ii\pi kx/\ell}&=\sqrt2\sigma\,e^{-\beta_k^2}\int_{\xi^-_q+\ii\beta_k}^{\xi^+_q+\ii\beta_k}\dd v\,e^{-v^2}\\
&=\sqrt{\frac\pi2}\,\sigma\,e^{-\beta_k^2}\Big[\mathrm{erf}\big(\xi^+_q+\ii\beta_k\big)-\mathrm{erf}\big(\xi^-_q+\ii\beta_k\big)\Big],
\end{split}
\end{equation}
with $\beta_k^2=2\pi^2\sigma^2k^2/\ell^2$, hence Eq.~\eqref{eq:psi_kq_coherent}. Since $e^{-v^2}$ is entire, the integral depends only on its endpoints. For $q=0$, $\xi^\pm_0=\pm\ell/(2\sqrt2\,\sigma)$, and $-\mathrm{erf}(-\xi+\ii\beta)=\mathrm{erf}(\xi-\ii\beta)$ gives the bracket of Eq.~\eqref{eq:coherent_stellar}. For large $|k|$, the error functions grow as $e^{\beta_k^2}$, which compensates the Gaussian prefactor. The product is evaluated in a numerically stable way as $e^{-\beta^2}\mathrm{erf}(\xi+\ii\beta)=e^{-\beta^2}-e^{-\xi^2-2\ii\xi\beta}\,w(-\beta+\ii\xi)$ for $\xi\ge0$, with the Faddeeva function $w(z)=e^{-z^2}\mathrm{erfc}(-\ii z)$, and through $\mathrm{erf}(-z)=-\mathrm{erf}(z)$ for $\xi<0$.

\subsection{Single-cell approximation}
For $\sigma\ll\ell$, $\psi(\overline{x}+n\ell)$ is negligible for $n\neq0$ when $\overline{x}$ lies in the unit cell, and the Zak transform of Eq.~\eqref{eq:modular_wf} reduces to $\sqrt\ell\,\psi(\overline{x})$, independent of $\overline{p}$, which gives Eq.~\eqref{eq:coherent_single_cell}. The Fourier coefficients of the flat momentum function $H(\overline{p})$ are $\frac{\ell}{2\pi}\int_{-\pi/\ell}^{\pi/\ell}\dd\overline{p}\,e^{-\ii\ell q\overline{p}}=\delta_{q0}$, so that only $q=0$ contributes, which is the term $q=0$ of Eq.~\eqref{eq:psi_kq_coherent}. Inserting this term into Eq.~\eqref{eq:stellar}, where the sum over $q$ reduces to the single term $e^{-q^2/2}e^{qz_2}=1$, and combining $e^{-k^2/2}$ with $e^{-\beta_k^2}$ into $e^{-ak^2}$, gives Eq.~\eqref{eq:coherent_stellar}. The weight of the height $m_2$ of the momentum cylinder is then $P_2(m_2)=A^2\sum_k|\psi_{k0}|^2e^{-m_2^2}\propto e^{-m_2^2}$. Writing $\mathrm{erf}(t)=1-\mathrm{erfc}(t)$, the bracket of Eq.~\eqref{eq:coherent_stellar} differs from $2$ by terms bounded, once multiplied by $e^{-\beta_k^2}$, by $e^{-\ell^2/8\sigma^2}$.

\subsection{Zero of the modular wavefunction at the corner of the unit cell}
At $(\overline{x},\overline{p})=(\ell/2,\pi/\ell)$, Eq.~\eqref{eq:modular_wf} gives $\psi(\ell/2,\pi/\ell)=\sqrt\ell\sum_n(-1)^n\psi(\ell/2+n\ell)$. The terms $n$ and $-1-n$ involve $\psi(\ell/2+n\ell)$ and $\psi(-\ell/2-n\ell)$, which are equal for an even wavefunction, with the opposite signs $(-1)^n$ and $(-1)^{-1-n}=-(-1)^n$, so that the sum vanishes.

\subsection{Symmetry}
Since $\psi(x)$ is real, Eq.~\eqref{eq:psi_kq} gives $\psi^*_{kq}=\psi_{-k,q}$, hence $F_{\ket{\psi}}(z_1,z_2)^*=F_{\ket{\psi}}(-z_1^*,z_2^*)$. Since $\psi(x)$ is even, the change of variable $\overline{x}\to-\overline{x}$ gives $\psi_{-k,-q}=\psi_{kq}$, hence $F_{\ket{\psi}}(-z_1,-z_2)=F_{\ket{\psi}}(z_1,z_2)$. Combining both relations, $F_{\ket{\psi}}(z_1,z_2)^*=F_{\ket{\psi}}(z_1^*,-z_2^*)$. On the slices $m_2=0$, $z_2=\ii\ell\overline{p}$ and $-z_2^*=z_2$; on the circles $\overline{x}\in\{0,\ell/2\}$, $z_1^*=z_1$ modulo $2\ii\pi$, so that $F_{\ket{\psi}}(z_1,z_2)$ is real on these circles. The same argument, with $F_{\ket{\psi}}(z_1,z_2)^*=F_{\ket{\psi}}(-z_1^*,z_2^*)$, shows that $F_{\ket{\psi}}(z_1,z_2)$ is real on the circles $\overline{p}\in\{0,\pi/\ell\}$ of the slices $m_1=0$.

\section{Numerical location of the zeros}
\label{app:numerics}

The coefficients $\psi_{kq}$ are computed from Eq.~\eqref{eq:psi_kq} by trapezoidal quadrature over each cell with $4096$ points, keeping all the peaks of Eq.~\eqref{eq:physical_gkp} that contribute to the cell, and are normalized numerically. The sums over $k$ and $q$ are truncated at $|k|,|q|\le K$, with $K$ large enough for the neglected terms to be below $e^{-30}$ relative to the largest term. On each slice $(m_1,m_2)$, the stellar function $F_{\ket{\psi}}(z_1,z_2)$ is evaluated on a $200\times200$ grid of the torus. A plaquette of the grid contains a zero when the phase of $F_{\ket{\psi}}$ winds by $\pm2\pi$ along the boundary of the plaquette, the phase differences between neighboring points being taken in $(-\pi,\pi]$, which also gives the index of Eq.~\eqref{eq:index}. The position of each zero is then refined by minimizing the ratio $|F_{\ket{\psi}}(z_1,z_2)|/\sum_{k,q\in\Z}|\psi_{kq}|e^{-(k^2+q^2)/2}e^{km_1+qm_2}$ with the Nelder--Mead algorithm, and the zero is retained when this ratio falls below $10^{-10}$. In practice, the ratio reaches $10^{-16}$ to $10^{-18}$, at the level of the numerical precision, and the positions of the zeros do not change when the quadrature and the grid are refined. In the large-comb approximation, the zeros are given analytically by Eqs.~\eqref{eq:zeros_lc_position} and~\eqref{eq:zeros_lc_momentum}; at these points the same ratio is at the level of the numerical precision, whereas the ratio is of order $10^{-2}$ at neighboring heights.

\section{Coefficients and stellar function of physical logical states}
\label{app:gkp_coefficients}

\subsection{Zak transform}
Inserting Eq.~\eqref{eq:physical_gkp} into Eq.~\eqref{eq:modular_wf} gives Eq.~\eqref{eq:zak_gkp}.

\subsection{Large-comb approximation}
Replacing $e^{-(\overline{x}+n\ell)^2/2\kappa^2}$ by $e^{-n^2\ell^2/2\kappa^2}$ neglects the factor $e^{-(2\overline{x}n\ell+\overline{x}^2)/2\kappa^2}$, whose deviation from one is of order $|n|\ell^2/\kappa^2$, that is of order $\ell/\kappa$ over the populated cells $|n|\lesssim\kappa/\ell$. Setting $s'=n-s$, the double sum of Eq.~\eqref{eq:zak_gkp} then factorizes,
\begin{equation}
\psi_{x_0}(\overline{x},\overline{p})\simeq\sqrt\ell\,N\Big[\sum_{s'\in\Z}e^{-(\overline{x}-x_0+s'\ell)^2/2\sigma^2}\Big]
\times\Big[\sum_{n\in\Z}e^{-n^2\ell^2/2\kappa^2}e^{-\ii n\ell\overline{p}}\Big],
\end{equation}
and the Poisson summation formula gives
\begin{equation}
\sum_{n\in\Z}e^{-n^2\ell^2/2\kappa^2}e^{-\ii n\ell\overline{p}}=\frac{\sqrt{2\pi}\,\kappa}{\ell}\sum_{j\in\Z}e^{-\kappa^2(\overline{p}+2\pi j/\ell)^2/2},
\end{equation}
hence Eq.~\eqref{eq:large_comb} with $N'=\sqrt{2\pi/\ell}\,N\kappa$.

\subsection{Coefficients with error functions}
With the decomposition of Eq.~\eqref{eq:zak_decomposition} and the first relation of Eq.~\eqref{eq:integer_basis}, the coefficients are
\begin{equation}
\psi_{kq}=\frac{1}{2\pi}\int_{-\ell/2}^{\ell/2}\!\!\dd\overline{x}\int_{-\pi/\ell}^{\pi/\ell}\!\!\dd\overline{p}\,\psi(\overline{x},\overline{p})\,e^{-2\ii\pi k\overline{x}/\ell}e^{-\ii\ell q\overline{p}}.
\end{equation} For the product of Gaussians of Eq.~\eqref{eq:physGKP}, the integral over $\overline{x}$,
\begin{equation}
I_k(x_0)=\int_{-\ell/2}^{\ell/2}\dd\overline{x}\;e^{-(\overline{x}-x_0)^2/2\sigma^2}\,e^{-2\ii\pi k\overline{x}/\ell},
\end{equation}
is evaluated in four steps, with $b_k=2\pi k/\ell$.
\begin{enumerate}
\item The substitution $u=\overline{x}-x_0$ centers the Gaussian, maps the integration interval onto $[-\ell/2-x_0,\ell/2-x_0]$, and extracts the phase $e^{-\ii b_kx_0}=e^{-2\ii\pi kx_0/\ell}$. For $x_0\neq0$, the interval is asymmetric about the center of the Gaussian, so that the result does not reduce to a single error function.
\item Completing the square, $-\frac{u^2}{2\sigma^2}-\ii b_ku=-\frac{(u+\ii b_k\sigma^2)^2}{2\sigma^2}-\frac{b_k^2\sigma^2}{2}$, extracts the factor $e^{-b_k^2\sigma^2/2}=e^{-2\pi^2\sigma^2k^2/\ell^2}$, and leaves a Gaussian centered at the complex point $u=-\ii b_k\sigma^2$.
\item The substitution $v=(u+\ii b_k\sigma^2)/(\sqrt2\,\sigma)$, with $\dd u=\sqrt2\,\sigma\,\dd v$, maps the endpoints onto $t_\pm=\frac{\pm\ell/2-x_0}{\sqrt2\,\sigma}+\ii\frac{\sqrt2\pi\sigma k}{\ell}$, since $b_k\sigma/\sqrt2=\sqrt2\pi\sigma k/\ell$. Since $e^{-v^2}$ is entire, the integral depends only on its endpoints, by the Cauchy theorem.
\item With $\mathrm{erf}(z)=\frac{2}{\sqrt\pi}\int_0^z e^{-t^2}\dd t$, analytically continued to $z\in\C$, $\int_{t_-}^{t_+}\dd v\,e^{-v^2}=\frac{\sqrt\pi}{2}[\mathrm{erf}(t_+)-\mathrm{erf}(t_-)]$.
\end{enumerate}
Collecting the factors and using $-\mathrm{erf}(t_-)=\mathrm{erf}(-t_-)$,
\begin{equation}
\begin{split}
I_k(x_0)=\sqrt{\frac{\pi}{2}}\,\sigma\,e^{-\frac{2\pi^2\sigma^2k^2}{\ell^2}}e^{-\frac{2\ii\pi kx_0}{\ell}}&\Big[\mathrm{erf}\Big(\frac{\ell/2-x_0}{\sqrt2\,\sigma}+\ii\frac{\sqrt2\pi\sigma k}{\ell}\Big)\\
&+\mathrm{erf}\Big(\frac{\ell/2+x_0}{\sqrt2\,\sigma}-\ii\frac{\sqrt2\pi\sigma k}{\ell}\Big)\Big],
\end{split}
\end{equation}
which gives $c_k$ in Eq.~\eqref{eq:psi_kq_erf}. The integral over $\overline{p}$ is obtained in the same way, with $x_0\to0$, $\sigma\to1/\kappa$, $\ell/2\to\pi/\ell$ and $2\pi k/\ell\to\ell q$, and gives $d_q$. The constant $\mathcal{N}$ collects the prefactors and is fixed by the normalization.

\subsection{Proof of Eq.~(\ref{eq:psi_kq_lc})}
The coefficients of the large-comb modular wavefunction of Eq.~\eqref{eq:large_comb} follow from the first relation of Eq.~\eqref{eq:integer_basis} and from the decomposition of Eq.~\eqref{eq:zak_decomposition}. Since the modular wavefunction is a product, the double integral over the unit cell factorizes,
\begin{equation}
\label{eq:app_lc_factor}
\psi_{kq}=\frac{N'}{2\pi}\,X_k\,P_q,\qquad
\begin{gathered}
X_k=\int_{-\ell/2}^{\ell/2}\dd\overline{x}\,g_\sigma(\overline{x}-x_0)\,e^{-2\ii\pi k\overline{x}/\ell},\\
P_q=\int_{-\pi/\ell}^{\pi/\ell}\dd\overline{p}\,g_{1/\kappa}(\overline{p})\,e^{-\ii\ell q\overline{p}} .
\end{gathered}
\end{equation}
The integral $X_k$ is evaluated in three steps.
\begin{enumerate}
\item Inserting the sum that defines $g_\sigma$ in Eq.~\eqref{eq:large_comb} gives
\begin{equation}
X_k=\sum_{s\in\Z}\int_{-\ell/2}^{\ell/2}\dd\overline{x}\,e^{-(\overline{x}-x_0+s\ell)^2/2\sigma^2}e^{-2\ii\pi k\overline{x}/\ell}.
\end{equation}
\item In the term $s$, the substitution $x=\overline{x}+s\ell$ maps the unit cell onto the interval $[s\ell-\ell/2,s\ell+\ell/2)$, and leaves the exponential unchanged, since $e^{-2\ii\pi k(x-s\ell)/\ell}=e^{-2\ii\pi kx/\ell}e^{2\ii\pi ks}=e^{-2\ii\pi kx/\ell}$. Each term is thus the integral of the same Gaussian $e^{-(x-x_0)^2/2\sigma^2}e^{-2\ii\pi kx/\ell}$ over a different cell. These cells tile the real line as $s$ runs over $\Z$, so that the sum of the integrals is the integral over $\R$:
\begin{equation}
X_k=\int_\R\dd x\,e^{-(x-x_0)^2/2\sigma^2}\,e^{-2\ii\pi kx/\ell}.
\end{equation}
\item With $u=x-x_0$ and $\omega_k=2\pi k/\ell$, completing the square, $-\frac{u^2}{2\sigma^2}-\ii\omega_ku=-\frac{(u+\ii\omega_k\sigma^2)^2}{2\sigma^2}-\frac{\omega_k^2\sigma^2}{2}$, and the Gaussian integral $\int_\R\dd u\,e^{-(u+\ii\omega_k\sigma^2)^2/2\sigma^2}=\sqrt{2\pi}\,\sigma$, valid for a complex shift of the integration variable since the integrand is entire and decays at infinity, give
\begin{equation}
X_k=\sqrt{2\pi}\,\sigma\,e^{-2\pi^2\sigma^2k^2/\ell^2}\,e^{-2\ii\pi kx_0/\ell}.
\end{equation}
\end{enumerate}
The integral $P_q$ follows in the same way. Inserting the sum that defines $g_{1/\kappa}$ and substituting $p=\overline{p}+2\pi j/\ell$ in the term $j$, the phase $e^{-\ii\ell q(p-2\pi j/\ell)}=e^{-\ii\ell qp}e^{2\ii\pi qj}=e^{-\ii\ell qp}$ is unchanged, and the shifted momentum cells tile the real line, so that
\begin{equation}
P_q=\int_\R\dd p\,e^{-\kappa^2p^2/2}\,e^{-\ii\ell qp}=\frac{\sqrt{2\pi}}{\kappa}\,e^{-\ell^2q^2/2\kappa^2}.
\end{equation}
Inserting $X_k$ and $P_q$ into Eq.~\eqref{eq:app_lc_factor} gives
\begin{equation}
\psi_{kq}=\frac{N'\sigma}{\kappa}\,e^{-2\pi^2\sigma^2k^2/\ell^2}\,e^{-2\ii\pi kx_0/\ell}\,e^{-\ell^2q^2/2\kappa^2},
\end{equation}
which is Eq.~\eqref{eq:psi_kq_lc} with $\mathcal{N}=N'\sigma/\kappa=\sqrt{2\pi/\ell}\,N\sigma$, the value of $\mathcal{N}$ being fixed by the normalization $\sum_{k,q\in\Z}|\psi_{kq}|^2=1$.

\subsection{Relation with the coefficients with error functions}
The coefficients of Eq.~\eqref{eq:psi_kq_erf} correspond to the single term $s=0$ of step~1, that is, to a single Gaussian truncated to the unit cell, whose integral is $I_k(x_0)$. The terms $s\neq0$ supply the parts of the same Gaussian that lie outside the unit cell, so that the full sum gives the integral over $\R$. In terms of error functions, the integral of $e^{-v^2}$ over the whole line equals $\sqrt\pi$, so that the bracket of $I_k(x_0)$, which is equal to $\mathrm{erf}(t_+)-\mathrm{erf}(t_-)$, is replaced by $2$, and the half-sums of error functions in $c_k$ and $d_q$ are replaced by one. Writing $\mathrm{erf}(t)=1-\mathrm{erfc}(t)$, the difference between the two sets of coefficients is given by $\mathrm{erfc}$ terms whose modulus, once multiplied by the Gaussian prefactor, is bounded by $e^{-(\ell/2\mp x_0)^2/2\sigma^2}$ for $c_k$ and by $e^{-\pi^2\kappa^2/2\ell^2}$ for $d_q$. For $x_0=\mp\ell/4$, the largest of these terms is $e^{-\ell^2/32\sigma^2}$. These terms decay only as $1/|k|$ and $1/|q|$, since $e^{-\beta^2}\mathrm{erfc}(\xi+\ii\beta)\simeq e^{-\xi^2-2\ii\xi\beta}/[\sqrt\pi(\xi+\ii\beta)]$ for large $\beta$ and $\xi>0$, whereas the coefficients of Eq.~\eqref{eq:psi_kq_lc} decay as Gaussians. The difference is therefore negligible for small $|k|$, but exceeds the coefficients of Eq.~\eqref{eq:psi_kq_lc} for $|k|\gtrsim\ell^2/(8\pi\sigma^2)$.

\subsection{Normalization}
The condition $\sum_{k,q\in\Z}|\psi_{kq}|^2=1$ gives $\mathcal{N}$ in Eq.~\eqref{eq:psi_kq_lc}. For $\sigma\ll\ell\ll\kappa$, the Poisson summation formula gives $\sum_ke^{-4\pi^2\sigma^2k^2/\ell^2}\simeq\ell/(2\sqrt\pi\,\sigma)$ and $\sum_qe^{-\ell^2q^2/\kappa^2}\simeq\sqrt\pi\,\kappa/\ell$, hence $\mathcal{N}^2\simeq2\sigma/\kappa$.

\subsection{Extent along the heights}
With Eq.~\eqref{eq:psi_kq_lc}, the weight of the height $m_1$ of Eq.~\eqref{eq:height_weights} is
\begin{equation}
P_1(m_1)=A^2\sum_{k,q\in\Z}|\psi_{kq}|^2e^{-(k-m_1)^2}\propto\sum_ke^{-ck^2}e^{-(k-m_1)^2},\qquad c=\frac{4\pi^2\sigma^2}{\ell^2}.
\end{equation}
Approximating the sum by an integral, which is accurate up to terms of relative order $e^{-\pi^2/(1+c)}$ by the Poisson summation formula, and completing the square gives $\int\dd k\,e^{-ck^2-(k-m_1)^2}=\sqrt{\pi/(1+c)}\,e^{-cm_1^2/(1+c)}$, hence Eq.~\eqref{eq:extent}. The weight $P_2(m_2)$ follows in the same way, with $c$ replaced by $d=\ell^2/\kappa^2$.

\subsection{Stellar function and zeros in the large-comb approximation}
Inserting Eq.~\eqref{eq:psi_kq_lc} into Eq.~\eqref{eq:stellar}, the Gaussian factors combine into
\begin{equation}
e^{-k^2/2}e^{-2\pi^2\sigma^2k^2/\ell^2}=e^{-ak^2},\qquad e^{-q^2/2}e^{-\ell^2q^2/2\kappa^2}=e^{-bq^2},
\end{equation}
which gives Eq.~\eqref{eq:stellar_lc}. Lemma~\ref{lemma:zeros} with $\zeta=z_1-2\ii\pi x_0/\ell=m_1+2\ii\pi(\overline{x}-x_0)/\ell$ gives $m_1=(2n+1)a$ and $2\pi(\overline{x}-x_0)/\ell=(2p+1)\pi$, that is $\overline{x}=x_0+\ell/2$ modulo $\ell$, which equals $-x_0$ for $x_0=\mp\ell/4$. With $\zeta=z_2=m_2+\ii\ell\overline{p}$, Lemma~\ref{lemma:zeros} gives $m_2=(2n'+1)b$ and $\ell\overline{p}=\pi$ modulo $2\pi$. Substituting the values of $a$ and $b$ gives Eqs.~\eqref{eq:zeros_lc_position} and~\eqref{eq:zeros_lc_momentum}.

\subsection{Fine-tuned widths}
The conditions $\nu_1a=\mu_1$ and $\nu_2b=\mu_2$, with $a=\frac12+\frac{2\pi^2\sigma^2}{\ell^2}$ and $b=\frac12+\frac{\ell^2}{2\kappa^2}$, give $2\pi^2\sigma^2/\ell^2=(2\mu_1-\nu_1)/(2\nu_1)$ and $\ell^2/2\kappa^2=(2\mu_2-\nu_2)/(2\nu_2)$, that is, Eq.~\eqref{eq:tuned_widths}. With $a=u_1/v_1$ in lowest terms, the height $(2n+1)a$ is an integer if and only if $v_1$ divides $2n+1$. Since $v_1$ divides the odd integer $\nu_1$, $v_1$ is odd, and $2n+1=(2j+1)v_1$, which gives the heights $m_1=(2j+1)u_1$. The same argument applies to $b=u_2/v_2$. For $\sigma=\ell/(2\sqrt3\,\pi)$, $a=\frac12+\frac16=\frac23$, so that $u_1=2$, $v_1=3$, and the heights are $m_1=\pm2,\pm6,\pm10,\dots$; for $\kappa=3\ell$, $b=\frac12+\frac1{18}=\frac59$, so that $u_2=5$, $v_2=9$, and the heights are $m_2=\pm5,\pm15,\dots$.

\section{Overlap between logical states}
\label{app:overlap}

\subsection{Large-comb approximation}
Eq.~\eqref{eq:psi_kq_lc} gives, with $x_0=-\ell/4$ for $\ket{0}_x$ and $x_0=+\ell/4$ for $\ket{1}_x$,
\begin{equation}
\begin{split}
\psi^{(-\ell/4)*}_{kq}\,\psi^{(+\ell/4)}_{kq}&=\mathcal{N}^2e^{-4\pi^2\sigma^2k^2/\ell^2}e^{-\ell^2q^2/\kappa^2}\,e^{-\ii\pi k/2}e^{-\ii\pi k/2}\\
&=\mathcal{N}^2e^{-4\pi^2\sigma^2k^2/\ell^2}e^{-\ell^2q^2/\kappa^2}(-1)^k,
\end{split}
\end{equation}
and the sums over $q$ cancel in the ratio with the normalization, which gives Eq.~\eqref{eq:overlap_lc}. The Poisson summation formula gives
\begin{equation}
\sum_{k\in\Z}(-1)^ke^{-ck^2}=\sqrt{\frac{\pi}{c}}\sum_{j\in\Z}e^{-\pi^2(j+1/2)^2/c}\simeq2\sqrt{\frac{\pi}{c}}\,e^{-\pi^2/4c}
\end{equation}
and $\sum_ke^{-ck^2}\simeq\sqrt{\pi/c}$, and $\pi^2/4c=\ell^2/16\sigma^2$ gives the leading term.

\subsection{Exact overlap}
For two Gaussian peaks centered at $c$ and $c'$,
\begin{equation}
e^{-(x-c)^2/2\sigma^2}e^{-(x-c')^2/2\sigma^2}=e^{-(c-c')^2/4\sigma^2}e^{-(x-\bar c)^2/\sigma^2},
\end{equation}
with $\bar c=(c+c')/2$. The product of the two envelopes is $e^{-x^2/\kappa^2}$. The sum of the exponents is $-x^2/\tilde\sigma^2+2x\bar c/\sigma^2-\bar c^2/\sigma^2$, with $1/\tilde\sigma^2=1/\kappa^2+1/\sigma^2$, that is $\tilde\sigma=\sigma\kappa/\sqrt{\kappa^2+\sigma^2}$, and completing the square gives
\begin{equation}
\int_\R\dd x\,e^{-x^2/\kappa^2}e^{-(x-\bar c)^2/\sigma^2}=\sqrt\pi\,\tilde\sigma\,e^{-\bar c^2/(\kappa^2+\sigma^2)},
\end{equation}
since $\bar c^2\tilde\sigma^2/\sigma^4-\bar c^2/\sigma^2=-\bar c^2/(\kappa^2+\sigma^2)$. With $c=s\ell-\ell/4$ for $\ket{0}_x$ and $c'=s'\ell+\ell/4$ for $\ket{1}_x$, one has $c-c'=(s-s'-\frac12)\ell$ and $\bar c=(s+s')\ell/2$, so that
\begin{equation}
{}_x\!\bra{0}\ket{1}_x=N^2\sqrt\pi\,\tilde\sigma\sum_{s,s'\in\Z}e^{-\frac{(2s-2s'-1)^2\ell^2}{16\sigma^2}}e^{-\frac{(s+s')^2\ell^2}{4(\kappa^2+\sigma^2)}} .
\end{equation}
The normalization of $\ket{0}_x$, with $c=s\ell-\ell/4$ and $c'=s'\ell-\ell/4$, gives
\begin{equation}
N^{-2}=\sqrt\pi\,\tilde\sigma\sum_{s,s'\in\Z}e^{-\frac{(s-s')^2\ell^2}{4\sigma^2}}e^{-\frac{(2s+2s'-1)^2\ell^2}{16(\kappa^2+\sigma^2)}},
\end{equation}
which is also the normalization of $\ket{1}_x$, by the symmetry $s\to-s$. The ratio gives the exact overlap
\begin{equation}
\label{eq:overlap}
{}_x\!\bra{0}\ket{1}_x=\frac{\displaystyle\sum_{s,s'\in\Z}e^{-\frac{(2s-2s'-1)^2\ell^2}{16\sigma^2}}\,e^{-\frac{(s+s')^2\ell^2}{4(\kappa^2+\sigma^2)}}}
{\displaystyle\sum_{s,s'\in\Z}e^{-\frac{(s-s')^2\ell^2}{4\sigma^2}}\,e^{-\frac{(2s+2s'-1)^2\ell^2}{16(\kappa^2+\sigma^2)}}}.
\end{equation}
In the numerator, the dominant terms have $s-s'\in\{0,1\}$ and contribute
\begin{equation}
e^{-\ell^2/16\sigma^2}\sum_{t\in\Z}e^{-t^2\ell^2/4(\kappa^2+\sigma^2)},
\end{equation}
with $t=s+s'$ running over all integers. In the denominator, the dominant terms have $s=s'$ and contribute $\sum_se^{-(s\ell-\ell/4)^2/(\kappa^2+\sigma^2)}$. For $\kappa\gg\ell$, the two sums over $t$ and $s$ are approximated by integrals equal to $2\sqrt{\pi(\kappa^2+\sigma^2)}/\ell$ and $\sqrt{\pi(\kappa^2+\sigma^2)}/\ell$, respectively, which gives ${}_x\!\bra{0}\ket{1}_x\simeq2e^{-\ell^2/16\sigma^2}$.

%%%%%%%%%%%%%%%%%%%%%%%%%%%%%%%%%%%%%%%%%%
% The aps option selects the apsrev4-2 bibliography style automatically.
\bibliography{bibliodouble_submission}

\end{document}